\documentclass[aps,prl,showpacs,amssymb,nofootinbib,superscriptaddress,twocolumn,longbibliography]{revtex4-1}
\pdfoutput=1
\usepackage[english]{babel}

\usepackage{tikz}
\usetikzlibrary{shapes,backgrounds,fit,decorations.pathreplacing,arrows,decorations.markings}
\usepackage{verbatim}
\tikzstyle{vecArrow} = [thick, decoration={markings,mark=at position
   1 with {\arrow[semithick]{open triangle 60}}},
   double distance=1.4pt, shorten >= 5.5pt,
   preaction = {decorate},
   postaction = {draw,line width=1.4pt, white,shorten >= 4.5pt}]
\tikzstyle{innerWhite} = [semithick, white,line width=1.4pt, shorten >= 4.5pt]

\usepackage[ansinew]{inputenc}
\usepackage{bbm}
\usepackage{bm}
\usepackage{amsbsy}
\usepackage{amsthm}
\usepackage{amssymb}
\usepackage{amsfonts}
\usepackage{amsmath}
\usepackage{dsfont} % for symbols like the identity: \mathds{1}
\usepackage{graphicx} % for graphics
\usepackage{epsfig}
\usepackage{epstopdf}
\usepackage{dsfont}
\usepackage{multibib}
\usepackage{color}

\usepackage[colorlinks]{hyperref}
\makeatletter
\newcommand\org@hypertarget{}
\let\org@hypertarget\hypertarget
\renewcommand\hypertarget[2]{%
  \Hy@raisedlink{\org@hypertarget{#1}{}}#2%
  }
\makeatother
\usepackage[figure,table]{hypcap}
\usepackage{MnSymbol}
\usepackage{enumerate}%allows different styles of enumerate environment
\usepackage{float}
\hypersetup{
	bookmarksnumbered,
	pdfstartview={FitH},
	citecolor={darkgreen},
	linkcolor={darkred},
	urlcolor={darkblue},
	pdfpagemode={UseOutlines}}
\definecolor{darkgreen}{RGB}{50,190,50}
\definecolor{darkblue}{RGB}{0,0,190}
\definecolor{darkred}{RGB}{238,0,0}
\definecolor{quantum}{RGB}{83,37,127}
\definecolor{quantumlight}{RGB}{169,146,191}
\usepackage{soul}
\newcommand{\pr}{^{\prime}}

\newcommand{\ket}[1]{\ensuremath{\left|\right.\!{#1}\!\left.\right\rangle}}

\newcommand{\bra}[1]{\ensuremath{\left\langle\right.\!{#1}\!\left.\right|}}

\newcommand{\ketbra}[2]{\ensuremath{|{#1}\rangle\!\langle{#2}|}}

\newcommand{\nl}{\ensuremath{\hspace*{-0.5pt}}}
\newcommand{\nr}{\ensuremath{\hspace*{0.5pt}}}

\newcommand{\subtiny}[3]{\ensuremath{_{\hspace{#1 pt}\protect\raisebox{#2 pt}{\tiny{$ #3$}}}}}
\newcommand{\suptiny}[3]{\ensuremath{^{\hspace{#1 pt}\protect\raisebox{#2 pt}{\tiny{$ #3$}}}}}

\newcommand{\Sys}{\ensuremath{_{\hspace{-0.5pt}\protect\raisebox{0pt}{\tiny{$S$}}}}}
\newcommand{\Poi}{\ensuremath{_{\hspace{-0.5pt}\protect\raisebox{0pt}{\tiny{$P$}}}}}
\newcommand{\PoiN}[1]{\ensuremath{_{\hspace{-0.5pt}\protect\raisebox{0pt}{\tiny{$P_{#1}$}}}}}
\newcommand{\SP}{\ensuremath{_{\hspace{-0.5pt}\protect\raisebox{0pt}{\tiny{$S\hspace*{-1pt}P$}}}}}

\newcommand{\Field}{\ensuremath{_{\hspace{-0.5pt}\protect\raisebox{0pt}{\tiny{$F$}}}}}

\newcommand{\tr}{\textnormal{Tr}}

\newcommand{\djj}{d\kern-0.4em\char"16\kern-0.1em}

\newtheorem{prop}{Proposition}[section]

\renewcommand{\thesection}{\arabic{section}}
\renewcommand{\thesubsection}{\arabic{section}.\Alph{subsection}}
\makeatletter
\renewcommand{\p@subsection}{}
\renewcommand{\p@subsubsection}{}
\makeatother

\newcommand{\mean}[1]{\langle #1 \rangle }
\newcommand{\M}[1]{\mathcal{#1}}

\newcommand{\id}{\mathbb{I}}
\newcommand{\ch}[1]{\left( { #1} \right)}
\newcommand{\N}[1]{\left \Vert { #1} \right \Vert}

\newcommand{\U}[1]{U_{\Lambda} { #1} U^{\dagger}_{\Lambda}}

\usepackage[customcolors]{hf-tikz}
\tikzset{style green/.style={
    set fill color=green!50!lime!60,
    set border color=white,
  },
  style cyan/.style={
    set fill color=cyan!90!blue!60,
    set border color=white,
  },
  style orange/.style={
    set fill color=orange!80!red!60,
    set border color=white,
  },
  style hordash/.style={
    set fill color=white,
    set border color=black,
  },
     style rose/.style={
    set fill color= magenta!70!pink!70, %pink!20!blue!40!,
    set border color=white,
  },
  hor/.style={
    above left offset={-0.09,0.25},
    below right offset={0.09,-0.05},
    #1
  },
  ver/.style={
    above left offset={-0.09,0.35},
    below right offset={0.09,-0.1},
    #1
  }
}
\usepackage[framemethod=tikz]{mdframed}
\definecolor{mycolor}{rgb}{0.122, 0.435, 0.698}
\newmdenv[innerlinewidth=0.5pt, roundcorner=4pt,linecolor=mycolor,innerleftmargin=6pt,
innerrightmargin=6pt,innertopmargin=6pt,innerbottommargin=6pt]{mybox}
\usepackage{paracol}
\usepackage{tcolorbox}
\tcbuselibrary{theorems}

\newtcbtheorem{Definitions}{Definition}%
{colback=mycolor!5,colframe=mycolor,fonttitle=\bfseries,
left=4pt,
right=4pt,
top=5pt,
bottom=5pt}{defi}

\newtcbtheorem{Lemmas}{Lemma}%
{colback=green!5,colframe=green!50!blue,fonttitle=\bfseries,
left=4pt,
right=4pt,
top=5pt,
bottom=5pt
}{lemma}

\newtcolorbox[blend into=figures]{boxdefi}[3][]
{ float*=ht,width=\textwidth,lower separated=false, center upper,
title={#2},label= def:#3,#1}

\usepackage{filecontents}

\renewcommand{\thesection}{\Roman{section}}
\renewcommand{\thesubsection}{\Roman{section}.\arabic{subsection}}

\begin{document}

\title{Work estimation and work fluctuations in the presence of non-ideal measurements}
\author{Tiago Debarba}
\email{debarba@utfpr.edu.br}
\affiliation{Departamento Acad{\^ e}mico de Ci{\^ e}ncias da Natureza, Universidade Tecnol{\'o}gica Federal do Paran{\'a} (UTFPR), Campus Corn{\'e}lio Proc{\'o}pio, Avenida Alberto Carazzai 1640, Corn{\'e}lio Proc{\'o}pio, Paran{\'a} 86300-000, Brazil.}
\affiliation{Institute for Quantum Optics and Quantum Information - IQOQI Vienna, Austrian Academy of Sciences, Boltzmanngasse 3, 1090 Vienna, Austria}
\author{Gonzalo Manzano}
\email{gmanzano@ucm.es}
\affiliation{Scuola Normale Superiore, Piazza dei Cavalieri 7, I-56126, Pisa, Italy}
\affiliation{International Center for Theoretical Physics, Strada Costiera 11, Trieste 34151, Italy}
\author{Yelena Guryanova}
\email{yelena.guryanova@oeaw.ac.at}
\affiliation{Institute for Quantum Optics and Quantum Information - IQOQI Vienna, Austrian Academy of Sciences, Boltzmanngasse 3, 1090 Vienna, Austria}
\author{Marcus Huber}
\email{marcus.huber@univie.ac.at}
\affiliation{Institute for Quantum Optics and Quantum Information - IQOQI Vienna, Austrian Academy of Sciences, Boltzmanngasse 3, 1090 Vienna, Austria}
\author{Nicolai Friis}
\email{nicolai.friis@univie.ac.at}
\affiliation{Institute for Quantum Optics and Quantum Information - IQOQI Vienna, Austrian Academy of Sciences, Boltzmanngasse 3, 1090 Vienna, Austria}
%%%
%%

%%
%%

\begin{abstract}
From the perspective of quantum thermodynamics, realisable measurements cost work and result in measurement devices that are not perfectly correlated with the measured systems. We investigate the consequences for the estimation of work in non-equilibrium processes and for the fundamental structure of the work fluctuations when one assumes that the measurements are non-ideal. We show that obtaining work estimates and their statistical moments at finite work cost implies an imperfection of the estimates themselves: more accurate estimates incur higher costs. Our results provide a qualitative relation between the cost of obtaining information about work and the trustworthiness of this information. Moreover, we show that Jarzynski's equality can be maintained exactly at the expense of a correction that depends only on the system's energy scale, while the more general fluctuation relation due to Crooks no longer holds when the cost of the work estimation procedure is finite. We show that precise links between dissipation and irreversibility can be extended to the non-ideal situation.
\end{abstract}

\maketitle

%%%%%%%%%%%%%%%%%%%%%%%%%%%%%%%%%%%%%%%%%%%%%%%%%%%%%%%%%%%%%%%%%%%%%%%%%%%%%%%%%%%%%%%%%%%%%%
\section{Introduction}

Energy is a resource and, as with any resource, it is of interest to understand how much of it is spent or can be obtained during a given process, or simply, how much of it is stored, for instance, in a battery. A quite different, but familiar, resource that one handles on a daily basis is money. Money does not usually come for free: it is exchanged for goods and services, and as a consequence it is in one's interest to know how much things cost and how much money is at hand, e.g., stored in a wallet or bank account. But while checking the exact amount of money (or lack thereof) in one's wallet is free, it is not unusual to expect that banks charge certain fees for storing and transferring money. Unfortunately, when it comes to energy, Nature is similarly unforthcoming.
Fees apply to the storage and transfer of energy and an energy cost is incurred for obtaining estimates of the work transferred within any thermodynamic process, or stored in a quantum system. In this work, we show that obtaining these estimates with a finite amount of work implies an imperfection in the estimates themselves: estimates which are more accurate incur higher costs.

From a thermodynamic point of view, acknowledging the energetic cost of measurements is crucial, e.g., for a complete understanding of Maxwell's demon or Szilard's engine~\cite{LeffRex2003, MayuramaNoriVedral2009}. The work-cost of those measurements that are ideal and projective has been investigated by means of the work-value of measurement outcomes~\cite{SagawaUeda2009, Jacobs2012, LipkaBartosikDemkowiczDobrzanski2018} or via Landauer's erasure bound for resetting the memory that stores these outcomes~\cite{Landauer1961, Bennett1982, EspositoVanDenBroeck2011, ReebWolf2014, AbdelkhalekNakataReeb2016}. However, a common observation among Refs.~\cite{ParrondoHorowitzSagawa2015, KammerlanderAnders2016, ManzanoPlastinaZambrini2018} is that the benefits derived from using measurements as sources of free energy are either matched or surpassed by the corresponding costs.

The crux of our argument is that energy delivered by measurements is not free of charge and must be supplied to realise the measurement in the first place. As we show, this statement is bolstered by the first, second and, in particular, third law of thermodynamics. It was recently shown in~\cite{GuryanovaFriisHuber2018} that ideal projective measurements require one to prepare the measurement apparatus in a pure initial state. The third law stipulates that such zero-entropy states can only be prepared asymptotically using infinite time, infinite energy, or operations of infinite complexity (see e.g.,~\cite{MasanesOppenheim2017}). Consequently, ideal measurements do not exist, in a strict sense, since they always incur diverging costs. This implies that any realistic measurement using finite resources is non-ideal. It is precisely these considerations that become conceptually important when the purpose  of the measurement is to assess the energy consumption itself.

Significant focus in quantum statistical mechanics has been dedicated to the quantification of work and its fluctuations in thermodynamic processes~\cite{DornerClarkHeaneyFazioGooldVedral2013, MazzolaDeChiaraPaternostro2013, FuscoEtAl2014, RoncagliaCerisolaPaz2014}. Studies have also looked at the two-point measurement (TPM) scheme (one of the most prominent approaches for estimating work in an out-of-equilibrium process)~\cite{TalknerLutzHaenggi2007} in the context of Jarzynski's and Crooks' fluctuation relations~\cite{CampisiHaenggiTalkner2011}. In this work we revisit these concepts and investigate the consequences for these quantities when one does \textit{not} assume ideal measurements. We explicitly show how the average work of the ideal TPM is modified and discuss the operational meaning of the corresponding estimates. We show that while Jarzynski's equality can be maintained exactly at the expense of a correction that only depends on the system's Hamiltonian, the more general relation due to Crooks (as well as related results linking irreversibility and dissipation~\cite{KawaiParrondoVanDenBroeck2007, ParrondoVanDenBroeckKawai2009}) no longer hold in the presence of non-ideal measurements. Our results provide a qualitative connection between the cost of obtaining information about work and the trustworthiness of this information.

The paper is structured as follows. In Sec.~\ref{sec:framework} we set the stage for our investigation: First, we review the usual TPM scheme in Sec.~\ref{sec:TPM} before discussing the key properties of ideal and non-ideal measurements following Ref.~\cite{GuryanovaFriisHuber2018} in Secs.~\ref{sec:ideal measurements} and~\ref{sec:nonideal measurements}, respectively. In Sec.~\ref{sec:stimating work with non-ideal measurements} we discuss a modified TPM scheme based on non-ideal measurements and the resulting work estimates, before investigating the implications for fluctuation relations of Jarzynski and Crooks in Sec.~\ref{sec:fluctuation relations}. Finally, we discuss our findings in Sec.~\ref{sec:discussion}.

%%%%%%%%%%%%%%%%%%%%%%%%%%%%%%%%%%%%%%%%%%%%%%%%%%%%%%%%%%%%%%%%%%%%%%%%%%%%%%%%%%%%%%%%%%%%%%%%%%%%%%%%%

\section{Framework}\label{sec:framework}

\subsection{Two-point measurement scheme}\label{sec:TPM}

To formulate our ideas we adopt a commonplace view in quantum thermodynamics, namely, that work is a central resource that is required to move systems away from freely available thermal equilibrium states~\cite{GourMuellerNarasimhacharSpekkensHalpern2015} --- an approach that has staged a diverse range of investigations within the broader field~\cite{VinjanampathyAnders2016, MillenXuereb2016, GooldHuberRieraDelRioSkrzypczyk2016}.
In this paradigm, previous research has investigated the work-cost (or gain) of quantum processes~\cite{HorodeckiOppenheim2013b, SkrzypczykShortPopescu2014, FaistDupuisOppenheimRenner2015, WilmingGallegoEisert2016, FaistRenner2018}, refrigeration~\cite{ClivazSilvaHaackBohrBraskBrunnerHuber2019a, ClivazSilvaHaackBohrBraskBrunnerHuber2019b}, or for establishing correlations~\cite{HuberPerarnauHovhannisyanSkrzypczykKloecklBrunnerAcin2015, BruschiPerarnauLlobetFriisHovhannisyanHuber2015, FriisHuberPerarnauLlobet2016, VitaglianoKloecklHuberFriis2019}.

In this so called `resource-theoretic' approach, consider a quantum system with Hamiltonian $H\suptiny{0}{0}{(0)}=\sum_{i}E_{i}\suptiny{0}{0}{(0)}\ketbra{E_{i}\suptiny{0}{0}{(0)}}{E_{i}\suptiny{0}{0}{(0)}}$ initially at thermal equilibrium with its environment at temperature $T$, described by a Gibbs state $\tau\suptiny{0}{0}{(0)} = \exp(-\beta H\suptiny{0}{0}{(0)})/\mathcal{Z}\suptiny{0}{0}{(0)}$ with partition function $\mathcal{Z}\suptiny{0}{0}{(0)}=\tr\bigl(\exp(-\beta H\suptiny{0}{0}{(0)})\bigr)$ and
$\beta=(k\subtiny{0}{0}{\mathrm{B}}T)^{-1}$. Suppose the system is driven out of equilibrium by a process $\Lambda$
\begin{align}
    \left(\tau\suptiny{0}{0}{(0)}, H\suptiny{0}{0}{(0)}\right) \xrightarrow{\phantom{b..}\Lambda \phantom{b..}}\left( \rho\suptiny{0}{0}{(\mathrm{f})} , H\suptiny{0}{0}{(\mathrm{f})}\right),
    \label{eq:process}
\end{align}
resulting in a final Hamiltonian $H\suptiny{0}{0}{(\mathrm{f})}=\sum_{i}E_{i}\suptiny{0}{0}{(\mathrm{f})}\ketbra{E_{i}\suptiny{0}{0}{(\mathrm{f})}}{E_{i}\suptiny{0}{0}{(\mathrm{f})}}$ and a final state $\rho\suptiny{0}{0}{(\mathrm{f})}=\U{\tau\suptiny{0}{0}{(0)}}$, where $U_{\Lambda}$ is a unitary determined by $\Lambda$. The work that is performed on or extracted from the system during such a process can be estimated via the two-point measurement (TPM) scheme~\cite{TalknerLutzHaenggi2007, CampisiHaenggiTalkner2011} consisting of two \emph{ideal projective measurements} with respect to the eigenbases of $H\suptiny{0}{0}{(0)}$ and $H\suptiny{0}{0}{(\mathrm{f})}$ before and after the protocol $\Lambda$ is implemented, respectively. After obtaining the outcomes labelled by ``$n$" and ``$m$" in these measurements one concludes that the system is left in the states $\ket{E_{n}\suptiny{0}{0}{(0)}}$ and $\ket{E_{m}\suptiny{0}{0}{(\mathrm{f})}}$, respectively. To any transition between these pure states one may associate a probability $p_{\raisebox{-1pt}{\scriptsize{$n\!\rightarrow\!m$}}} =|\bra{E_{m}\suptiny{0}{0}{(\mathrm{f})}}U_{\Lambda}\ket{E_{n}\suptiny{0}{0}{(0)}}|^{2}$ together with a work value $W_{\raisebox{-1pt}{\scriptsize{$n\!\rightarrow\!m$}}}=E_{m}\suptiny{0}{0}{(\mathrm{f})}-E_{n}\suptiny{0}{0}{(0)}$, while the probability for obtaining the first outcome is $p_{n}\suptiny{0}{0}{(0)}=\exp(-\beta E_{n}\suptiny{0}{0}{(0)})/\mathcal{Z}\suptiny{0}{0}{(0)}$. The average work performed during the protocol is thus
\begin{align}
    \mean{W}_{\Lambda} &=\,\sum\limits_{m,n}p_{n}\suptiny{0}{0}{(0)}\,p_{\raisebox{-1pt}{\scriptsize{$n\!\rightarrow\!m$}}}\,
    \bigl(E_{m}\suptiny{0}{0}{(\mathrm{f})}-E_{n}\suptiny{0}{0}{(0)}\bigr),
    \label{eq:ideal tpm work}
\end{align}
which equals the change in average energy during the protocol $\Lambda$, i.e.,
\begin{align}
    \mean{W}_{\Lambda} &=\,\tr(H\suptiny{0}{0}{(\mathrm{f})}\rho\suptiny{0}{0}{(\mathrm{f})})-\tr(H\suptiny{0}{0}{(0)}\tau\suptiny{0}{0}{(0)})
    =:\Delta E_{\Lambda}.
\end{align}
Estimates of $\mean{W}_{\Lambda}$ could thus be obtained from performing ideal measurements and collecting the corresponding outcome statistics. In the following, we will show how the quantity $\mean{W}_{\Lambda}$ in Eq.~(\ref{eq:ideal tpm work}) and its estimate are modified when replacing the two ideal measurements in the TPM by more general non-ideal measurements (see Fig.~\ref{fig:tpmfig}).

%%%%%%%%%%%%%%%%%%%%%%%%%%%%%%%%%%%%%%%%%%%%%%%%%%%%%%%%%%%%%%%%%%%%%%%%%%%%%%%%%%%%%%%%%%%%%%%%%%%%%%%%%

\subsection{Ideal measurements}\label{sec:ideal measurements}

The notion of perfect projective measurements that leave the system in pure states with certainty is of course idealized. To  understand this idealization and its consequences, we review the framework for non-ideal measurements in Ref.~\cite{GuryanovaFriisHuber2018}: Measurements are performed by coupling the measured system to a suitably prepared measurement apparatus (the ``pointer") via an energy investment. Assuming that the initial system and pointer states are $\rho\Sys$ and $\rho\Poi$, respectively, the measurement can be described by a physical process that correlates the system and pointer, resulting in a joint post-measurement state $\tilde{\rho}\SP$. For each of the states $\ket{n}\Sys$ (with $n=0,1,\ldots,d\Sys-1$) in the measurement basis $\{\ket{n}\Sys\}_{n}$, one assigns a corresponding outcome subspace of the pointer Hilbert space via a projector $\Pi_{n}$, such that $\sum_{n}\Pi_{n}=\mathds{1}\Poi$ and $\Pi_{m}\Pi_{n}=\delta_{mn}\Pi_{n}$. We then define an ideal measurement to have the following three properties:
\begin{enumerate}[(i)]
    \item{\label{item:unbiased main}
    \textbf{Unbiased}: The (post-interaction) pointer reproduces the measurement statistics of the (pre-interaction) system exactly, i.e.,
    \begin{align}
         \tr\bigl(\mathds{1}\Sys\otimes\Pi_{n}\tilde{\rho}\SP\bigr)    &=\,\tr\bigl(\ket{n}\!\!\bra{n}\Sys\,\rho\Sys\bigr)\ \forall\,n\ \ \forall \,
         \rho\Sys\,.
         \label{eq:unbiased main}
    \end{align}
    }
    \item{\label{item:faithful main}
    \textbf{Faithful}: The post-interaction pointer and the post-interaction system are perfectly correlated w.r.t. the measurement basis (projectors), that is,
    \begin{align}
        C(\tilde{\rho}\SP) &:=\sum\limits_{n}\tr\bigl(\ket{n}\!\!\bra{n}\Sys\otimes\Pi_{n}\,\tilde{\rho}\SP\bigr)\,=\,1\,.
        \label{eq:correlation function faithful}
    \end{align}
    In other words, given a measurement outcome $n$, the probability that the system is left in the state $\ket{n}\Sys$ is $1$.
    }
    \item{\label{item:non-invasive main}
    \textbf{Non-invasive}: The diagonal entries (w.r.t. the measurement basis) of the pre-measurement system state and the unconditional post-measurement system state are the same, i.e.,
    \begin{align}\label{eq:noninv}
        \tr\bigl(\ket{n}\!\!\bra{n}\Sys\,\rho\Sys\bigr) &=\,
        \tr\bigl(\ket{n}\!\!\bra{n}\Sys
    \,\tilde{\rho}\Sys\bigr)
        \quad \forall\,n\ \forall \,
         \rho\Sys\, ,
    \end{align}
    where $\tilde{\rho}\Sys := \tr_{\Poi}(\tilde{\rho}\SP)$.
    }
\end{enumerate}
A measurement that \textit{does not} satisfy all three properties is called \emph{non-ideal}. In particular, it was shown in Ref.~\cite{GuryanovaFriisHuber2018} that ideal projective measurements are not exactly realizable in practise because they require the preparation of initially pure pointer states (at least in some nontrivial subspace) to satisfy condition (\ref{item:faithful main}). However, the third law of thermodynamics prevents one from reaching the ground-state of any system with finite resources~\cite{SchulmanMorWeinstein2005, SilvaManzanoSkrzypczyBrunner2016, WilmingGallego2017, MasanesOppenheim2017, ClivazSilvaHaackBohrBraskBrunnerHuber2019a, ScharlauMueller2018, ClivazSilvaHaackBohrBraskBrunnerHuber2019b}. Since any other (pure) state necessarily has higher energy than the ground state, the third law thus excludes ideal measurements.

\subsection{Non-ideal measurements}\label{sec:nonideal measurements}

In order to understand non-ideal measurements better, we consider how the laws of thermodynamics place constraints on the ability to perform measurements in quantum mechanics. To begin, we note that the first law establishes a lower bound on the work-cost of measurement when one explicitly considers the system and measurement apparatus in the physical description (note that work-costs in terms of lower bounds are sometimes assumed, even if the costs are not always explicitly considered~\cite{ElouardHerreraMartiHuardAuffeves2017, ElouardJordan2018, BuffoniEtAl2019}). The second law implies that any reduction of the system's entropy (e.g., by a projective measurement that leaves the system in a pure state) must be compensated by an entropy increase of at least the same magnitude in the environment or the measurement apparatus~\cite{ParrondoHorowitzSagawa2015}. Finally, the third law provides the most severe constraint on the cost of measurement. From the above~\cite{GuryanovaFriisHuber2018} an ideal measurement can only be implemented using a pure state measurement apparatus. By the third law, it is impossible to create pure states using finite resources (e.g., with finite time, energy or complexity), which implies that any physical measurement using finite resources is non-ideal. Nevertheless, it was shown that non-ideal measurements employing finite resources can approximate ideal projective measurements arbitrarily well.

To understand the sense in which a non-ideal measurement can be considered `close' to ideal, we recall properties (\ref{item:unbiased main})-(\ref{item:non-invasive main}). All three properties are independent -- for a given measurement any one of them can be satisfied, while the other two are not~\cite[Appendix~A.3]{GuryanovaFriisHuber2018}. At the same time, satisfying any two properties (for all $\rho\Sys$) also implies the third. Since the third law of thermodynamics prevents measurements from being exactly faithful (i.e., satisfying~\eqref{item:faithful main} exactly), this means that any non-ideal measurement can only be unbiased or non-invasive, but not both.

Consider a non-ideal measurement that is non-invasive. From the above we know it can be neither faithful nor unbiased. Then neither the individual measurement outcomes nor the statistics generated from many measurements allow reliable inferences about either the post- or pre-measurement system state, respectively. Such a measurement does not seem to reveal any information about the measured system. Instead we consider non-ideal measurements that are \emph{unbiased}, and take this to be the relevant requirement to speak meaningfully about a measurement. For unbiased measurements one may then attempt to maximise the correlation between pointer outcomes and post-measurement system states to approach an ideal measurement.

To formalise this, we consider an arbitrary system state $\rho\Sys$ in finite dimension $d\Sys$, measured in the energy eigenbasis $\{\ket{E_{i}}\Sys\}_{i=0}^{d\Sys-1}$. The system interacts with the measurement apparatus, represented as a finite-size pointer with Hamiltonian $H\Poi = \sum_{i} E_i\suptiny{0}{0}{({P})} \ketbra{E_i\suptiny{0}{0}{({P})}}{E_i\suptiny{0}{0}{({P})}}$ and dimension $d\Poi$. In order to account transparently for all resources from a thermodynamic point of view, we assume the pointer is initially in a thermal\footnote{If the system is assumed to be initially thermal as well, such as in the TPM scheme, there need not be a relation between the inverse temperatures of the pointer and the system in principle.} state $\tau\Poi = \exp(-\beta\Poi H\Poi)/\mathcal{Z}\Poi$. The interaction can be modelled by a suitable unitary\footnote{In principle, one may effectively model such processes by completely positive and trace preserving (CPTP) maps, see~\cite[Appendix A.4]{GuryanovaFriisHuber2018}, but for an exact account of the invested work it is necessary to specify a corresponding dilation to a unitary acting on a larger space of the pointer and its environment.} evolution $U_{\mathrm{meas}}$ leading to
\begin{align}
    \tilde{\rho}\SP &:= U_{\mathrm{meas}}(\rho\Sys\otimes\tau\Poi) U_{\mathrm{meas}}^{\dagger}.
\end{align}
The interaction between system and pointer is unitary but does not generally preserve energy, requiring an investment of energy
\begin{align}\label{eq:specificWork}
    \Delta E_{\mathrm{meas}} &= \tr\bigl[(H\Sys + H\Poi)(\tilde{\rho}\SP - \rho\Sys \otimes \tau\Poi)\bigr]
\end{align}
in the form of work. For details see Appendix~\ref{appendix:Ideal and Non-Ideal Measurements}. Following Ref.~\cite[Lemma~2]{GuryanovaFriisHuber2018}, one may then construct $U_{\mathrm{meas}}$ to realise an unbiased measurement with finite energy cost, as we illustrate for a $3$-dimensional system in Appendix~\ref{appsec:unbiased measurements}. The probability $p_{n}$ to obtain the measurement outcome $n$ in such a measurement is given by $p_{n}=\tr\bigl(\ket{n}\!\!\bra{n}\rho\Sys\bigr)$. To ensure this is the case for all $\rho\Sys$, the unitary $U_{\mathrm{meas}}$ must result in a final state $\tilde{\rho}\SP$ which satisfies the equivalence relation
\begin{align}
    \mathds{1}\Sys\otimes\Pi_{n} \tilde{\rho}\SP \mathds{1}\Sys\otimes\Pi_{n}
    \mathrel{\hat=}
    p_{n}\tilde{U}\suptiny{0}{0}{(n)}\tau\Poi \tilde{U}\suptiny{0}{0}{(n)}{}^{\dagger},
\end{align}
i.e., $\mathds{1}\Sys\otimes\Pi_{n} \tilde{\rho}\SP \mathds{1}\Sys\otimes\Pi_{n}$ are $d\Poi\times d\Poi$ matrices with the same spectra as $p_{n}\tau\Poi$, i.e., they are equivalent up to applications of arbitrary unitaries $\tilde{U}\suptiny{0}{0}{(n)}$ on the pointer Hilbert space, as we discuss in detail in Appendix~\ref{appsec:unbiased measurements}. This implies that $\rho_{n}$, the conditional post-measurement system state for outcome $n$, is independent of the initial system state $\rho\Sys$. Note that this is also the case for ideal measurements (which are unbiased by definition), where the pure state after the measurement only depends on the measurement outcome, but not on the initial system state. In this regard, the difference with unbiased non-ideal measurements is that the conditional state for the latter are not pure. The most general form of any post-measurement mixed state conditioned on the outcome $n$ is
\begin{align}
    \rho_{n}    &= \sum\limits_{l,l\pr}q_{ll\pr|n}\,\ket{E_{l}}\!\!\bra{E_{l\pr}}.
\end{align}
Each state is normalised $\sum_{l}q_{ll|n}=1 ,\;\forall \,n$ and in general the coefficients $q_{ll\pr|n}$ are independent of $\rho\Sys$. In the rest of the paper we formulate our results as conditions and constraints on $q_{ll\pr|n}$.

One should also bear in mind that for non-ideal measurements the pointer outcomes and post-measurement system states are not perfectly correlated. This can be quantified by
the correlation function in Eq.~(\ref{eq:correlation function faithful}),
\begin{align}
    C(\tilde{\rho}\SP)  &=\,\sum\limits_{i=0}^{d\Sys-1}\tr\bigl(\ket{E_{i}}\!\!\bra{E_{i}}\Sys\otimes\Pi_{i}\,\tilde{\rho}\SP\bigr),
    \label{eq:correlation function}
\end{align}
where $\Pi_{i}$ are the orthogonal pointer projectors associated to different outcomes. The value of $C(\tilde{\rho}\SP)$ represents the average probability of correctly inferring the post-measurement state upon observing the pointer. For ideal measurements $C(\tilde{\rho}\SP)=1$. However, for non-ideal measurements there is an algebraic maximum, the \emph{maximal correlation} $C_{\mathrm{max}}<1$, which can be unitarily achieved.
$C_{\mathrm{max}}$ is given by the sum of the largest $\tfrac{d\Poi}{d\Sys}$ eigenvalues of~$\tau\Poi$, see~\cite[Appendix A.7]{GuryanovaFriisHuber2018}. We call unbiased measurements that achieve $C(\tilde{\rho}\SP)=C_{\mathrm{max}}$ \emph{unbiased maximally correlating} (UMC) measurements, discussed in detail in Appendix~\ref{appsec:unbiased max corr measurements}. In the same Appendix we also show that UMC measurements lead to the following constraints on the coefficients of the post-measurement state: $q_{nn|n}=C_{\mathrm{max}}\, ,\forall\,n$ and $q_{nl|n}=q_{ln|n}=0\ ,\forall l\neq n$. As a consequence, the trace distance between the conditional post-measurement system state $\rho_{n}$ and the pure state $\ket{E_{n}}$ evaluates to $D(\rho_{n}, \ket{E_{n}}\!\!\bra{E_{n}}) = 1-C_{\mathrm{max}}$.

By further restricting to UMC measurements of \emph{minimal energy} (Appendix~\ref{appsec:min energy unbiased max corr measurements}) we obtain $q_{ll\pr|n}=0\ \forall E_{l}\neq E_{l\pr}$, because any off-diagonal elements with respect to energy eigenstates with different energies would imply extractable work, see, e.g.,~\cite{PuszWoronowicz1978}. Minimal energy UMC measurements thus imply a back-action on the measured system. Up to off-diagonal elements in degenerate subspaces, the unconditional post-measurement state for minimal energy UMC measurements is given by $\tilde{\rho}\Sys = \sum_{n} p_{n} \rho_n= \sum_{n,l} q_{ll|n} p_{n} \ket{E_{l}}\!\!\bra{E_{l}}$. By~\eqref{eq:noninv} this would be non-invasive when $\rho_{\Sys} = \tilde{\rho}_{\Sys}$ implying $q_{ll|n}=\delta_{ln} \;\forall n$.

%%%%%%%%%%%%%%%%%%%%%%%%%%%%%%%%%%%%%%%%%%%%%%%%%%%%%%%%%%%%%%%%%%%%%%%%%%%%%%%%%%%%%%%%%%%%%%%%%%%%%%%%%
%%
\begin{figure}[t!]
\begin{center}
\includegraphics[width=0.95\columnwidth,trim={0cm 0mm 0cm 0mm}]{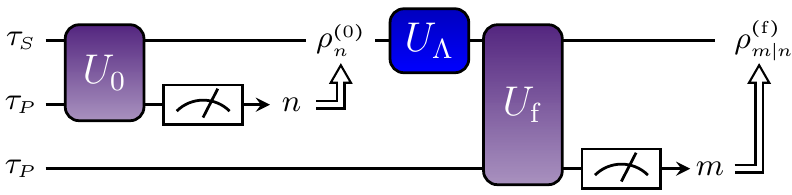}
\caption{Work estimation using two non-ideal measurements. To estimate the work done on or extracted from a system during a process $\Lambda$, two measurements are carried out before and after the process occurs. The respective outcomes labelled ``$n$" and ``$m$" allow concluding that the system is left in states $\rho\suptiny{0}{-1}{(0)}\subtiny{0}{-1}{n}$ and $\rho\suptiny{0}{-1}{(\mathrm{f})}\subtiny{0}{-1}{m|n}$. In the TPM scheme~\cite{TalknerLutzHaenggi2007}, these states are (pure) eigenstates of the system Hamiltonian. For non-ideal measurements modelled by unitaries $U_{0}$ and $U_{\mathrm{f}}$ coupling the system state to pointers
originally in thermal states $\tau\Sys$, $\rho\suptiny{0}{-1}{(0)}\subtiny{0}{-1}{n}$ and $\rho\suptiny{0}{-1}{(\mathrm{f})}\subtiny{0}{-1}{m|n}$ are mixed states.
}
\label{fig:tpmfig}
\end{center}
\end{figure}
%%
%%%%%%%%%%%%%%%%%%%%%%%%%%%%%%%%%%%%%%%%%%%%%%%%%%%%%%%%%%%%%%%%%%%%%%%%%%%%%%%%%%%%%%%%%%%%%%%%%%%%%%%%%

At the same time, $C_{\mathrm{max}}$ can be understood as an indicator of the resource cost of a measurement: Qualitatively, increasing $C_{\mathrm{max}}$ requires more work, control over more complex pointers, more time to carry out operations, or combinations thereof~\cite{GuryanovaFriisHuber2018}. In particular, for fixed pointer systems, increasing $C_{\mathrm{max}}$ can only be achieved by investing work in the preparation of the initial state of the pointer.

%%%%%%%%%%%%%%%%%%%%%%%%%

\section{Estimating work with non-ideal measurements}\label{sec:stimating work with non-ideal measurements}

It is hardly surprising that the procedure of checking how much work is spent or extracted during a protocol itself costs work. In other words, obtaining a work estimate is accompanied by an additional, non-negative work-cost. In Eq.~\eqref{eq:specificWork} we saw that the specific (average) work $\Delta E_{\mathrm{meas}}$ of UMC measurements depends on the details of the Hamiltonians, the initial state $\rho\Sys$, the temperature, and the association of the states $\ket{E_{i}}\Sys$ with the projectors $\Pi_{i}$. It has a finite positive minimum value $\Delta E_{\mathrm{meas}}^{\mathrm{min}}>0$ and as pointed out in Ref.~\cite{GuryanovaFriisHuber2018} should not be taken for granted, as it may significantly outweigh the ideal expectation~$\mean{W}_{\Lambda}$ in Eq.~\eqref{eq:ideal tpm work}. Particular attention should be paid when it comes to machines using measurements as a means of injecting free energy into the system~\cite{ElouardHerreraMartiHuardAuffeves2017, ElouardJordan2018, BuffoniEtAl2019}, as such costs need to be included in an evaluation of the machine's efficiency.

In this work, we do not wish to focus on the specific cost of the measurement, but rather on the consequences for the work estimate itself. The imperfection of measurements is unavoidable and has immediate and interesting consequences for work estimation. To investigate, we modify the ordinary TPM scheme  and replace the ideal measurements by non-ideal (minimal energy UMC) measurements\footnote{In principle these two measurements can be different but we assume that they are both minimal energy UMC.} (Fig.~\ref{fig:tpmfig}). This assumption can be interpreted as the desire to restrict to measurements that are as close as possible to ideal ones while choosing the energetically cheapest way of doing so. A guiding intuition for the following analysis is that work estimation requires two measurements. If the first is non-ideal but unbiased, it necessarily disturbs the system. In particular, the invasiveness of the first measurement changes the statistics of the second, leading to deviations in the work estimation from the ideal case.

Because the measurements are assumed to be unbiased, the probability $p_{n}\suptiny{0}{0}{(0)}$ for obtaining the outcome $n$ in the first measurement is unchanged w.r.t. to the ideal scenario. However, since the system is disturbed by the measurement-induced back action, the conditional post-measurement state $\rho\suptiny{0}{0}{(0)}_{n}$ is no longer an eigenstate $\ket{E\suptiny{0}{0}{(0)}_{n}}$ of $H\suptiny{0}{0}{(0)}$. In particular, the conditional probability that the system is left in the eigenstate $\ket{E\suptiny{0}{0}{(0)}_{l}}$ after observing the pointer outcome $n$ in the first measurement is $q\suptiny{0}{-1}{(0)}_{ll|n}\neq\delta_{ln}$. Given the outcome $n$, the process $\Lambda$, thus acts on the state $\rho_{n}\suptiny{0}{0}{(0)} = \sum_{l} q\suptiny{0}{0}{(0)}_{ll|n}  \ketbra{E_l\suptiny{0}{0}{(0)}}{E_l\suptiny{0}{0}{(0)}}$ via a unitary $U_{\Lambda}$. Unbiasedness then implies that the conditional probability to obtain outcome $m$ in the second measurement given outcome $n$ in the first measurement is
\begin{align}
    p(m|n)  &=\,
    \bra{E_m\suptiny{0}{0}{(\mathrm{f})}}
        U_{\Lambda} \rho_{n}\suptiny{0}{0}{(0)} U^{\dagger}_{\Lambda}
    \ket{E_m\suptiny{0}{0}{(\mathrm{f})}}
    \nonumber\\[1mm]
    &=\,\sum\limits_{l}  q\suptiny{0}{-1}{(0)}_{ll|n}\,
    p_{\raisebox{-1pt}{\scriptsize{$l\!\rightarrow\!m$}}}\, \neq\, p_{\raisebox{-1pt}{\scriptsize{$n\!\rightarrow\!m$}}},
\end{align}
where $p_{\raisebox{-1pt}{\scriptsize{$n\!\rightarrow\!m$}}} =|\bra{E_{m}\suptiny{0}{0}{(\mathrm{f})}}U_{\Lambda}\ket{E_{n}\suptiny{0}{0}{(0)}}|^{2}$ as before (for details of this scheme see Appendix~\ref{sec:ttpm}). For the work estimate $\mean{W}_{\mathrm{non-id}}$ obtained within the non-ideal TPM scheme by (erroneously) associating outcomes $n$ and $m$ with energies $E_{n}\suptiny{0}{0}{(0)}$ and $E_{m}\suptiny{0}{0}{(\mathrm{f})}$ one obtains
\begin{align}
    \mean{W}_{\mathrm{non-id}}    &=\,\sum\limits_{m,n}
    p(m|n)\,p_{n}\suptiny{0}{0}{(0)}\,
    \bigl(E_{m}\suptiny{0}{0}{(\mathrm{f})}-E_{n}\suptiny{0}{0}{(0)}\bigr)\,\neq\,\mean{W}_{\Lambda},
    \label{eq:work nonideal measurements}
\end{align}
which generally does not match the ideal value in Eq.~(\ref{eq:ideal tpm work}).
In Appendix~\ref{sec:work nonid tpm} we explicitly derive the corrected expression to be
\begin{align}
    \mean{W}_{\mathrm{non-id}}    &=C_{\mathrm{max}}\mean{W}_{\Lambda}+
    \sum\limits_{\substack{m,n\\ l\neq n}}
    p_{\raisebox{-1pt}{\scriptsize{$l\!\rightarrow\!m$}}}\, q\suptiny{0}{-1}{(0)}_{ll|n}\,
    p_{n}\suptiny{0}{0}{(0)}\,
    \bigl(E_{m}\suptiny{0}{0}{(\mathrm{f})}-E_{n}\suptiny{0}{0}{(0)}\bigr).\nonumber\\[-4mm]
    \label{eq:ideal vs nonideal work}
\end{align}
Thus, we find that between the modified and ordinary TPM schemes, the deviation of the estimated work $\mean{W}_{\mathrm{non-id}}$ from the ideal value $\mean{W}_{\Lambda}$\footnote{obtained in the scenario with minimal energy UMC measurements using finite resources.} is characterised by two values: a modifying prefactor $C_{\mathrm{max}}<1$, along with an additional term which may be either positive or negative. Despite the generally complicated dependence on the details of the initial state, the measurement, energy spectrum and on the process $\Lambda$, we show, in Appendix~\ref{sec:id vs nonid tpm}, that
this deviation can be bounded,
\begin{align}
    |\mean{W}_{\mathrm{non-id}} - \mean{W}_{\Lambda}| &\leq\, (1-C_{\mathrm{max}})\N{H\suptiny{0}{0}{(\mathrm{f})}}_{\infty},
\end{align}
%%%%%%%%%%%%%%%%%%%%%%%%
\begin{figure}[ht!]
\begin{center}
\includegraphics[width=0.95\columnwidth,trim={0cm 3mm 0cm 0mm}]{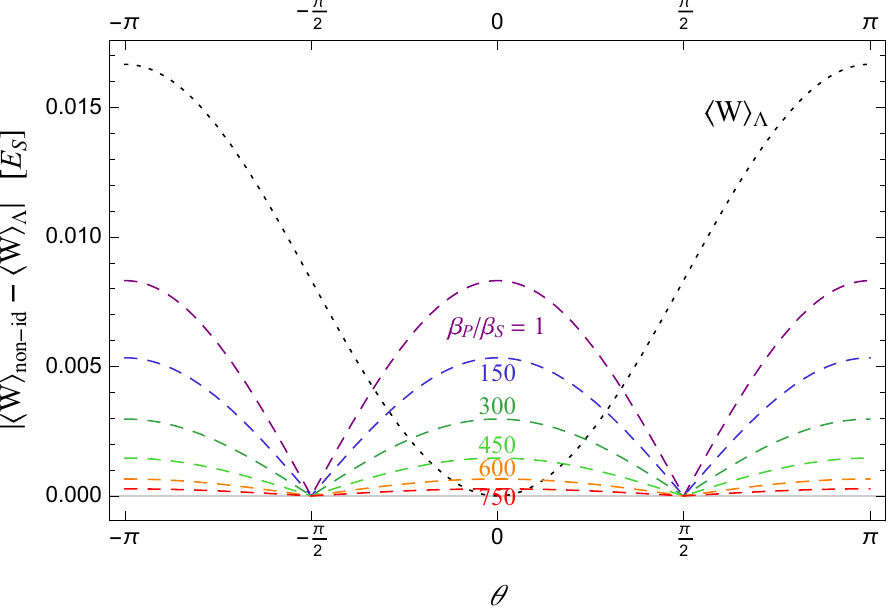}
\caption{Deviation from the ideal work estimate for a qubit system and a $3$-qubit pointer. For an atomic two-level system driven out of equilibrium by a classical electromagnetic field, the Rabi oscillations can be represented by the unitary $U_{\Lambda}(\theta) \equiv\exp\ch{-i\frac{\theta}{2}\sigma_y }$ (for a derivation see Appendix~\ref{appsec:driven atom}). The qubit is initially thermal, $\rho\suptiny{0}{0}{(0)}\Sys = \tau\suptiny{0}{0}{(0)}\Sys$, at room temperature $(k\subtiny{0}{0}{\mathrm{B}}\beta\Sys)^{-1} = 300$K, and has a Hamiltonian $H\Sys\suptiny{0}{0}{(0)}=H\Sys\suptiny{0}{0}{(\mathrm{f})}=-E\Sys\sigma_{z}/2$ with $\sigma_{z}=\ket{0}\!\!\bra{0}-\ket{1}\!\!\bra{1}$ and an energy gap in the microwave regime such that $\beta\Sys E\Sys \approx 1/30$. The 3-qubit pointer is initially in the thermal state $\tau(\beta\Poi)^{\otimes 3}$ where each of the three qubits has the same Hamiltonian $H\Poi=-E\Poi\sigma_{z}/2$ and we assume $E\Poi = E\Sys/10$. The dashed lines show the deviation $ |\mean{W}_{\mathrm{non-id}} - \mean{W}_{\Lambda}|$ in units of $E\Sys$ as a function of $\theta$ for different ratios of the system and pointer temperature, i.e., $\beta\Poi/\beta\Sys=1$, and from $\beta\Poi/\beta\Sys=150$ to $750$ in steps of $150$. The dotted black line shows the ideal work $\mean{W}_{\Lambda}$ (i.e., a pure state pointer). At $\theta = 0$, the deviation is maxmial while the ideal work estimate vanishes, whereas for $\theta=\pi$, the ratio $ |\mean{W}_{\mathrm{non-id}} - \mean{W}_{\Lambda}|/\mean{W}_{\Lambda}$ approaches $1/2$ (the precise value is $0.49875$).
}
\label{fig: difference vs wideal}
\end{center}
\end{figure}
%%%%%%%%%%
\noindent
where $\N{A}_{\infty} := \max_i || \sum_j a_{ij}^*||_1$ and $A=(a_{ij})$. In this sense, $C_{\mathrm{max}}$ can be thought of as representing the \emph{trustworthiness} of the work estimate, in addition to its resource cost. For a fixed process $\Lambda$, the closer $C_{\mathrm{max}}$ is to $1$, the smaller the potential distance of the work estimate from its ideal value, but also the higher the involved costs.

A principal purpose of this paper is to highlight that the difference $\mean{W}_{\mathrm{non-id}} - \mean{W}_{\Lambda}$ is non-negligible and in some cases, rather significant. Indeed, in Fig.~\ref{fig: difference vs wideal} we plot this difference for a standard system-pointer Hamiltonian with realistic parameters and find that the difference between the estimated ideal and non-ideal work is at times more than twice as large.

One should be careful not to confuse $\mean{W}_{\mathrm{non-id}}$ with the total work performed on the system or with the system's change in average energy $\Delta E_{\mathrm{non-id}}$. The latter can be expressed as
\begin{align}
    &\Delta E_{\mathrm{non-id}}    \,=\,\sum\limits_{m,n,k}
    q\suptiny{0}{-1}{(\mathrm{f})}_{kk|m}\,p(m|n)\,p_{n}\suptiny{0}{0}{(0)}  \,
    \bigl(E_{k}\suptiny{0}{0}{(\mathrm{f})}-E_{n}\suptiny{0}{0}{(0)}\bigr)
    \label{eq:av energy change nonideal measurements}\\
    &\ =\,
    C_{\mathrm{max}}^{2}\,\mean{W}_{\Lambda}
    +\!\!\!\!\sum\limits_{\substack{m,n\\ l\neq n, k\neq m}}\!\!
    q\suptiny{0}{-1}{(\mathrm{f})}_{kk|m}
    p_{\raisebox{-1pt}{\scriptsize{$l\!\rightarrow\!m$}}}
    q\suptiny{0}{-1}{(0)}_{ll|n}
    p_{m}\suptiny{0}{0}{(0)}
    \bigl(E_{k}\suptiny{0}{0}{(\mathrm{f})}\!-\!E_{n}\suptiny{0}{0}{(0)}\bigr)\nonumber
    \vspace*{-2mm}\\
    &\ \ \ \ +\,C_{\mathrm{max}}\sum\limits_{m,n} p_{n}\suptiny{0}{0}{(0)}
    \bigl[
        \sum\limits_{k\neq m}q\suptiny{0}{-1}{(\mathrm{f})}_{kk|m}\,
        p_{\raisebox{-1pt}{\scriptsize{$n\!\rightarrow\!m$}}}\,
        \bigl(E_{k}\suptiny{0}{0}{(\mathrm{f})}\!-\!E_{n}\suptiny{0}{0}{(0)}\bigr)
        \nonumber\\
    &\hspace*{2.75cm}+
        \sum\limits_{l\neq n}q\suptiny{0}{-1}{(0)}_{ll|n}\,
        p_{\raisebox{-1pt}{\scriptsize{$l\!\rightarrow\!m$}}}\,
        \bigl(E_{m}\suptiny{0}{0}{(\mathrm{f})}\!-\!E_{n}\suptiny{0}{0}{(0)}\bigr)
    \bigr],\nonumber
\end{align}
and generally contains contributions associated with {both the} work done on the system and heat transferred from the pointer to the system. Here, we make two observations. First, the work estimate $\mean{W}_{\mathrm{non-id}}$ generally does not match the ideal work estimate $\mean{W}_{\Lambda}=\Delta E_{\Lambda}$, the actual work done on the system, or the change in average energy $\Delta E_{\mathrm{non-id}}$. Second, while $\Delta E_{\mathrm{non-id}}$ captures the average energy change of the system, the average energy change of the pointer is not yet included and has to be considered separately~\cite{GuryanovaFriisHuber2018}. This contribution depends on the specific dimension of the pointer and the structure of its Hamiltonian, which in turn determine $C_{\mathrm{max}}$. Ideal measurements can be approached by increasing the pointer dimension (e.g., measuring a system with an $N$-qubit pointer and increasing $N$), or by cooling the pointer to a smaller but non-vanishing temperature (using a desired refrigeration paradigm~\cite{ClivazSilvaHaackBohrBraskBrunnerHuber2019a, RodriguezBrionesMartinMartinezKempfLaflamme2017, RodriguezBrionesLiPengMorWeinsteinLaflamme2017, AlhambraLostaglioPerry2019,  ClivazSilvaHaackBohrBraskBrunnerHuber2019b}); these measures all increase the work cost of the measurement~\cite[\ref{appsec:driven atom}]{GuryanovaFriisHuber2018}. Achieving the limit $C_{\mathrm{max}}\rightarrow 1$, requires infinite time, infinite energy, or infinite control (e.g., $N\rightarrow\infty$) and in this limit one also recovers $\mean{W}_{\mathrm{non-id}}=\mean{W}_{\Lambda}=\Delta E_{\Lambda}=\Delta E_{\mathrm{non-id}}$. When limited by finite resources, $\mean{W}_{\mathrm{non-id}}=\mean{W}_{\Lambda}$ can only be achieved for specific processes $\Lambda$, as we discuss in Appendix~\ref{sec:id vs nonid tpm}.

%%%%%%%%%%%%%%%%%%%%%%%%%%%%%%%%%%%%%%%%%%%%%%%%%%%%%%%%%%%%%%%%%%%%%%%%

\section{Fluctuation relations}\label{sec:fluctuation relations}

Besides work estimates, higher statistical moments of work are relevant in many contexts, from the study of quenched quantum many-body systems~\cite{Silva2008,DornerGooldCormickPaternostroVedral2012} to the performance of quantum thermal machines~\cite{CampisiPekolaFazio2015}, or the development of effective %unitary
charging protocols~\cite{FriisHuber2018}. For discussions about the quantum-classical correspondence of work distributions see Refs.~\cite{Allahverdyan2014, TalknerHaenggi2016, JarzynskiQuanRahav2015, PerarnauLlobetBaeumerHovhannisyanHuberAcin2017, Lostaglio2018}.

The quantification of work fluctuations~\cite{DornerClarkHeaneyFazioGooldVedral2013, MazzolaDeChiaraPaternostro2013, FuscoEtAl2014, RoncagliaCerisolaPaz2014} in the context of Jarzynski's and Crooks' fluctuation relations using the ideal TPM scheme has been studied in Ref.~\cite{CampisiHaenggiTalkner2011}. Here we are interested in exploring whether these universal relations are recovered within the non-ideal TPM scheme. In particular we explore the statistical properties of the work estimate $\mean{W}_{\mathrm{non-id}}$. We first focus on the Jarzynski equality~\cite{Jarzynski1997} for which the quantity of interest is the work functional $\mean{e^{-\beta W}}$. Using the properties of the non-ideal TPM scheme we obtain
\begin{align}
    \mean{e^{-\beta W}}_{\mathrm{non-id}}   &= \chi\, e^{-\beta \Delta F},
    \label{eq:JJ nonideal}
\end{align}
where we have introduced the correction term
\begin{align} \label{eq:correction}
    \chi := \frac{1}{\mathcal{Z} \suptiny{0}{0}{(\mathrm{f})}}
    \sum\limits_{n,m,l}\!
    e^{-\beta E_{m}\suptiny{0}{0}{(\mathrm{f})}}
    q\suptiny{0}{-1}{(0)}_{ll|n}
    |\bra{\nl E_{m}\suptiny{0}{0}{(\mathrm{f})}\nl}
    U_{\Lambda}
    \ket{\nl E_{l}\suptiny{0}{0}{(\mathrm{0})}\nl}|^2.
\end{align}
Here, the partition function is $\mathcal{Z}\suptiny{0}{0}{(\mathrm{f})}:=\tr[e^{-\beta H\suptiny{0}{0}{(\mathrm{f})}}]$, and $\Delta F:= k\subtiny{0}{0}{\mathrm{B}} T \log(\mathcal{Z}\suptiny{0}{0}{(0)}/\mathcal{Z}\suptiny{0}{0}{(\mathrm{f})})$ is the difference in Helmholtz free energies between the initial state and the thermal state $\tau\suptiny{0}{0}{(\mathrm{f})}$ w.r.t. %the final Hamiltonian
$H\suptiny{0}{0}{(\mathrm{f})}$. By construction $0 \leq \chi \leq d_S$.
When $\chi = 1$ one has the original Jarzynski relation, $\mean{e^{-\beta W}}   = e^{-\beta \Delta F}$, which is satisfied by all ideal measurements. We notice that the breakdown of the original Jarzynski equality is expected here, since the invasiveness of the first non-ideal measurement spoils the requirement that the system starts the process $\Lambda$ in thermal equilibrium.

Nonetheless, we find that there exists a class of non-ideal measurements for which $\chi = 1$ as well. This is the case when $\sum_{n}q\suptiny{0}{-1}{(0)}_{ll|n}=1\ \forall l$. This implies that in order to recover Jarzynski's relation, the matrix in Eq.~(\ref{eq:correction}) needs to be doubly stochastic. We call these measurements \emph{minimally invasive UMC measurements}, and they correspond to unital maps, whose average effect is to preserve the identity operator on the system (for details see Appendix~\ref{appsec:min invasive UMC measurements}). Using Eq.~\eqref{eq:JJ nonideal} and Jensen's inequality, $\exp{\mean{x}} \leq \mean{\exp{x}}$, we can write the following second-law-like inequality
\begin{align}
    \mean{W}_{\mathrm{non-id}} \geq \Delta F - k\subtiny{0}{0}{\mathrm{B}} T \log{\chi},
\end{align}
which provides a lower bound for the non-ideal TPM work estimate $\mean{W}_{\mathrm{non-id}}$. Since $0 \leq \chi \leq d_S$, the extra term $- k\subtiny{0}{0}{\mathrm{B}} T \log{\chi}$ can be either positive or negative. For non-ideal measurements that are unital ($\chi = 1$) this term vanishes and the equation above reduces to the usual second-law inequality $\mean{W}_{\mathrm{non-id}} \geq \Delta F$.

Interestingly, we observe that minimal energy UMC measurements do not, in general, correspond to minimally invasive measurements (except for the special case of a single-qubit system, see Appendix~\ref{appsec:min invasive UMC measurements}). This can be understood in the following way: For minimal energy UMC measurements, all elements of the post-interaction state $\tilde{\rho}\SP$ that depend on the outcome probabilities $\rho_{ii}=\bra{i}\rho\Sys\ket{i}$ can be collected in a correlation matrix of dimension $d\Poi/d\Sys\times d\Poi/d\Sys$ with blocks $\Gamma_{ij}=\ket{i}\!\!\bra{i}\otimes\Pi_{j}\tilde{\rho}\SP\ket{i}\!\!\bra{i}\otimes\Pi_{j}$. Property \eqref{item:unbiased main} (unbiasedness) fixes $\tr(\Gamma_{ij})=\rho_{jj}q_{ii|j}$ and maximal correlations are achieved when $q_{ii|i}=C_{\mathrm{max}}$ is satisfied. Minimal energy means that the remaining $q_{ii|j}$ for $i\neq j$ within each column are ordered such that $q_{ii|j}\geq q_{kk|j}\forall k\geq i$ with $k\neq j$. This is generally not compatible with row sums $\sum_{j}q_{ii|j}=1$ required for $q_{ii|j}$ to be doubly stochastic and thus correspond to a minimally invasive measurement.

One can, however, ensure that the measurement is minimally invasive \textit{and} unbiased at the expense of moving away from the energy minimum. One can show that this additional cost only depends on the system Hamiltonian (and not the pointer Hamiltonian). Thus, by accepting this cost, non-ideal measurements can satisfy Jarzynski's relation.

%%%%%%%%%%%%%%%%%%%%%%%%

\begin{figure}[t]
\begin{center}
\includegraphics[scale=0.334]{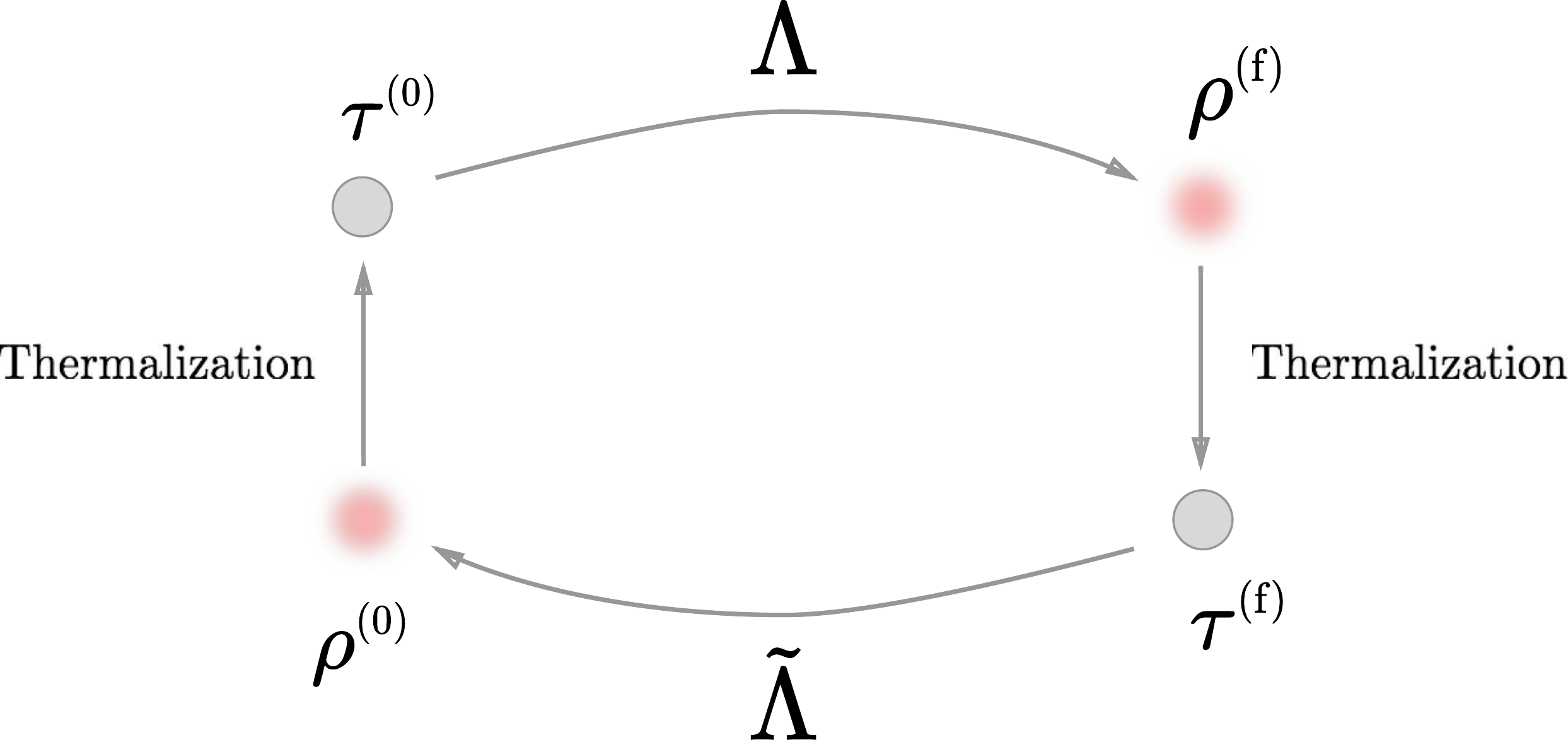}
\caption{Schematic representing the non-reversibility of the process $\Lambda$. The system is initially at a thermal state $\tau\suptiny{0}{0}{(0)} = \exp(-\beta H\suptiny{0}{0}{(0)})/\mathcal{Z}\suptiny{0}{0}{(0)}$ is driven out of equilibrium by the unitary transformation $U_{\Lambda}$ to the state $\rho\suptiny{0}{0}{(\mathrm{f})} = \U{ \tau\suptiny{0}{0}{(0)}}$. At the end of the process, the system is left to thermalise to the equilibrium state $\tau\suptiny{0}{0}{(\mathrm{f})} = \exp(-\beta H\suptiny{0}{0}{(\mathrm{f})})/\mathcal{Z}\suptiny{0}{0}{(\mathrm{f})}$,  described by the Hamiltonian $H\suptiny{0}{0}{(\mathrm{f})}$. In the time reversed protocol, the thermal state $\tau\suptiny{0}{0}{(\mathrm{f})}$ is driven out of equilibrium to state $\rho\suptiny{0}{0}{(0)} = {U_{\tilde\Lambda} \tau\suptiny{0}{0}{(\mathrm{f})}U^{\dagger}_{\tilde\Lambda}} $. The process is reversible if $\rho\suptiny{0}{0}{(\mathrm{f})} = \tau\suptiny{0}{0}{(\mathrm{f})}$ and $\rho\suptiny{0}{0}{(0)} = \tau\suptiny{0}{0}{(0)}$.}
\label{fig:dissipation}
\end{center}
\end{figure}
To investigate the effect of non-ideal measurements on the irreversibilty of a process, we turn to
Crooks' theorem~\cite{Crooks1999}. This relates the probabilities $P_{\mathrm{F}}(W)$ of performing some work during a realisation of the TPM scheme and $P_{\mathrm{B}}(-W)$ for extracting the same amount of work when the time-reversed protocol $\tilde{\Lambda}$ is implemented:
\begin{equation}
    P_{\mathrm{B}}(-W) = e^{-\beta(W - \Delta F)} {P}_{\mathrm{F}}(W).
    \label{eq:Crooks}
\end{equation}
The quantity $W-\Delta F$ is usually referred to as the {\it dissipated work}, the work which is lost when the final state of the TPM after the protocol, $\rho\suptiny{0}{0}{(\mathrm{f})}$ in Eq.~(\ref{eq:process}), relaxes back to equilibrium at temperature $T$.
As we discuss in  Appendix \ref{sec:non ideal crooks}, Crooks' relation \eqref{eq:Crooks} is not recovered in the non-ideal TPM scheme, not even for minimally invasive UMC measurements (unlike the Jarzynski equality).
This is in contrast to the ideal TPM scheme, where Crooks' theorem can be recovered for all unitary and unital maps~\cite{CampisiTalknerHaenggi2010, Rastegin2013, WatanabeVenkateshTalknerCampisiHaenggi2014, ManzanoHorowitzParrondo2015}.
The reason is that, here, both non-ideal measurements act as (independent) noise sources, disturbing the initial states of the forward and backward TPM processes, respectively.
Reestablishing Crooks' fluctuation theorem may eventually require considering the work performed in the measurement processes and, therefore, taking into account the energy changes in the pointers.

Eq.~(\ref{eq:Crooks}) expresses the fact that finite-time processes that drive systems out of equilibrium  are irreversible and thus the consumed work is unlikely to be recovered when reversing the protocol (See Fig.~\ref{fig:dissipation}). This irreversibility can be captured~\cite{KawaiParrondoVanDenBroeck2007, ParrondoVanDenBroeckKawai2009} by the average of the dissipated work appearing in Eq.~(\ref{eq:Crooks}), and is related to the entropy production during a hypothetical relaxation of $\rho\suptiny{0}{0}{(\mathrm{f})}$ to the thermal equilibrium state $ \tau\suptiny{0}{0}{(\mathrm{f})}$~\cite{ManzanoHorowitzParrondo2015, DeffnerPazZurek2016}. When the measurements for determining this dissipated work are non-ideal, additional entropy is produced, resulting in greater energy dissipation in the final relaxation. In Appendix~\ref{app:irre}, we calculate the average estimated work performed in addition to the free energy change to be
\begin{align}
    {\mean{W}_{\mathrm{non-id}}} - \Delta F =  k\subtiny{0}{0}{\mathrm{B}} T [\Delta S\suptiny{0}{0}{(0)} + D(\tilde{\rho}\suptiny{0}{0}{(\mathrm{f})}||\hspace*{1pt} \tau\suptiny{0}{0}{(\mathrm{f})})],
    \label{eq:dissipation}
\end{align}
where $\tilde{\rho}\suptiny{0}{0}{(\mathrm{f})}$ is the unconditional final state after the process $\Lambda$, $D(\rho\suptiny{0}{0}{(\mathrm{f})}||\tau\suptiny{0}{0}{(\mathrm{f})})  \geq 0$ is the relative entropy quantifying the irreversibility in the ideal process~\cite{Spohn1978, BreuerPetruccione2002} and $\Delta S\suptiny{0}{0}{(0)} = S(\tilde{\rho}\suptiny{0}{0}{(\mathrm{0})})- S(\tau\suptiny{0}{0}{(0)})$ is the corresponding change in von~Neumann entropy of the system due to the first non-ideal measurement. For non-ideal, minimally invasive measurements $({\mean{W}_{\mathrm{non-id}}} - \Delta F) \ge 0$  in which case it may be interpreted as the entropy production. For these measurements, the von~Neumann entropy of the system cannot decrease~\cite{BreuerPetruccione2002} and the entropy of the pointer during the measurement does not change (see, e.g.,~\cite{ManzanoHorowitzParrondo2015}). Thus, the total entropy produced in the measurement coincides with the entropy change in the system, $\Delta S\suptiny{0}{0}{(0)} \geq 0$. When ideal projective measurements are considered, $\Delta S\suptiny{0}{0}{(0)} = 0$, and the usual expression for the entropy production is recovered. In this ideal case, Kawai, Parrondo and Van den Broeck ~\cite{KawaiParrondoVanDenBroeck2007, ParrondoVanDenBroeckKawai2009} derived an important result in non-equilibrium thermodynamics, which is closely related to Crook's fluctuation theorem
\begin{align}
    \mean{W}_\Lambda - \Delta F = k\subtiny{0}{0}{\mathrm{B}} T D(\rho_{\mathrm{F}}(t)|| \Theta^\dagger \rho_{\mathrm{B}}(t_\mathrm{f}- t) \Theta).
    \label{eq:Idealdissipation}
\end{align}
Here, $\rho_{\mathrm{F}}(t)$ and $\rho_{\mathrm{B}}(t_\mathrm{f} - t)$ are the density operators in the forward and backward processes taken at the same instance of time $t$ and $\Theta$ is the time-reversal operator in quantum mechanics. $\Theta$ is anti-unitary, satisfies $\Theta i \mathds{1} = -i \mathds{1} \Theta$ and $\Theta \Theta^\dagger = \Theta^\dagger \Theta = \mathds{1}$, and is responsible for changing the sign of odd variables under time reversal (momentum, magnetic field, etc.)~\cite{Haake}.
Equation~\eqref{eq:Idealdissipation} establishes a deep relationship between the physical and information-theoretical notion of irreversibility. Namely, it connects the dissipated work (left-hand side), which is a physical measure of irreversibility, with the relative entropy at any snapshot of time (right-hand side), which is an information-theoretical measure. We are able to extend this result to the non-ideal TPM scheme by using the generalised dissipation relation in Eq.~\eqref{eq:dissipation}. In the modified scheme we get that
\begin{align}
    \mean{W}_{\mathrm{non-id}} - \Delta F = k\subtiny{0}{0}{\mathrm{B}} T [&\Delta S_0 + \Delta D_{\mathrm{f}} +\nonumber\\ +& D(\rho_{\mathrm{F}}(t)||\hspace*{1pt} \Theta^\dagger \rho_{\mathrm{B}}(t_\mathrm{f}-t)\Theta)],
    \label{eq:dissipation2}
\end{align}
where $\rho_{\mathrm{F}}(t) = U_\Lambda(t,0) \tilde \rho \suptiny{0}{0}{(\mathrm{f})} U^\dagger_\Lambda(t,0)$ is the system state at intermediate time $0 \leq t \leq t_\mathrm{f}$ in the forward process and $\rho_{\mathrm{B}}(t_\mathrm{f} - t) = U_{\tilde \Lambda}(t_\mathrm{f}-t,0) \rho_{\mathrm{B}} \suptiny{0}{0}{(\mathrm{f})} U^\dagger_{\tilde \Lambda}(t_\mathrm{f}-t,0)$ is the state of the system at the same instance of time in the backward process\footnote{
$U_{\tilde \Lambda}(t_\mathrm{f} - t,0)$ is the unitary evolution generated by the time-reversed protocol $\tilde \Lambda$.}
(see Appendix~\ref{sec:non ideal crooks}). We have introduced the correction term
\begin{align}
    \Delta D_{\mathrm{f}}
    &= \tr[\tilde{\rho}\suptiny{0}{0}{(\mathrm{f})} \left( \log \tau\suptiny{0}{0}{(\mathrm{f})} - \log \rho_{\mathrm{B}}\suptiny{0}{0}{(\mathrm{f})} \right)],
\end{align}
whereby starting in $\tau\suptiny{0}{0}{(\mathrm{f})}$, the average system state $\rho_{\mathrm{B}}\suptiny{0}{0}{(\mathrm{f})}$ is obtained after the first non-ideal measurement in the backward process. The ideal case in Eq.~\eqref{eq:Idealdissipation} is recovered when  $\rho_{\mathrm{B}}\suptiny{0}{0}{(\mathrm{f})} = \tau\suptiny{0}{0}{(\mathrm{f})}$ and $\tilde{\rho}\suptiny{0}{0}{(\mathrm{0})} = \tau\suptiny{0}{0}{(0)}$, making $\Delta D_\mathrm{f} = 0$ and $\Delta S\suptiny{0}{0}{(0)} = 0$.
Finally, we stress that, analogously to what happens with the result in Eq.~\eqref{eq:Idealdissipation} for the ideal TPM scheme, Eqs.~\eqref{eq:dissipation} and \eqref{eq:dissipation2} can be turned into inequalities for open system dynamics (see Appendix~\ref{app:irre}).

%%%%%%%%%%%%%%%%%%%%%%%%%%%%%%%%%%%%%%%%%%%%%%%%%%%%%%%%%%%%%%%%%%%%%%%%%%%%%%%%%%%%%%%%%%%%%%%%%%

\section{Discussion}\label{sec:discussion}

We have studied the consequences of fundamentally unavoidable measurement imperfections on the estimation of work and its fluctuations in out-of-equilibrium processes. Non-ideal measurements lead to a mismatch between the obtained estimate $\mean{W}_{\mathrm{non-id}}$, the desired ideal estimate $\mean{W}_{\Lambda}$, and the actual work performed on the system during the non-ideal TPM. In addition, an energy cost is incurred for operating the measurement apparatus. This leads to the conclusion that the process of estimating work itself has a work cost, which increases with increasing precision of the estimate. Moreover, we find that the statistical properties of the non-ideal estimate $\mean{W}_{\mathrm{non-id}}$ are modified. While the celebrated Jarzynski relation may be recovered exactly by imposing specific conditions on the measurement scheme, the more general Crooks theorem no longer holds. In this context, we discussed the connection between the non-ideal work estimate and the entropy production in the TPM scheme, and extended previous results for the relation between dissipation and irreversibility.

These results are of particular relevance for work extraction: When the costs for estimating the extracted work are of the order of the extracted work itself the usefulness of the procedure is dramatically limited. Conceptually, our results can be seen as a constructive resolution of the perceived shortcomings of the TPM discussed in~\cite{DeffnerPazZurek2016}. It might also be interesting to consider work estimates as well as Jarzynski's equality and Crooks' theorem in more general contexts, such as including feedback control strategies~\cite{FunoUedaSagawa2018, PottsSamuelsson2018}.

Our results about the validity of the Jarzynski equality and Crooks' fluctuation theorem for the non-ideal TPM scheme are in agreement with very recent results reported in Ref.~\protect\cite{ItoTalknerVenkatesh2019}, which appeared during the final stages of preparing this manuscript.

%%%

\emph{Acknowledgements}.
T.D. acknowledges support from the Brazilian agency CNPq INCT-IQ through the project (465469/2014-0) and from the Austrian Science Fund (FWF) through the project P 31339-N27.
Y.G., M.H. and N.F. acknowledge support from the Austrian Science Fund (FWF) through the START project Y879-N2, the joint Czech-Austrian project MultiQUEST (I 3053-N27 and GF17-33780L) and the project P 31339-N27.

%%%%%%%%%%%%%%%%%%%%%%%%%%%%%%%%%%%%%%%%%%%%%%%%%%%%%%%%%%%%%%%%%%%%%%%%%%%%%%%%%%%%%%%%%%%%%%%%%%%%

\bibliographystyle{apsrev4-1fixed_with_article_titles_full_names}
\bibliography{bibfile}

%merlin.mbs apsrev4-1.bst 2010-07-25 4.21a (PWD, AO, DPC) hacked
%Control: key (0)
%Control: author (72) initials jnrlst
%Control: editor formatted (1) identically to author
%Control: production of article title (-1) disabled
%Control: page (0) single
%Control: year (1) truncated
%Control: production of eprint (0) enabled
\begin{thebibliography}{77}%
\makeatletter
\providecommand \@ifxundefined [1]{%
 \@ifx{#1\undefined}
}%
\providecommand \@ifnum [1]{%
 \ifnum #1\expandafter \@firstoftwo
 \else \expandafter \@secondoftwo
 \fi
}%
\providecommand \@ifx [1]{%
 \ifx #1\expandafter \@firstoftwo
 \else \expandafter \@secondoftwo
 \fi
}%
\providecommand \natexlab [1]{#1}%
\providecommand \enquote  [1]{#1}%
\providecommand \bibnamefont  [1]{#1}%
\providecommand \bibfnamefont [1]{#1}%
\providecommand \citenamefont [1]{#1}%
\providecommand \href@noop [0]{\@secondoftwo}%
\providecommand \href [0]{\begingroup \@sanitize@url \@href}%
\providecommand \@href[1]{\@@startlink{#1}\@@href}%
\providecommand \@@href[1]{\endgroup#1\@@endlink}%
\providecommand \@sanitize@url [0]{\catcode `\\12\catcode `\$12\catcode
  `\&12\catcode `\#12\catcode `\^12\catcode `\_12\catcode `\%12\relax}%
\providecommand \@@startlink[1]{}%
\providecommand \@@endlink[0]{}%
\providecommand \url  [0]{\begingroup\@sanitize@url \@url }%
\providecommand \@url [1]{\endgroup\@href {#1}{\urlprefix }}%
\providecommand \urlprefix  [0]{URL }%
\providecommand \Eprint [0]{\href }%
\providecommand \doibase [0]{http://dx.doi.org/}%
\providecommand \selectlanguage [0]{\@gobble}%
\providecommand \bibinfo  [0]{\@secondoftwo}%
\providecommand \bibfield  [0]{\@secondoftwo}%
\providecommand \translation [1]{[#1]}%
\providecommand \BibitemOpen [0]{}%
\providecommand \bibitemStop [0]{}%
\providecommand \bibitemNoStop [0]{.\EOS\space}%
\providecommand \EOS [0]{\spacefactor3000\relax}%
\providecommand \BibitemShut  [1]{\csname bibitem#1\endcsname}%
\let\auto@bib@innerbib\@empty
%</preamble>
\bibitem [{\citenamefont {Leff}\ and\ \citenamefont {Rex}(2003)}]{LeffRex2003}%
  \BibitemOpen
  \bibinfo {editor} {\bibfnamefont {Harvey}\ \bibnamefont {Leff}}\ and\
  \bibinfo {editor} {\bibfnamefont {Andrew~F.}\ \bibnamefont {Rex}},\ eds.,\
  \href@noop {} {\emph {\bibinfo {title} {Maxwell Demon 2: Entropy, Classical
  and Quantum Information, Computing}}}\ (\bibinfo  {publisher} {Institute of
  Physics},\ \bibinfo {address} {Bristol},\ \bibinfo {year} {2003})\BibitemShut
  {NoStop}%
\bibitem [{\citenamefont {Mayurama}\ \emph {et~al.}(2009)\citenamefont
  {Mayurama}, \citenamefont {Nori},\ and\ \citenamefont
  {Vedral}}]{MayuramaNoriVedral2009}%
  \BibitemOpen
  \bibfield  {author} {\bibinfo {author} {\bibfnamefont {Koji}\ \bibnamefont
  {Mayurama}}, \bibinfo {author} {\bibfnamefont {Franco}\ \bibnamefont {Nori}},
  \ and\ \bibinfo {author} {\bibfnamefont {Vlatko}\ \bibnamefont {Vedral}},\
  }\emph {\enquote {\bibinfo {title} {{Colloquium: The physics of Maxwell's
  demon and information}},}\ }\href {\doibase 10.1103/RevModPhys.81.1}
  {\bibfield  {journal} {\bibinfo  {journal} {Rev. Mod. Phys.}\ }\textbf
  {\bibinfo {volume} {81}},\ \bibinfo {pages} {1} (\bibinfo {year} {2009})},\
  \Eprint {http://arxiv.org/abs/arXiv:0707.3400} {arXiv:0707.3400}\BibitemShut
  {NoStop}%
\bibitem [{\citenamefont {Sagawa}\ and\ \citenamefont
  {Ueda}(2009)}]{SagawaUeda2009}%
  \BibitemOpen
  \bibfield  {author} {\bibinfo {author} {\bibfnamefont {Takahiro}\
  \bibnamefont {Sagawa}}\ and\ \bibinfo {author} {\bibfnamefont {Masahito}\
  \bibnamefont {Ueda}},\ }\emph {\enquote {\bibinfo {title} {{Minimal Energy
  Cost for Thermodynamic Information Processing: Measurement and Information
  Erasure}},}\ }\href {\doibase 10.1103/PhysRevLett.102.250602} {\bibfield
  {journal} {\bibinfo  {journal} {Phys. Rev. Lett.}\ }\textbf {\bibinfo
  {volume} {102}},\ \bibinfo {pages} {250602} (\bibinfo {year} {2009})},\
  \Eprint {http://arxiv.org/abs/arXiv:0809.4098} {arXiv:0809.4098}\BibitemShut
  {NoStop}%
\bibitem [{\citenamefont {Jacobs}(2012)}]{Jacobs2012}%
  \BibitemOpen
  \bibfield  {author} {\bibinfo {author} {\bibfnamefont {Kurt}\ \bibnamefont
  {Jacobs}},\ }\emph {\enquote {\bibinfo {title} {Quantum measurement and the
  first law of thermodynamics: the energy cost of measurement is the work value
  of the acquired information},}\ }\href {\doibase 10.1103/PhysRevE.86.040106}
  {\bibfield  {journal} {\bibinfo  {journal} {Phys. Rev. E}\ }\textbf {\bibinfo
  {volume} {86}},\ \bibinfo {pages} {040106(R)} (\bibinfo {year} {2012})},\
  \Eprint {http://arxiv.org/abs/arXiv:1208.1561} {arXiv:1208.1561}\BibitemShut
  {NoStop}%
\bibitem [{\citenamefont {Lipka-Bartosik}\ and\ \citenamefont
  {Demkowicz-Dobrzanski}(2018)}]{LipkaBartosikDemkowiczDobrzanski2018}%
  \BibitemOpen
  \bibfield  {author} {\bibinfo {author} {\bibfnamefont {Patryk}\ \bibnamefont
  {Lipka-Bartosik}}\ and\ \bibinfo {author} {\bibfnamefont {Rafal}\
  \bibnamefont {Demkowicz-Dobrzanski}},\ }\emph {\enquote {\bibinfo {title}
  {Thermodynamic work cost of quantum estimation protocols},}\ }\href {\doibase
  10.1088/1751-8121/aae664} {\bibfield  {journal} {\bibinfo  {journal} {J.
  Phys. A: Math. Theor.}\ }\textbf {\bibinfo {volume} {51}},\ \bibinfo {pages}
  {474001} (\bibinfo {year} {2018})},\ \Eprint
  {http://arxiv.org/abs/arXiv:1805.01477} {arXiv:1805.01477}\BibitemShut
  {NoStop}%
\bibitem [{\citenamefont {Landauer}(1961)}]{Landauer1961}%
  \BibitemOpen
  \bibfield  {author} {\bibinfo {author} {\bibfnamefont {Rolf}\ \bibnamefont
  {Landauer}},\ }\emph {\enquote {\bibinfo {title} {{Irreversibility and Heat
  Generation in the Computing Process}},}\ }\href {\doibase 10.1147/rd.53.0183}
  {\bibfield  {journal} {\bibinfo  {journal} {IBM J. Res. Dev.}\ }\textbf
  {\bibinfo {volume} {5}},\ \bibinfo {pages} {183} (\bibinfo {year}
  {1961})}\BibitemShut {NoStop}%
\bibitem [{\citenamefont {Bennett}(1982)}]{Bennett1982}%
  \BibitemOpen
  \bibfield  {author} {\bibinfo {author} {\bibfnamefont {Charles~H.}\
  \bibnamefont {Bennett}},\ }\emph {\enquote {\bibinfo {title} {The
  thermodynamics of computation -- a review},}\ }\href {\doibase
  10.1007/BF02084158} {\bibfield  {journal} {\bibinfo  {journal} {Int. J.
  Theor. Phys.}\ }\textbf {\bibinfo {volume} {21}},\ \bibinfo {pages} {905}
  (\bibinfo {year} {1982})}\BibitemShut {NoStop}%
\bibitem [{\citenamefont {Esposito}\ and\ \citenamefont {Van~den
  Broeck}(2011)}]{EspositoVanDenBroeck2011}%
  \BibitemOpen
  \bibfield  {author} {\bibinfo {author} {\bibfnamefont {Massimiliano}\
  \bibnamefont {Esposito}}\ and\ \bibinfo {author} {\bibfnamefont {Christian}\
  \bibnamefont {Van~den Broeck}},\ }\emph {\enquote {\bibinfo {title} {Second
  law and landauer principle far from equilibrium},}\ }\href {\doibase
  10.1209/0295-5075/95/40004} {\bibfield  {journal} {\bibinfo  {journal}
  {Europhys. Lett.}\ }\textbf {\bibinfo {volume} {95}},\ \bibinfo {pages}
  {40004} (\bibinfo {year} {2011})},\ \Eprint
  {http://arxiv.org/abs/arXiv:1104.5165} {arXiv:1104.5165}\BibitemShut
  {NoStop}%
\bibitem [{\citenamefont {Reeb}\ and\ \citenamefont
  {Wolf}(2014)}]{ReebWolf2014}%
  \BibitemOpen
  \bibfield  {author} {\bibinfo {author} {\bibfnamefont {David}\ \bibnamefont
  {Reeb}}\ and\ \bibinfo {author} {\bibfnamefont {Michael~M.}\ \bibnamefont
  {Wolf}},\ }\emph {\enquote {\bibinfo {title} {{An improved Landauer Principle
  with finite-size corrections}},}\ }\href {\doibase
  10.1088/1367-2630/16/10/103011} {\bibfield  {journal} {\bibinfo  {journal}
  {New J. Phys.}\ }\textbf {\bibinfo {volume} {16}},\ \bibinfo {pages} {103011}
  (\bibinfo {year} {2014})},\ \Eprint {http://arxiv.org/abs/arXiv:1306.4352}
  {arXiv:1306.4352}\BibitemShut {NoStop}%
\bibitem [{\citenamefont {Abdelkhalek}\ \emph {et~al.}(2016)\citenamefont
  {Abdelkhalek}, \citenamefont {Nakata},\ and\ \citenamefont
  {Reeb}}]{AbdelkhalekNakataReeb2016}%
  \BibitemOpen
  \bibfield  {author} {\bibinfo {author} {\bibfnamefont {Kais}\ \bibnamefont
  {Abdelkhalek}}, \bibinfo {author} {\bibfnamefont {Yoshifumi}\ \bibnamefont
  {Nakata}}, \ and\ \bibinfo {author} {\bibfnamefont {David}\ \bibnamefont
  {Reeb}},\ }\href@noop {} {\emph {\enquote {\bibinfo {title} {Fundamental
  energy cost for quantum measurement},}\ }} (\bibinfo {year} {2016}),\ \Eprint
  {http://arxiv.org/abs/arXiv:1609.06981} {arXiv:1609.06981}\BibitemShut
  {NoStop}%
\bibitem [{\citenamefont {Parrondo}\ \emph {et~al.}(2015)\citenamefont
  {Parrondo}, \citenamefont {Horowitz},\ and\ \citenamefont
  {Sagawa}}]{ParrondoHorowitzSagawa2015}%
  \BibitemOpen
  \bibfield  {author} {\bibinfo {author} {\bibfnamefont {Juan M.~R.}\
  \bibnamefont {Parrondo}}, \bibinfo {author} {\bibfnamefont {Jordan~M.}\
  \bibnamefont {Horowitz}}, \ and\ \bibinfo {author} {\bibfnamefont {Takahiro}\
  \bibnamefont {Sagawa}},\ }\emph {\enquote {\bibinfo {title} {Thermodynamics
  of information},}\ }\href {\doibase 10.1038/nphys3230} {\bibfield  {journal}
  {\bibinfo  {journal} {Nat. Phys.}\ }\textbf {\bibinfo {volume} {11}},\
  \bibinfo {pages} {131} (\bibinfo {year} {2015})}\BibitemShut {NoStop}%
\bibitem [{\citenamefont {Kammerlander}\ and\ \citenamefont
  {Anders}(2016)}]{KammerlanderAnders2016}%
  \BibitemOpen
  \bibfield  {author} {\bibinfo {author} {\bibfnamefont {Philipp}\ \bibnamefont
  {Kammerlander}}\ and\ \bibinfo {author} {\bibfnamefont {Janet}\ \bibnamefont
  {Anders}},\ }\emph {\enquote {\bibinfo {title} {Coherence and measurement in
  quantum thermodynamics},}\ }\href {\doibase 10.1038/srep22174} {\bibfield
  {journal} {\bibinfo  {journal} {Sci. Rep.}\ }\textbf {\bibinfo {volume}
  {6}},\ \bibinfo {pages} {22174} (\bibinfo {year} {2016})},\ \Eprint
  {http://arxiv.org/abs/arXiv:1502.02673} {arXiv:1502.02673}\BibitemShut
  {NoStop}%
\bibitem [{\citenamefont {Manzano}\ \emph {et~al.}(2018)\citenamefont
  {Manzano}, \citenamefont {Plastina},\ and\ \citenamefont
  {Zambrini}}]{ManzanoPlastinaZambrini2018}%
  \BibitemOpen
  \bibfield  {author} {\bibinfo {author} {\bibfnamefont {Gonzalo}\ \bibnamefont
  {Manzano}}, \bibinfo {author} {\bibfnamefont {Francesco}\ \bibnamefont
  {Plastina}}, \ and\ \bibinfo {author} {\bibfnamefont {Roberta}\ \bibnamefont
  {Zambrini}},\ }\emph {\enquote {\bibinfo {title} {{Optimal Work Extraction
  and Thermodynamics of Quantum Measurements and Correlations}},}\ }\href
  {\doibase 10.1103/PhysRevLett.121.120602} {\bibfield  {journal} {\bibinfo
  {journal} {Phys. Rev. Lett.}\ }\textbf {\bibinfo {volume} {121}},\ \bibinfo
  {pages} {120602} (\bibinfo {year} {2018})},\ \Eprint
  {http://arxiv.org/abs/arXiv:1805.08184} {arXiv:1805.08184}\BibitemShut
  {NoStop}%
\bibitem [{\citenamefont {Guryanova}\ \emph {et~al.}(2018)\citenamefont
  {Guryanova}, \citenamefont {Friis},\ and\ \citenamefont
  {Huber}}]{GuryanovaFriisHuber2018}%
  \BibitemOpen
  \bibfield  {author} {\bibinfo {author} {\bibfnamefont {Yelena}\ \bibnamefont
  {Guryanova}}, \bibinfo {author} {\bibfnamefont {Nicolai}\ \bibnamefont
  {Friis}}, \ and\ \bibinfo {author} {\bibfnamefont {Marcus}\ \bibnamefont
  {Huber}},\ }\href {https://arxiv.org/abs/1805.11899} {\emph {\enquote
  {\bibinfo {title} {Ideal projective measurements have infinite resource
  costs},}\ }} (\bibinfo {year} {2018}),\ \Eprint
  {http://arxiv.org/abs/arXiv:1805.11899} {arXiv:1805.11899}\BibitemShut
  {NoStop}%
\bibitem [{\citenamefont {Masanes}\ and\ \citenamefont
  {Oppenheim}(2017)}]{MasanesOppenheim2017}%
  \BibitemOpen
  \bibfield  {author} {\bibinfo {author} {\bibfnamefont {Lluis}\ \bibnamefont
  {Masanes}}\ and\ \bibinfo {author} {\bibfnamefont {Jonathan}\ \bibnamefont
  {Oppenheim}},\ }\emph {\enquote {\bibinfo {title} {{A general derivation and
  quantification of the third law of thermodynamics}},}\ }\href {\doibase
  10.1038/ncomms14538} {\bibfield  {journal} {\bibinfo  {journal} {Nat.
  Commun.}\ }\textbf {\bibinfo {volume} {8}},\ \bibinfo {pages} {14538}
  (\bibinfo {year} {2017})},\ \Eprint {http://arxiv.org/abs/arXiv:1412.3828}
  {arXiv:1412.3828}\BibitemShut {NoStop}%
\bibitem [{\citenamefont {Dorner}\ \emph {et~al.}(2013)\citenamefont {Dorner},
  \citenamefont {Clark}, \citenamefont {Heaney}, \citenamefont {Fazio},
  \citenamefont {Goold},\ and\ \citenamefont
  {Vedral}}]{DornerClarkHeaneyFazioGooldVedral2013}%
  \BibitemOpen
  \bibfield  {author} {\bibinfo {author} {\bibfnamefont {Ross}\ \bibnamefont
  {Dorner}}, \bibinfo {author} {\bibfnamefont {S.~R.}\ \bibnamefont {Clark}},
  \bibinfo {author} {\bibfnamefont {L.}~\bibnamefont {Heaney}}, \bibinfo
  {author} {\bibfnamefont {Rosario}\ \bibnamefont {Fazio}}, \bibinfo {author}
  {\bibfnamefont {John}\ \bibnamefont {Goold}}, \ and\ \bibinfo {author}
  {\bibfnamefont {Vlatko}\ \bibnamefont {Vedral}},\ }\emph {\enquote {\bibinfo
  {title} {{Extracting Quantum Work Statistics and Fluctuation Theorems by
  Single-Qubit Interferometry}},}\ }\href {\doibase
  10.1103/PhysRevLett.110.230601} {\bibfield  {journal} {\bibinfo  {journal}
  {Phys. Rev. Lett.}\ }\textbf {\bibinfo {volume} {110}},\ \bibinfo {pages}
  {230601} (\bibinfo {year} {2013})},\ \Eprint
  {http://arxiv.org/abs/arXiv:1301.7021} {arXiv:1301.7021}\BibitemShut
  {NoStop}%
\bibitem [{\citenamefont {Mazzola}\ \emph {et~al.}(2013)\citenamefont
  {Mazzola}, \citenamefont {De~Chiara},\ and\ \citenamefont
  {Paternostro}}]{MazzolaDeChiaraPaternostro2013}%
  \BibitemOpen
  \bibfield  {author} {\bibinfo {author} {\bibfnamefont {Laura}\ \bibnamefont
  {Mazzola}}, \bibinfo {author} {\bibfnamefont {Gabriele}\ \bibnamefont
  {De~Chiara}}, \ and\ \bibinfo {author} {\bibfnamefont {Mauro}\ \bibnamefont
  {Paternostro}},\ }\emph {\enquote {\bibinfo {title} {{Measuring the
  Characteristic Function of the Work Distribution}},}\ }\href {\doibase
  10.1103/PhysRevLett.110.230602} {\bibfield  {journal} {\bibinfo  {journal}
  {Phys. Rev. Lett.}\ }\textbf {\bibinfo {volume} {110}},\ \bibinfo {pages}
  {230602} (\bibinfo {year} {2013})},\ \Eprint
  {http://arxiv.org/abs/arXiv:1301.7030} {arXiv:1301.7030}\BibitemShut
  {NoStop}%
\bibitem [{\citenamefont {Fusco}\ \emph {et~al.}(2014)\citenamefont {Fusco},
  \citenamefont {Pigeon}, \citenamefont {Apollaro}, \citenamefont {Xuereb},
  \citenamefont {Mazzola}, \citenamefont {Campisi}, \citenamefont {Ferraro},
  \citenamefont {Paternostro},\ and\ \citenamefont
  {De~Chiara}}]{FuscoEtAl2014}%
  \BibitemOpen
  \bibfield  {author} {\bibinfo {author} {\bibfnamefont {Lorenzo}\ \bibnamefont
  {Fusco}}, \bibinfo {author} {\bibfnamefont {Simon}\ \bibnamefont {Pigeon}},
  \bibinfo {author} {\bibfnamefont {Tony J.~G.}\ \bibnamefont {Apollaro}},
  \bibinfo {author} {\bibfnamefont {Andr{\'e}}\ \bibnamefont {Xuereb}},
  \bibinfo {author} {\bibfnamefont {Laura}\ \bibnamefont {Mazzola}}, \bibinfo
  {author} {\bibfnamefont {Michele}\ \bibnamefont {Campisi}}, \bibinfo {author}
  {\bibfnamefont {Alessandro}\ \bibnamefont {Ferraro}}, \bibinfo {author}
  {\bibfnamefont {Mauro}\ \bibnamefont {Paternostro}}, \ and\ \bibinfo {author}
  {\bibfnamefont {Gabriele}\ \bibnamefont {De~Chiara}},\ }\emph {\enquote
  {\bibinfo {title} {{Assessing the Nonequilibrium Thermodynamics in a Quenched
  Quantum Many-Body System via Single Projective Measurements}},}\ }\href
  {\doibase 10.1103/PhysRevX.4.031029} {\bibfield  {journal} {\bibinfo
  {journal} {Phys. Rev. X}\ }\textbf {\bibinfo {volume} {4}},\ \bibinfo {pages}
  {031029} (\bibinfo {year} {2014})},\ \Eprint
  {http://arxiv.org/abs/arXiv:1404.3150} {arXiv:1404.3150}\BibitemShut
  {NoStop}%
\bibitem [{\citenamefont {Roncaglia}\ \emph {et~al.}(2014)\citenamefont
  {Roncaglia}, \citenamefont {Cerisola},\ and\ \citenamefont
  {Paz}}]{RoncagliaCerisolaPaz2014}%
  \BibitemOpen
  \bibfield  {author} {\bibinfo {author} {\bibfnamefont {Augusto~J.}\
  \bibnamefont {Roncaglia}}, \bibinfo {author} {\bibfnamefont {Federico}\
  \bibnamefont {Cerisola}}, \ and\ \bibinfo {author} {\bibfnamefont
  {Juan~Pablo}\ \bibnamefont {Paz}},\ }\emph {\enquote {\bibinfo {title} {{Work
  Measurement as a Generalized Quantum Measurement}},}\ }\href {\doibase
  10.1103/PhysRevLett.113.250601} {\bibfield  {journal} {\bibinfo  {journal}
  {Phys. Rev. Lett.}\ }\textbf {\bibinfo {volume} {113}},\ \bibinfo {pages}
  {250601} (\bibinfo {year} {2014})},\ \Eprint
  {http://arxiv.org/abs/arXiv:1409.3812} {arXiv:1409.3812}\BibitemShut
  {NoStop}%
\bibitem [{\citenamefont {Talkner}\ \emph {et~al.}(2007)\citenamefont
  {Talkner}, \citenamefont {Lutz},\ and\ \citenamefont
  {H{\"a}nggi}}]{TalknerLutzHaenggi2007}%
  \BibitemOpen
  \bibfield  {author} {\bibinfo {author} {\bibfnamefont {Peter}\ \bibnamefont
  {Talkner}}, \bibinfo {author} {\bibfnamefont {Eric}\ \bibnamefont {Lutz}}, \
  and\ \bibinfo {author} {\bibfnamefont {Peter}\ \bibnamefont {H{\"a}nggi}},\
  }\emph {\enquote {\bibinfo {title} {{Fluctuation theorems: Work is not an
  observable}},}\ }\href {\doibase 10.1103/PhysRevE.75.050102} {\bibfield
  {journal} {\bibinfo  {journal} {Phys. Rev. E}\ }\textbf {\bibinfo {volume}
  {75}},\ \bibinfo {pages} {050102(R)} (\bibinfo {year} {2007})},\ \Eprint
  {http://arxiv.org/abs/arXiv:cond-mat/0703189}
  {arXiv:cond-mat/0703189}\BibitemShut {NoStop}%
\bibitem [{\citenamefont {Campisi}\ \emph {et~al.}(2011)\citenamefont
  {Campisi}, \citenamefont {H{\"a}nggi},\ and\ \citenamefont
  {Talkner}}]{CampisiHaenggiTalkner2011}%
  \BibitemOpen
  \bibfield  {author} {\bibinfo {author} {\bibfnamefont {Michele}\ \bibnamefont
  {Campisi}}, \bibinfo {author} {\bibfnamefont {Peter}\ \bibnamefont
  {H{\"a}nggi}}, \ and\ \bibinfo {author} {\bibfnamefont {Peter}\ \bibnamefont
  {Talkner}},\ }\emph {\enquote {\bibinfo {title} {{Colloquium. Quantum
  Fluctuation Relations: Foundations and Applications}},}\ }\href {\doibase
  10.1103/RevModPhys.83.771} {\bibfield  {journal} {\bibinfo  {journal} {Rev.
  Mod. Phys.}\ }\textbf {\bibinfo {volume} {83}},\ \bibinfo {pages} {771}
  (\bibinfo {year} {2011})},\ \Eprint {http://arxiv.org/abs/arXiv:1012.2268}
  {arXiv:1012.2268}\BibitemShut {NoStop}%
\bibitem [{\citenamefont {Kawai}\ \emph {et~al.}(2007)\citenamefont {Kawai},
  \citenamefont {Parrondo},\ and\ \citenamefont {{V}an~den
  Broeck}}]{KawaiParrondoVanDenBroeck2007}%
  \BibitemOpen
  \bibfield  {author} {\bibinfo {author} {\bibfnamefont {Ryoichi}\ \bibnamefont
  {Kawai}}, \bibinfo {author} {\bibfnamefont {Juan M.~R.}\ \bibnamefont
  {Parrondo}}, \ and\ \bibinfo {author} {\bibfnamefont {Christian}\
  \bibnamefont {{V}an~den Broeck}},\ }\emph {\enquote {\bibinfo {title}
  {{Dissipation: The Phase-Space Perspective}},}\ }\href {\doibase
  10.1103/PhysRevLett.98.080602} {\bibfield  {journal} {\bibinfo  {journal}
  {Phys. Rev. Lett.}\ }\textbf {\bibinfo {volume} {98}},\ \bibinfo {pages}
  {080602} (\bibinfo {year} {2007})},\ \Eprint
  {http://arxiv.org/abs/arXiv:cond-mat/0701397}
  {arXiv:cond-mat/0701397}\BibitemShut {NoStop}%
\bibitem [{\citenamefont {Parrondo}\ \emph {et~al.}(2009)\citenamefont
  {Parrondo}, \citenamefont {{V}an~den Broeck},\ and\ \citenamefont
  {Kawai}}]{ParrondoVanDenBroeckKawai2009}%
  \BibitemOpen
  \bibfield  {author} {\bibinfo {author} {\bibfnamefont {Juan M.~R.}\
  \bibnamefont {Parrondo}}, \bibinfo {author} {\bibfnamefont {Christian}\
  \bibnamefont {{V}an~den Broeck}}, \ and\ \bibinfo {author} {\bibfnamefont
  {Ryoichi}\ \bibnamefont {Kawai}},\ }\emph {\enquote {\bibinfo {title}
  {Entropy production and the arrow of time},}\ }\href {\doibase
  10.1088/1367-2630/11/7/073008} {\bibfield  {journal} {\bibinfo  {journal}
  {New J. Phys.}\ }\textbf {\bibinfo {volume} {11}},\ \bibinfo {pages} {073008}
  (\bibinfo {year} {2009})},\ \Eprint {http://arxiv.org/abs/arXiv:0904.1573}
  {arXiv:0904.1573}\BibitemShut {NoStop}%
\bibitem [{\citenamefont {Gour}\ \emph {et~al.}(2015)\citenamefont {Gour},
  \citenamefont {M{\"u}ller}, \citenamefont {Narasimhachar}, \citenamefont
  {Spekkens},\ and\ \citenamefont
  {Yunger~Halpern}}]{GourMuellerNarasimhacharSpekkensHalpern2015}%
  \BibitemOpen
  \bibfield  {author} {\bibinfo {author} {\bibfnamefont {Gilad}\ \bibnamefont
  {Gour}}, \bibinfo {author} {\bibfnamefont {Markus~P.}\ \bibnamefont
  {M{\"u}ller}}, \bibinfo {author} {\bibfnamefont {Varun}\ \bibnamefont
  {Narasimhachar}}, \bibinfo {author} {\bibfnamefont {Robert~W.}\ \bibnamefont
  {Spekkens}}, \ and\ \bibinfo {author} {\bibfnamefont {Nicole}\ \bibnamefont
  {Yunger~Halpern}},\ }\emph {\enquote {\bibinfo {title} {The resource theory
  of informational nonequilibrium in thermodynamics},}\ }\href {\doibase
  10.1016/j.physrep.2015.04.003} {\bibfield  {journal} {\bibinfo  {journal}
  {Phys. Rep.}\ }\textbf {\bibinfo {volume} {583}},\ \bibinfo {pages} {1}
  (\bibinfo {year} {2015})},\ \Eprint {http://arxiv.org/abs/arXiv:1309.6586}
  {arXiv:1309.6586}\BibitemShut {NoStop}%
\bibitem [{\citenamefont {Vinjanampathy}\ and\ \citenamefont
  {Anders}(2016)}]{VinjanampathyAnders2016}%
  \BibitemOpen
  \bibfield  {author} {\bibinfo {author} {\bibfnamefont {Sai}\ \bibnamefont
  {Vinjanampathy}}\ and\ \bibinfo {author} {\bibfnamefont {Janet}\ \bibnamefont
  {Anders}},\ }\emph {\enquote {\bibinfo {title} {{Quantum Thermodynamics}},}\
  }\href {\doibase 10.1080/00107514.2016.1201896} {\bibfield  {journal}
  {\bibinfo  {journal} {Contemp. Phys.}\ }\textbf {\bibinfo {volume} {57}},\
  \bibinfo {pages} {1} (\bibinfo {year} {2016})},\ \Eprint
  {http://arxiv.org/abs/arXiv:1508.06099} {arXiv:1508.06099}\BibitemShut
  {NoStop}%
\bibitem [{\citenamefont {Millen}\ and\ \citenamefont
  {Xuereb}(2016)}]{MillenXuereb2016}%
  \BibitemOpen
  \bibfield  {author} {\bibinfo {author} {\bibfnamefont {James}\ \bibnamefont
  {Millen}}\ and\ \bibinfo {author} {\bibfnamefont {Andr{\'e}}\ \bibnamefont
  {Xuereb}},\ }\emph {\enquote {\bibinfo {title} {Perspective on quantum
  thermodynamics},}\ }\href {\doibase 10.1088/1367-2630/18/1/011002} {\bibfield
   {journal} {\bibinfo  {journal} {New J. Phys.}\ }\textbf {\bibinfo {volume}
  {18}},\ \bibinfo {pages} {011002} (\bibinfo {year} {2016})},\ \Eprint
  {http://arxiv.org/abs/arXiv:1509.01086} {arXiv:1509.01086}\BibitemShut
  {NoStop}%
\bibitem [{\citenamefont {Goold}\ \emph {et~al.}(2016)\citenamefont {Goold},
  \citenamefont {Huber}, \citenamefont {Riera}, \citenamefont {del Rio},\ and\
  \citenamefont {Skrzypczyk}}]{GooldHuberRieraDelRioSkrzypczyk2016}%
  \BibitemOpen
  \bibfield  {author} {\bibinfo {author} {\bibfnamefont {John}\ \bibnamefont
  {Goold}}, \bibinfo {author} {\bibfnamefont {Marcus}\ \bibnamefont {Huber}},
  \bibinfo {author} {\bibfnamefont {Arnau}\ \bibnamefont {Riera}}, \bibinfo
  {author} {\bibfnamefont {L{\'i}dia}\ \bibnamefont {del Rio}}, \ and\ \bibinfo
  {author} {\bibfnamefont {Paul}\ \bibnamefont {Skrzypczyk}},\ }\emph {\enquote
  {\bibinfo {title} {The role of quantum information in thermodynamics
  \textemdash\ a topical review},}\ }\href {\doibase
  10.1088/1751-8113/49/14/143001} {\bibfield  {journal} {\bibinfo  {journal}
  {J. Phys. A: Math. Theor.}\ }\textbf {\bibinfo {volume} {49}},\ \bibinfo
  {pages} {143001} (\bibinfo {year} {2016})},\ \Eprint
  {http://arxiv.org/abs/arXiv:1505.07835} {arXiv:1505.07835}\BibitemShut
  {NoStop}%
\bibitem [{\citenamefont {Horodecki}\ and\ \citenamefont
  {Oppenheim}(2013)}]{HorodeckiOppenheim2013b}%
  \BibitemOpen
  \bibfield  {author} {\bibinfo {author} {\bibfnamefont {Micha{\l}}\
  \bibnamefont {Horodecki}}\ and\ \bibinfo {author} {\bibfnamefont {Jonathan}\
  \bibnamefont {Oppenheim}},\ }\emph {\enquote {\bibinfo {title} {Fundamental
  limitations for quantum and nanoscale thermodynamics},}\ }\href {\doibase
  10.1038/ncomms3059} {\bibfield  {journal} {\bibinfo  {journal} {Nat.
  Commun.}\ }\textbf {\bibinfo {volume} {4}},\ \bibinfo {pages} {2059}
  (\bibinfo {year} {2013})},\ \Eprint {http://arxiv.org/abs/arXiv:1111.3834}
  {arXiv:1111.3834}\BibitemShut {NoStop}%
\bibitem [{\citenamefont {Skrzypczyk}\ \emph {et~al.}(2014)\citenamefont
  {Skrzypczyk}, \citenamefont {Short},\ and\ \citenamefont
  {Popescu}}]{SkrzypczykShortPopescu2014}%
  \BibitemOpen
  \bibfield  {author} {\bibinfo {author} {\bibfnamefont {Paul}\ \bibnamefont
  {Skrzypczyk}}, \bibinfo {author} {\bibfnamefont {Anthony~J.}\ \bibnamefont
  {Short}}, \ and\ \bibinfo {author} {\bibfnamefont {Sandu}\ \bibnamefont
  {Popescu}},\ }\emph {\enquote {\bibinfo {title} {Work extraction and
  thermodynamics for individual quantum systems},}\ }\href {\doibase
  10.1038/ncomms5185} {\bibfield  {journal} {\bibinfo  {journal} {Nat.
  Commun.}\ }\textbf {\bibinfo {volume} {5}},\ \bibinfo {pages} {4185}
  (\bibinfo {year} {2014})},\ \Eprint {http://arxiv.org/abs/arXiv:1307.1558}
  {arXiv:1307.1558}\BibitemShut {NoStop}%
\bibitem [{\citenamefont {Faist}\ \emph {et~al.}(2015)\citenamefont {Faist},
  \citenamefont {Dupuis}, \citenamefont {Oppenheim},\ and\ \citenamefont
  {Renner}}]{FaistDupuisOppenheimRenner2015}%
  \BibitemOpen
  \bibfield  {author} {\bibinfo {author} {\bibfnamefont {Philippe}\
  \bibnamefont {Faist}}, \bibinfo {author} {\bibfnamefont {Fr{\'e}d{\'e}ric}\
  \bibnamefont {Dupuis}}, \bibinfo {author} {\bibfnamefont {Jonathan}\
  \bibnamefont {Oppenheim}}, \ and\ \bibinfo {author} {\bibfnamefont {Renato}\
  \bibnamefont {Renner}},\ }\emph {\enquote {\bibinfo {title} {{The minimal
  work cost of information processing}},}\ }\href {\doibase 10.1038/ncomms8669}
  {\bibfield  {journal} {\bibinfo  {journal} {Nat. Commun.}\ }\textbf {\bibinfo
  {volume} {6}},\ \bibinfo {pages} {7669} (\bibinfo {year} {2015})},\ \Eprint
  {http://arxiv.org/abs/arXiv:1211.1037} {arXiv:1211.1037}\BibitemShut
  {NoStop}%
\bibitem [{\citenamefont {Wilming}\ \emph {et~al.}(2016)\citenamefont
  {Wilming}, \citenamefont {Gallego},\ and\ \citenamefont
  {Eisert}}]{WilmingGallegoEisert2016}%
  \BibitemOpen
  \bibfield  {author} {\bibinfo {author} {\bibfnamefont {Henrik}\ \bibnamefont
  {Wilming}}, \bibinfo {author} {\bibfnamefont {Rodrigo}\ \bibnamefont
  {Gallego}}, \ and\ \bibinfo {author} {\bibfnamefont {Jens}\ \bibnamefont
  {Eisert}},\ }\emph {\enquote {\bibinfo {title} {Second law of thermodynamics
  under control restrictions},}\ }\href {\doibase 10.1103/PhysRevE.93.042126}
  {\bibfield  {journal} {\bibinfo  {journal} {Phys. Rev. E}\ }\textbf {\bibinfo
  {volume} {93}},\ \bibinfo {pages} {042126} (\bibinfo {year} {2016})},\
  \Eprint {http://arxiv.org/abs/arXiv:1411.3754} {arXiv:1411.3754}\BibitemShut
  {NoStop}%
\bibitem [{\citenamefont {Faist}\ and\ \citenamefont
  {Renner}(2018)}]{FaistRenner2018}%
  \BibitemOpen
  \bibfield  {author} {\bibinfo {author} {\bibfnamefont {Philippe}\
  \bibnamefont {Faist}}\ and\ \bibinfo {author} {\bibfnamefont {Renato}\
  \bibnamefont {Renner}},\ }\emph {\enquote {\bibinfo {title} {{Fundamental
  Work Cost of Quantum Processes}},}\ }\href {\doibase
  10.1103/PhysRevX.8.021011} {\bibfield  {journal} {\bibinfo  {journal} {Phys.
  Rev. X}\ }\textbf {\bibinfo {volume} {8}},\ \bibinfo {pages} {021011}
  (\bibinfo {year} {2018})},\ \Eprint {http://arxiv.org/abs/arXiv:1709.00506}
  {arXiv:1709.00506}\BibitemShut {NoStop}%
\bibitem [{\citenamefont {Clivaz}\ \emph
  {et~al.}(2019{\natexlab{a}})\citenamefont {Clivaz}, \citenamefont {Silva},
  \citenamefont {Haack}, \citenamefont {Bohr~Brask}, \citenamefont {Brunner},\
  and\ \citenamefont {Huber}}]{ClivazSilvaHaackBohrBraskBrunnerHuber2019a}%
  \BibitemOpen
  \bibfield  {author} {\bibinfo {author} {\bibfnamefont {Fabien}\ \bibnamefont
  {Clivaz}}, \bibinfo {author} {\bibfnamefont {Ralph}\ \bibnamefont {Silva}},
  \bibinfo {author} {\bibfnamefont {G{\'e}raldine}\ \bibnamefont {Haack}},
  \bibinfo {author} {\bibfnamefont {Jonatan}\ \bibnamefont {Bohr~Brask}},
  \bibinfo {author} {\bibfnamefont {Nicolas}\ \bibnamefont {Brunner}}, \ and\
  \bibinfo {author} {\bibfnamefont {Marcus}\ \bibnamefont {Huber}},\ }\emph
  {\enquote {\bibinfo {title} {Unifying paradigms of quantum refrigeration:
  fundamental limits of cooling and associated work costs},}\ }\href {\doibase
  10.1103/PhysRevE.100.042130} {\bibfield  {journal} {\bibinfo  {journal}
  {Phys. Rev. E}\ }\textbf {\bibinfo {volume} {100}},\ \bibinfo {pages}
  {042130} (\bibinfo {year} {2019}{\natexlab{a}})},\ \Eprint
  {http://arxiv.org/abs/arXiv:1710.11624} {arXiv:1710.11624}\BibitemShut
  {NoStop}%
\bibitem [{\citenamefont {Clivaz}\ \emph
  {et~al.}(2019{\natexlab{b}})\citenamefont {Clivaz}, \citenamefont {Silva},
  \citenamefont {Haack}, \citenamefont {Bohr~Brask}, \citenamefont {Brunner},\
  and\ \citenamefont {Huber}}]{ClivazSilvaHaackBohrBraskBrunnerHuber2019b}%
  \BibitemOpen
  \bibfield  {author} {\bibinfo {author} {\bibfnamefont {Fabien}\ \bibnamefont
  {Clivaz}}, \bibinfo {author} {\bibfnamefont {Ralph}\ \bibnamefont {Silva}},
  \bibinfo {author} {\bibfnamefont {G{\'e}raldine}\ \bibnamefont {Haack}},
  \bibinfo {author} {\bibfnamefont {Jonatan}\ \bibnamefont {Bohr~Brask}},
  \bibinfo {author} {\bibfnamefont {Nicolas}\ \bibnamefont {Brunner}}, \ and\
  \bibinfo {author} {\bibfnamefont {Marcus}\ \bibnamefont {Huber}},\ }\emph
  {\enquote {\bibinfo {title} {{Unifying Paradigms of Quantum Refrigeration: A
  Universal and Attainable Bound on Cooling}},}\ }\href {\doibase
  10.1103/PhysRevLett.123.170605} {\bibfield  {journal} {\bibinfo  {journal}
  {Phys. Rev. Lett.}\ }\textbf {\bibinfo {volume} {123}},\ \bibinfo {pages}
  {170605} (\bibinfo {year} {2019}{\natexlab{b}})},\ \Eprint
  {http://arxiv.org/abs/arXiv:1903.04970} {arXiv:1903.04970}\BibitemShut
  {NoStop}%
\bibitem [{\citenamefont {Huber}\ \emph {et~al.}(2015)\citenamefont {Huber},
  \citenamefont {Perarnau-Llobet}, \citenamefont {Hovhannisyan}, \citenamefont
  {Skrzypczyk}, \citenamefont {Kl{\"o}ckl}, \citenamefont {Brunner},\ and\
  \citenamefont
  {Ac$\acute{\i}$n}}]{HuberPerarnauHovhannisyanSkrzypczykKloecklBrunnerAcin2015}%
  \BibitemOpen
  \bibfield  {author} {\bibinfo {author} {\bibfnamefont {Marcus}\ \bibnamefont
  {Huber}}, \bibinfo {author} {\bibfnamefont {Mart{\'i}}\ \bibnamefont
  {Perarnau-Llobet}}, \bibinfo {author} {\bibfnamefont {Karen~V.}\ \bibnamefont
  {Hovhannisyan}}, \bibinfo {author} {\bibfnamefont {Paul}\ \bibnamefont
  {Skrzypczyk}}, \bibinfo {author} {\bibfnamefont {Claude}\ \bibnamefont
  {Kl{\"o}ckl}}, \bibinfo {author} {\bibfnamefont {Nicolas}\ \bibnamefont
  {Brunner}}, \ and\ \bibinfo {author} {\bibfnamefont {Antonio}\ \bibnamefont
  {Ac$\acute{\i}$n}},\ }\emph {\enquote {\bibinfo {title} {Thermodynamic cost
  of creating correlations},}\ }\href {\doibase 10.1088/1367-2630/17/6/065008}
  {\bibfield  {journal} {\bibinfo  {journal} {New J. Phys.}\ }\textbf {\bibinfo
  {volume} {17}},\ \bibinfo {pages} {065008} (\bibinfo {year} {2015})},\
  \Eprint {http://arxiv.org/abs/arXiv:1404.2169} {arXiv:1404.2169}\BibitemShut
  {NoStop}%
\bibitem [{\citenamefont {Bruschi}\ \emph {et~al.}(2015)\citenamefont
  {Bruschi}, \citenamefont {Perarnau-Llobet}, \citenamefont {Friis},
  \citenamefont {Hovhannisyan},\ and\ \citenamefont
  {Huber}}]{BruschiPerarnauLlobetFriisHovhannisyanHuber2015}%
  \BibitemOpen
  \bibfield  {author} {\bibinfo {author} {\bibfnamefont {David~E.}\
  \bibnamefont {Bruschi}}, \bibinfo {author} {\bibfnamefont {Mart{\'i}}\
  \bibnamefont {Perarnau-Llobet}}, \bibinfo {author} {\bibfnamefont {Nicolai}\
  \bibnamefont {Friis}}, \bibinfo {author} {\bibfnamefont {Karen~V.}\
  \bibnamefont {Hovhannisyan}}, \ and\ \bibinfo {author} {\bibfnamefont
  {Marcus}\ \bibnamefont {Huber}},\ }\emph {\enquote {\bibinfo {title} {The
  thermodynamics of creating correlations: Limitations and optimal
  protocols},}\ }\href {\doibase 10.1103/PhysRevE.91.032118} {\bibfield
  {journal} {\bibinfo  {journal} {Phys. Rev. E}\ }\textbf {\bibinfo {volume}
  {91}},\ \bibinfo {pages} {032118} (\bibinfo {year} {2015})},\ \Eprint
  {http://arxiv.org/abs/arXiv:1409.4647} {arXiv:1409.4647}\BibitemShut
  {NoStop}%
\bibitem [{\citenamefont {Friis}\ \emph {et~al.}(2016)\citenamefont {Friis},
  \citenamefont {Huber},\ and\ \citenamefont
  {Perarnau-Llobet}}]{FriisHuberPerarnauLlobet2016}%
  \BibitemOpen
  \bibfield  {author} {\bibinfo {author} {\bibfnamefont {Nicolai}\ \bibnamefont
  {Friis}}, \bibinfo {author} {\bibfnamefont {Marcus}\ \bibnamefont {Huber}}, \
  and\ \bibinfo {author} {\bibfnamefont {Mart{\'i}}\ \bibnamefont
  {Perarnau-Llobet}},\ }\emph {\enquote {\bibinfo {title} {Energetics of
  correlations in interacting systems},}\ }\href {\doibase
  10.1103/PhysRevE.93.042135} {\bibfield  {journal} {\bibinfo  {journal} {Phys.
  Rev. E}\ }\textbf {\bibinfo {volume} {93}},\ \bibinfo {pages} {042135}
  (\bibinfo {year} {2016})},\ \Eprint {http://arxiv.org/abs/arXiv:1511.08654}
  {arXiv:1511.08654}\BibitemShut {NoStop}%
\bibitem [{\citenamefont {Vitagliano}\ \emph {et~al.}(2019)\citenamefont
  {Vitagliano}, \citenamefont {Kl{\"o}ckl}, \citenamefont {Huber},\ and\
  \citenamefont {Friis}}]{VitaglianoKloecklHuberFriis2019}%
  \BibitemOpen
  \bibfield  {author} {\bibinfo {author} {\bibfnamefont {Giuseppe}\
  \bibnamefont {Vitagliano}}, \bibinfo {author} {\bibfnamefont {Claude}\
  \bibnamefont {Kl{\"o}ckl}}, \bibinfo {author} {\bibfnamefont {Marcus}\
  \bibnamefont {Huber}}, \ and\ \bibinfo {author} {\bibfnamefont {Nicolai}\
  \bibnamefont {Friis}},\ }in\ \href {\doibase 10.1007/978-3-319-99046-0_30}
  {\emph {\bibinfo {booktitle} {Thermodynamics in the Quantum Regime}}},\
  \bibinfo {editor} {edited by\ \bibinfo {editor} {\bibfnamefont {Felix}\
  \bibnamefont {Binder}}, \bibinfo {editor} {\bibfnamefont {Luis~A.}\
  \bibnamefont {Correa}}, \bibinfo {editor} {\bibfnamefont {Christian}\
  \bibnamefont {Gogolin}}, \bibinfo {editor} {\bibfnamefont {Janet}\
  \bibnamefont {Anders}}, \ and\ \bibinfo {editor} {\bibfnamefont {Gerardo}\
  \bibnamefont {Adesso}}}\ (\bibinfo  {publisher} {Springer},\ \bibinfo {year}
  {2019})\ Chap.~\bibinfo {chapter} {30}, pp.\ \bibinfo {pages} {731--750},\
  \Eprint {http://arxiv.org/abs/arXiv:1803.06884}
  {arXiv:1803.06884}\BibitemShut {NoStop}%
\bibitem [{\citenamefont {Schulman}\ \emph {et~al.}(2005)\citenamefont
  {Schulman}, \citenamefont {Mor},\ and\ \citenamefont
  {Weinstein}}]{SchulmanMorWeinstein2005}%
  \BibitemOpen
  \bibfield  {author} {\bibinfo {author} {\bibfnamefont {Leonard~J.}\
  \bibnamefont {Schulman}}, \bibinfo {author} {\bibfnamefont {Tal}\
  \bibnamefont {Mor}}, \ and\ \bibinfo {author} {\bibfnamefont {Yossi}\
  \bibnamefont {Weinstein}},\ }\emph {\enquote {\bibinfo {title} {{Physical
  Limits of Heat-Bath Algorithmic Cooling}},}\ }\href {\doibase
  10.1103/PhysRevLett.94.120501} {\bibfield  {journal} {\bibinfo  {journal}
  {Phys. Rev. Lett.}\ }\textbf {\bibinfo {volume} {94}},\ \bibinfo {pages}
  {120501} (\bibinfo {year} {2005})}\BibitemShut {NoStop}%
\bibitem [{\citenamefont {Silva}\ \emph {et~al.}(2016)\citenamefont {Silva},
  \citenamefont {Manzano}, \citenamefont {Skrzypczyk},\ and\ \citenamefont
  {Brunner}}]{SilvaManzanoSkrzypczyBrunner2016}%
  \BibitemOpen
  \bibfield  {author} {\bibinfo {author} {\bibfnamefont {Ralph}\ \bibnamefont
  {Silva}}, \bibinfo {author} {\bibfnamefont {Gonzalo}\ \bibnamefont
  {Manzano}}, \bibinfo {author} {\bibfnamefont {Paul}\ \bibnamefont
  {Skrzypczyk}}, \ and\ \bibinfo {author} {\bibfnamefont {Nicolas}\
  \bibnamefont {Brunner}},\ }\emph {\enquote {\bibinfo {title} {Performance of
  autonomous quantum thermal machines: {H}ilbert space dimension as a
  thermodynamical resource},}\ }\href {\doibase 10.1103/PhysRevE.94.032120}
  {\bibfield  {journal} {\bibinfo  {journal} {Phys. Rev. E}\ }\textbf {\bibinfo
  {volume} {94}},\ \bibinfo {pages} {032120} (\bibinfo {year} {2016})},\
  \Eprint {http://arxiv.org/abs/arXiv:1604.04098}
  {arXiv:1604.04098}\BibitemShut {NoStop}%
\bibitem [{\citenamefont {Wilming}\ and\ \citenamefont
  {Gallego}(2017)}]{WilmingGallego2017}%
  \BibitemOpen
  \bibfield  {author} {\bibinfo {author} {\bibfnamefont {Henrik}\ \bibnamefont
  {Wilming}}\ and\ \bibinfo {author} {\bibfnamefont {Rodrigo}\ \bibnamefont
  {Gallego}},\ }\emph {\enquote {\bibinfo {title} {{Third Law of Thermodynamics
  as a Single Inequality}},}\ }\href {\doibase 10.1103/PhysRevX.7.041033}
  {\bibfield  {journal} {\bibinfo  {journal} {Phys Rev. X}\ }\textbf {\bibinfo
  {volume} {7}},\ \bibinfo {pages} {041033} (\bibinfo {year} {2017})},\ \Eprint
  {http://arxiv.org/abs/arXiv:1701.07478} {arXiv:1701.07478}\BibitemShut
  {NoStop}%
\bibitem [{\citenamefont {Scharlau}\ and\ \citenamefont
  {M{\"u}ller}(2018)}]{ScharlauMueller2018}%
  \BibitemOpen
  \bibfield  {author} {\bibinfo {author} {\bibfnamefont {Jakob}\ \bibnamefont
  {Scharlau}}\ and\ \bibinfo {author} {\bibfnamefont {Markus~P.}\ \bibnamefont
  {M{\"u}ller}},\ }\emph {\enquote {\bibinfo {title} {{Quantum Horn's lemma,
  finite heat baths, and the third law of thermodynamics}},}\ }\href {\doibase
  10.22331/q-2018-02-22-54} {\bibfield  {journal} {\bibinfo  {journal}
  {Quantum}\ }\textbf {\bibinfo {volume} {2}},\ \bibinfo {pages} {54} (\bibinfo
  {year} {2018})},\ \Eprint {http://arxiv.org/abs/arXiv:1605.06092}
  {arXiv:1605.06092}\BibitemShut {NoStop}%
\bibitem [{\citenamefont {Elouard}\ \emph {et~al.}(2017)\citenamefont
  {Elouard}, \citenamefont {Herrera-Mart{\'i}}, \citenamefont {Huard},\ and\
  \citenamefont {Auff{\`e}ves}}]{ElouardHerreraMartiHuardAuffeves2017}%
  \BibitemOpen
  \bibfield  {author} {\bibinfo {author} {\bibfnamefont {Cyril}\ \bibnamefont
  {Elouard}}, \bibinfo {author} {\bibfnamefont {David}\ \bibnamefont
  {Herrera-Mart{\'i}}}, \bibinfo {author} {\bibfnamefont {Benjamin}\
  \bibnamefont {Huard}}, \ and\ \bibinfo {author} {\bibfnamefont {Alexia}\
  \bibnamefont {Auff{\`e}ves}},\ }\emph {\enquote {\bibinfo {title}
  {{Extracting work from quantum measurement in Maxwell demon engines}},}\
  }\href {\doibase 10.1103/PhysRevLett.118.260603} {\bibfield  {journal}
  {\bibinfo  {journal} {Phys. Rev. Lett.}\ }\textbf {\bibinfo {volume} {118}},\
  \bibinfo {pages} {260603} (\bibinfo {year} {2017})},\ \Eprint
  {http://arxiv.org/abs/arXiv:1702.01917} {arXiv:1702.01917}\BibitemShut
  {NoStop}%
\bibitem [{\citenamefont {Elouard}\ and\ \citenamefont
  {Jordan}(2018)}]{ElouardJordan2018}%
  \BibitemOpen
  \bibfield  {author} {\bibinfo {author} {\bibfnamefont {Cyril}\ \bibnamefont
  {Elouard}}\ and\ \bibinfo {author} {\bibfnamefont {Andrew~N.}\ \bibnamefont
  {Jordan}},\ }\emph {\enquote {\bibinfo {title} {{Efficient Quantum
  Measurement Engine}},}\ }\href {\doibase 10.1103/PhysRevLett.120.260601}
  {\bibfield  {journal} {\bibinfo  {journal} {Phys. Rev. Lett.}\ }\textbf
  {\bibinfo {volume} {120}},\ \bibinfo {pages} {260601} (\bibinfo {year}
  {2018})},\ \Eprint {http://arxiv.org/abs/arXiv:1801.03979}
  {arXiv:1801.03979}\BibitemShut {NoStop}%
\bibitem [{\citenamefont {Buffoni}\ \emph {et~al.}(2019)\citenamefont
  {Buffoni}, \citenamefont {Solfanelli}, \citenamefont {Verrucchi},
  \citenamefont {Cuccoli},\ and\ \citenamefont {Campisi}}]{BuffoniEtAl2019}%
  \BibitemOpen
  \bibfield  {author} {\bibinfo {author} {\bibfnamefont {Lorenzo}\ \bibnamefont
  {Buffoni}}, \bibinfo {author} {\bibfnamefont {Andrea}\ \bibnamefont
  {Solfanelli}}, \bibinfo {author} {\bibfnamefont {Paola}\ \bibnamefont
  {Verrucchi}}, \bibinfo {author} {\bibfnamefont {Alessandro}\ \bibnamefont
  {Cuccoli}}, \ and\ \bibinfo {author} {\bibfnamefont {Michele}\ \bibnamefont
  {Campisi}},\ }\emph {\enquote {\bibinfo {title} {{Quantum Measurement
  Cooling}},}\ }\href {\doibase 10.1103/PhysRevLett.122.070603} {\bibfield
  {journal} {\bibinfo  {journal} {Phys. Rev. Lett.}\ }\textbf {\bibinfo
  {volume} {122}},\ \bibinfo {pages} {070603} (\bibinfo {year} {2019})},\
  \Eprint {http://arxiv.org/abs/arXiv:1806.07814}
  {arXiv:1806.07814}\BibitemShut {NoStop}%
\bibitem [{\citenamefont {Pusz}\ and\ \citenamefont
  {Woronowicz}(1978)}]{PuszWoronowicz1978}%
  \BibitemOpen
  \bibfield  {author} {\bibinfo {author} {\bibfnamefont {Wies{\l}aw}\
  \bibnamefont {Pusz}}\ and\ \bibinfo {author} {\bibfnamefont
  {Stanis{\l}aw~L.}\ \bibnamefont {Woronowicz}},\ }\emph {\enquote {\bibinfo
  {title} {{Passive states and KMS states for general quantum systems}},}\
  }\href {\doibase 10.1007/BF01614224} {\bibfield  {journal} {\bibinfo
  {journal} {Comm. Math. Phys.}\ }\textbf {\bibinfo {volume} {58}},\ \bibinfo
  {pages} {273} (\bibinfo {year} {1978})},\ \bibinfo {note}
  {\href{https://projecteuclid.org/euclid.cmp/1103901491}{https://projecteuclid.org/euclid.cmp/1103901491}}\BibitemShut
  {NoStop}%
\bibitem [{\citenamefont {Rodr{\'i}guez-Briones}\ \emph
  {et~al.}(2017)\citenamefont {Rodr{\'i}guez-Briones}, \citenamefont
  {Mart{\'i}n-Mart{\'i}nez}, \citenamefont {Kempf},\ and\ \citenamefont
  {Laflamme}}]{RodriguezBrionesMartinMartinezKempfLaflamme2017}%
  \BibitemOpen
  \bibfield  {author} {\bibinfo {author} {\bibfnamefont {Nayeli~A.}\
  \bibnamefont {Rodr{\'i}guez-Briones}}, \bibinfo {author} {\bibfnamefont
  {Eduardo}\ \bibnamefont {Mart{\'i}n-Mart{\'i}nez}}, \bibinfo {author}
  {\bibfnamefont {Achim}\ \bibnamefont {Kempf}}, \ and\ \bibinfo {author}
  {\bibfnamefont {Raymond}\ \bibnamefont {Laflamme}},\ }\emph {\enquote
  {\bibinfo {title} {{Correlation-Enhanced Algorithmic Cooling}},}\ }\href
  {\doibase 10.1103/PhysRevLett.119.050502} {\bibfield  {journal} {\bibinfo
  {journal} {Phys. Rev. Lett.}\ }\textbf {\bibinfo {volume} {119}},\ \bibinfo
  {pages} {050502} (\bibinfo {year} {2017})},\ \Eprint
  {http://arxiv.org/abs/arXiv:1703.03816} {arXiv:1703.03816}\BibitemShut
  {NoStop}%
\bibitem [{\citenamefont {Rodr{\'{\i}}guez-Briones}\ \emph
  {et~al.}(2017)\citenamefont {Rodr{\'{\i}}guez-Briones}, \citenamefont {Li},
  \citenamefont {Peng}, \citenamefont {Mor}, \citenamefont {Weinstein},\ and\
  \citenamefont {Laflamme}}]{RodriguezBrionesLiPengMorWeinsteinLaflamme2017}%
  \BibitemOpen
  \bibfield  {author} {\bibinfo {author} {\bibfnamefont {Nayeli~A.}\
  \bibnamefont {Rodr{\'{\i}}guez-Briones}}, \bibinfo {author} {\bibfnamefont
  {Jun}\ \bibnamefont {Li}}, \bibinfo {author} {\bibfnamefont {Xinhua}\
  \bibnamefont {Peng}}, \bibinfo {author} {\bibfnamefont {Tal}\ \bibnamefont
  {Mor}}, \bibinfo {author} {\bibfnamefont {Yossi}\ \bibnamefont {Weinstein}},
  \ and\ \bibinfo {author} {\bibfnamefont {Raymond}\ \bibnamefont {Laflamme}},\
  }\emph {\enquote {\bibinfo {title} {Heat-bath algorithmic cooling with
  correlated qubit-environment interactions},}\ }\href {\doibase
  10.1088/1367-2630/aa8fe0} {\bibfield  {journal} {\bibinfo  {journal} {New J.
  Phys.}\ }\textbf {\bibinfo {volume} {19}},\ \bibinfo {pages} {113047}
  (\bibinfo {year} {2017})},\ \Eprint {http://arxiv.org/abs/arXiv:1703.02999}
  {arXiv:1703.02999}\BibitemShut {NoStop}%
\bibitem [{\citenamefont {Alhambra}\ \emph {et~al.}(2019)\citenamefont
  {Alhambra}, \citenamefont {Lostaglio},\ and\ \citenamefont
  {Perry}}]{AlhambraLostaglioPerry2019}%
  \BibitemOpen
  \bibfield  {author} {\bibinfo {author} {\bibfnamefont {{\'A}lvaro~M.}\
  \bibnamefont {Alhambra}}, \bibinfo {author} {\bibfnamefont {Matteo}\
  \bibnamefont {Lostaglio}}, \ and\ \bibinfo {author} {\bibfnamefont
  {Christopher}\ \bibnamefont {Perry}},\ }\emph {\enquote {\bibinfo {title}
  {{Heat-Bath Algorithmic Cooling with optimal thermalization strategies}},}\
  }\href {\doibase 10.22331/q-2019-09-23-188} {\bibfield  {journal} {\bibinfo
  {journal} {Quantum}\ }\textbf {\bibinfo {volume} {3}},\ \bibinfo {pages}
  {188} (\bibinfo {year} {2019})},\ \Eprint
  {http://arxiv.org/abs/arXiv:1807.07974} {arXiv:1807.07974}\BibitemShut
  {NoStop}%
\bibitem [{\citenamefont {Silva}(2008)}]{Silva2008}%
  \BibitemOpen
  \bibfield  {author} {\bibinfo {author} {\bibfnamefont {Alessandro}\
  \bibnamefont {Silva}},\ }\emph {\enquote {\bibinfo {title} {{Statistics of
  the Work Done on a Quantum Critical System by Quenching a Control
  Parameter}},}\ }\href {\doibase 10.1103/PhysRevLett.101.120603} {\bibfield
  {journal} {\bibinfo  {journal} {Phys. Rev. Lett.}\ }\textbf {\bibinfo
  {volume} {101}},\ \bibinfo {pages} {120603} (\bibinfo {year} {2008})},\
  \Eprint {http://arxiv.org/abs/arXiv:0806.4301} {arXiv:0806.4301}\BibitemShut
  {NoStop}%
\bibitem [{\citenamefont {Dorner}\ \emph {et~al.}(2012)\citenamefont {Dorner},
  \citenamefont {Goold}, \citenamefont {Cormick}, \citenamefont {Paternostro},\
  and\ \citenamefont {Vedral}}]{DornerGooldCormickPaternostroVedral2012}%
  \BibitemOpen
  \bibfield  {author} {\bibinfo {author} {\bibfnamefont {Ross}\ \bibnamefont
  {Dorner}}, \bibinfo {author} {\bibfnamefont {John}\ \bibnamefont {Goold}},
  \bibinfo {author} {\bibfnamefont {Cecilia}\ \bibnamefont {Cormick}}, \bibinfo
  {author} {\bibfnamefont {Mauro}\ \bibnamefont {Paternostro}}, \ and\ \bibinfo
  {author} {\bibfnamefont {Vlatko}\ \bibnamefont {Vedral}},\ }\emph {\enquote
  {\bibinfo {title} {{Emergent Thermodynamics in a Quenched Quantum Many-Body
  System}},}\ }\href {\doibase 10.1103/PhysRevLett.109.160601} {\bibfield
  {journal} {\bibinfo  {journal} {Phys. Rev. Lett.}\ }\textbf {\bibinfo
  {volume} {109}},\ \bibinfo {pages} {160601} (\bibinfo {year} {2012})},\
  \Eprint {http://arxiv.org/abs/arXiv:1207.4777} {arXiv:1207.4777}\BibitemShut
  {NoStop}%
\bibitem [{\citenamefont {Campisi}\ \emph {et~al.}(2015)\citenamefont
  {Campisi}, \citenamefont {Pekola},\ and\ \citenamefont
  {Fazio}}]{CampisiPekolaFazio2015}%
  \BibitemOpen
  \bibfield  {author} {\bibinfo {author} {\bibfnamefont {Michele}\ \bibnamefont
  {Campisi}}, \bibinfo {author} {\bibfnamefont {Jukka}\ \bibnamefont {Pekola}},
  \ and\ \bibinfo {author} {\bibfnamefont {Rosario}\ \bibnamefont {Fazio}},\
  }\emph {\enquote {\bibinfo {title} {Nonequilibrium fluctuations in quantum
  heat engines: theory, example, and possible solid state experiments},}\
  }\href {\doibase 10.1088/1367-2630/17/3/035012} {\bibfield  {journal}
  {\bibinfo  {journal} {New J. Phys.}\ }\textbf {\bibinfo {volume} {17}},\
  \bibinfo {pages} {035012} (\bibinfo {year} {2015})},\ \Eprint
  {http://arxiv.org/abs/arXiv:1412.0898} {arXiv:1412.0898}\BibitemShut
  {NoStop}%
\bibitem [{\citenamefont {Friis}\ and\ \citenamefont
  {Huber}(2018)}]{FriisHuber2018}%
  \BibitemOpen
  \bibfield  {author} {\bibinfo {author} {\bibfnamefont {Nicolai}\ \bibnamefont
  {Friis}}\ and\ \bibinfo {author} {\bibfnamefont {Marcus}\ \bibnamefont
  {Huber}},\ }\emph {\enquote {\bibinfo {title} {Precision and {W}ork
  {F}luctuations in {G}aussian {B}attery {C}harging},}\ }\href {\doibase
  10.22331/q-2018-04-23-61} {\bibfield  {journal} {\bibinfo  {journal}
  {{Quantum}}\ }\textbf {\bibinfo {volume} {2}},\ \bibinfo {pages} {61}
  (\bibinfo {year} {2018})},\ \Eprint {http://arxiv.org/abs/arXiv:1708.00749}
  {arXiv:1708.00749}\BibitemShut {NoStop}%
\bibitem [{\citenamefont {Allahverdyan}(2014)}]{Allahverdyan2014}%
  \BibitemOpen
  \bibfield  {author} {\bibinfo {author} {\bibfnamefont {Armen~E.}\
  \bibnamefont {Allahverdyan}},\ }\emph {\enquote {\bibinfo {title}
  {Nonequilibrium quantum fluctuations of work},}\ }\href {\doibase
  10.1103/PhysRevE.90.032137} {\bibfield  {journal} {\bibinfo  {journal} {Phys.
  Rev. E}\ }\textbf {\bibinfo {volume} {90}},\ \bibinfo {pages} {032137}
  (\bibinfo {year} {2014})},\ \Eprint {http://arxiv.org/abs/arXiv:1404.4190}
  {arXiv:1404.4190}\BibitemShut {NoStop}%
\bibitem [{\citenamefont {Talkner}\ and\ \citenamefont
  {H{\"a}nggi}(2016)}]{TalknerHaenggi2016}%
  \BibitemOpen
  \bibfield  {author} {\bibinfo {author} {\bibfnamefont {Peter}\ \bibnamefont
  {Talkner}}\ and\ \bibinfo {author} {\bibfnamefont {Peter}\ \bibnamefont
  {H{\"a}nggi}},\ }\emph {\enquote {\bibinfo {title} {Aspects of quantum
  work},}\ }\href {\doibase 10.1103/PhysRevE.93.022131} {\bibfield  {journal}
  {\bibinfo  {journal} {Phys. Rev. E}\ }\textbf {\bibinfo {volume} {93}},\
  \bibinfo {pages} {022131} (\bibinfo {year} {2016})},\ \Eprint
  {http://arxiv.org/abs/arXiv:1512.02516} {arXiv:1512.02516}\BibitemShut
  {NoStop}%
\bibitem [{\citenamefont {Jarzynski}\ \emph {et~al.}(2015)\citenamefont
  {Jarzynski}, \citenamefont {Quan},\ and\ \citenamefont
  {Rahav}}]{JarzynskiQuanRahav2015}%
  \BibitemOpen
  \bibfield  {author} {\bibinfo {author} {\bibfnamefont {Christopher}\
  \bibnamefont {Jarzynski}}, \bibinfo {author} {\bibfnamefont {H.~T.}\
  \bibnamefont {Quan}}, \ and\ \bibinfo {author} {\bibfnamefont {Saar}\
  \bibnamefont {Rahav}},\ }\emph {\enquote {\bibinfo {title}
  {{Quantum-Classical Correspondence Principle for Work Distributions}},}\
  }\href {\doibase 10.1103/PhysRevX.5.031038} {\bibfield  {journal} {\bibinfo
  {journal} {Phys. Rev. X}\ }\textbf {\bibinfo {volume} {5}},\ \bibinfo {pages}
  {031038} (\bibinfo {year} {2015})},\ \Eprint
  {http://arxiv.org/abs/arXiv:1507.05763} {arXiv:1507.05763}\BibitemShut
  {NoStop}%
\bibitem [{\citenamefont {Perarnau-Llobet}\ \emph {et~al.}(2017)\citenamefont
  {Perarnau-Llobet}, \citenamefont {B{\"a}umer}, \citenamefont {Hovhannisyan},
  \citenamefont {Huber},\ and\ \citenamefont
  {Ac{\'i}n}}]{PerarnauLlobetBaeumerHovhannisyanHuberAcin2017}%
  \BibitemOpen
  \bibfield  {author} {\bibinfo {author} {\bibfnamefont {Mart{\'i}}\
  \bibnamefont {Perarnau-Llobet}}, \bibinfo {author} {\bibfnamefont {Elisa}\
  \bibnamefont {B{\"a}umer}}, \bibinfo {author} {\bibfnamefont {Karen~V.}\
  \bibnamefont {Hovhannisyan}}, \bibinfo {author} {\bibfnamefont {Marcus}\
  \bibnamefont {Huber}}, \ and\ \bibinfo {author} {\bibfnamefont {Antonio}\
  \bibnamefont {Ac{\'i}n}},\ }\emph {\enquote {\bibinfo {title} {{No-Go Theorem
  for the Characterization of Work Fluctuations in Coherent Quantum
  Systems}},}\ }\href {\doibase 10.1103/PhysRevLett.118.070601} {\bibfield
  {journal} {\bibinfo  {journal} {Phys. Rev. Lett.}\ }\textbf {\bibinfo
  {volume} {118}},\ \bibinfo {pages} {070601} (\bibinfo {year} {2017})},\
  \Eprint {http://arxiv.org/abs/arXiv:1606.08368}
  {arXiv:1606.08368}\BibitemShut {NoStop}%
\bibitem [{\citenamefont {Lostaglio}(2018)}]{Lostaglio2018}%
  \BibitemOpen
  \bibfield  {author} {\bibinfo {author} {\bibfnamefont {Matteo}\ \bibnamefont
  {Lostaglio}},\ }\emph {\enquote {\bibinfo {title} {{Quantum Fluctuation
  Theorems, Contextuality, and Work Quasiprobabilities}},}\ }\href {\doibase
  10.1103/PhysRevLett.120.040602} {\bibfield  {journal} {\bibinfo  {journal}
  {Phys. Rev. Lett.}\ }\textbf {\bibinfo {volume} {120}},\ \bibinfo {pages}
  {040602} (\bibinfo {year} {2018})},\ \Eprint
  {http://arxiv.org/abs/arXiv:1705.05397} {arXiv:1705.05397}\BibitemShut
  {NoStop}%
\bibitem [{\citenamefont {Jarzynski}(1997)}]{Jarzynski1997}%
  \BibitemOpen
  \bibfield  {author} {\bibinfo {author} {\bibfnamefont {Christopher}\
  \bibnamefont {Jarzynski}},\ }\emph {\enquote {\bibinfo {title}
  {{Nonequilibrium Equality for Free Energy Differences}},}\ }\href {\doibase
  10.1103/PhysRevLett.78.2690} {\bibfield  {journal} {\bibinfo  {journal}
  {Phys. Rev. Lett.}\ }\textbf {\bibinfo {volume} {78}},\ \bibinfo {pages}
  {2690} (\bibinfo {year} {1997})},\ \Eprint
  {http://arxiv.org/abs/arXiv:cond-mat/9610209}
  {arXiv:cond-mat/9610209}\BibitemShut {NoStop}%
\bibitem [{\citenamefont {Crooks}(1999)}]{Crooks1999}%
  \BibitemOpen
  \bibfield  {author} {\bibinfo {author} {\bibfnamefont {Gavin~E.}\
  \bibnamefont {Crooks}},\ }\emph {\enquote {\bibinfo {title} {{The Entropy
  Production Fluctuation Theorem and the Nonequilibrium Work Relation for Free
  Energy Differences}},}\ }\href {\doibase 10.1103/PhysRevE.60.2721} {\bibfield
   {journal} {\bibinfo  {journal} {Phys. Rev. E}\ }\textbf {\bibinfo {volume}
  {60}},\ \bibinfo {pages} {2721} (\bibinfo {year} {1999})},\ \Eprint
  {http://arxiv.org/abs/arXiv:cond-mat/9901352}
  {arXiv:cond-mat/9901352}\BibitemShut {NoStop}%
\bibitem [{\citenamefont {Campisi}\ \emph {et~al.}(2010)\citenamefont
  {Campisi}, \citenamefont {Talkner},\ and\ \citenamefont
  {H{\"a}nggi}}]{CampisiTalknerHaenggi2010}%
  \BibitemOpen
  \bibfield  {author} {\bibinfo {author} {\bibfnamefont {Michele}\ \bibnamefont
  {Campisi}}, \bibinfo {author} {\bibfnamefont {Peter}\ \bibnamefont
  {Talkner}}, \ and\ \bibinfo {author} {\bibfnamefont {Peter}\ \bibnamefont
  {H{\"a}nggi}},\ }\emph {\enquote {\bibinfo {title} {{Fluctuation Theorems for
  Continuously Monitored Quantum Fluxes}},}\ }\href {\doibase
  10.1103/PhysRevLett.105.140601} {\bibfield  {journal} {\bibinfo  {journal}
  {Phys. Rev. Lett.}\ }\textbf {\bibinfo {volume} {105}},\ \bibinfo {pages}
  {140601} (\bibinfo {year} {2010})},\ \Eprint
  {http://arxiv.org/abs/arXiv:1006.1542} {arXiv:1006.1542}\BibitemShut
  {NoStop}%
\bibitem [{\citenamefont {Rastegin}(2013)}]{Rastegin2013}%
  \BibitemOpen
  \bibfield  {author} {\bibinfo {author} {\bibfnamefont {Alexey~E.}\
  \bibnamefont {Rastegin}},\ }\emph {\enquote {\bibinfo {title}
  {Non-equilibrium equalities with unital quantum channels},}\ }\href {\doibase
  10.1088/1742-5468/2013/06/p06016} {\bibfield  {journal} {\bibinfo  {journal}
  {J. Stat. Mech.}\ }\textbf {\bibinfo {volume} {2013}},\ \bibinfo {pages}
  {P06016} (\bibinfo {year} {2013})},\ \Eprint
  {http://arxiv.org/abs/arXiv:1301.0855} {arXiv:1301.0855}\BibitemShut
  {NoStop}%
\bibitem [{\citenamefont {Watanabe}\ \emph {et~al.}(2014)\citenamefont
  {Watanabe}, \citenamefont {Venkatesh}, \citenamefont {Talkner}, \citenamefont
  {Campisi},\ and\ \citenamefont
  {H{\"a}nggi}}]{WatanabeVenkateshTalknerCampisiHaenggi2014}%
  \BibitemOpen
  \bibfield  {author} {\bibinfo {author} {\bibfnamefont {Gentaro}\ \bibnamefont
  {Watanabe}}, \bibinfo {author} {\bibfnamefont {B.~Prasanna}\ \bibnamefont
  {Venkatesh}}, \bibinfo {author} {\bibfnamefont {Peter}\ \bibnamefont
  {Talkner}}, \bibinfo {author} {\bibfnamefont {Michele}\ \bibnamefont
  {Campisi}}, \ and\ \bibinfo {author} {\bibfnamefont {Peter}\ \bibnamefont
  {H{\"a}nggi}},\ }\emph {\enquote {\bibinfo {title} {Quantum fluctuation
  theorems and generalized measurements during the force protocol},}\ }\href
  {\doibase 10.1103/PhysRevE.89.032114} {\bibfield  {journal} {\bibinfo
  {journal} {Phys. Rev. E}\ }\textbf {\bibinfo {volume} {89}},\ \bibinfo
  {pages} {032114} (\bibinfo {year} {2014})},\ \Eprint
  {http://arxiv.org/abs/arXiv:1312.7104} {arXiv:1312.7104}\BibitemShut
  {NoStop}%
\bibitem [{\citenamefont {Manzano}\ \emph {et~al.}(2015)\citenamefont
  {Manzano}, \citenamefont {Horowitz},\ and\ \citenamefont
  {Parrondo}}]{ManzanoHorowitzParrondo2015}%
  \BibitemOpen
  \bibfield  {author} {\bibinfo {author} {\bibfnamefont {Gonzalo}\ \bibnamefont
  {Manzano}}, \bibinfo {author} {\bibfnamefont {Jordan~M.}\ \bibnamefont
  {Horowitz}}, \ and\ \bibinfo {author} {\bibfnamefont {Juan M.~R.}\
  \bibnamefont {Parrondo}},\ }\emph {\enquote {\bibinfo {title} {Nonequilibrium
  potential and fluctuation theorems for quantum maps},}\ }\href {\doibase
  10.1103/PhysRevE.92.032129} {\bibfield  {journal} {\bibinfo  {journal} {Phys.
  Rev. E}\ }\textbf {\bibinfo {volume} {92}},\ \bibinfo {pages} {032129}
  (\bibinfo {year} {2015})},\ \Eprint {http://arxiv.org/abs/arXiv:1505.04201}
  {arXiv:1505.04201}\BibitemShut {NoStop}%
\bibitem [{\citenamefont {Deffner}\ \emph {et~al.}(2016)\citenamefont
  {Deffner}, \citenamefont {Paz},\ and\ \citenamefont
  {Zurek}}]{DeffnerPazZurek2016}%
  \BibitemOpen
  \bibfield  {author} {\bibinfo {author} {\bibfnamefont {Sebastian}\
  \bibnamefont {Deffner}}, \bibinfo {author} {\bibfnamefont {Juan~Pablo}\
  \bibnamefont {Paz}}, \ and\ \bibinfo {author} {\bibfnamefont
  {Wojciech~Hubert}\ \bibnamefont {Zurek}},\ }\emph {\enquote {\bibinfo {title}
  {Quantum work and the thermodynamic cost of quantum measurements},}\ }\href
  {\doibase 10.1103/PhysRevE.94.010103} {\bibfield  {journal} {\bibinfo
  {journal} {Phys. Rev. E}\ }\textbf {\bibinfo {volume} {94}},\ \bibinfo
  {pages} {010103} (\bibinfo {year} {2016})},\ \Eprint
  {http://arxiv.org/abs/arXiv:1603.06509} {arXiv:1603.06509}\BibitemShut
  {NoStop}%
\bibitem [{\citenamefont {Spohn}(1978)}]{Spohn1978}%
  \BibitemOpen
  \bibfield  {author} {\bibinfo {author} {\bibfnamefont {Herbert}\ \bibnamefont
  {Spohn}},\ }\emph {\enquote {\bibinfo {title} {Entropy production for quantum
  dynamical semigroups},}\ }\href {\doibase 10.1063/1.523789} {\bibfield
  {journal} {\bibinfo  {journal} {J. Math. Phys.}\ }\textbf {\bibinfo {volume}
  {19}},\ \bibinfo {pages} {1227} (\bibinfo {year} {1978})}\BibitemShut
  {NoStop}%
\bibitem [{\citenamefont {Breuer}\ and\ \citenamefont
  {Petruccione}(2002)}]{BreuerPetruccione2002}%
  \BibitemOpen
  \bibfield  {author} {\bibinfo {author} {\bibfnamefont {Heinz-Peter}\
  \bibnamefont {Breuer}}\ and\ \bibinfo {author} {\bibfnamefont {Francesco}\
  \bibnamefont {Petruccione}},\ }\href@noop {} {\emph {\bibinfo {title} {The
  theory of open quantum systems}}}\ (\bibinfo  {publisher} {Oxford University
  Press},\ \bibinfo {address} {Oxford},\ \bibinfo {year} {2002})\BibitemShut
  {NoStop}%
\bibitem [{\citenamefont {Haake}(2010)}]{Haake}%
  \BibitemOpen
  \bibfield  {author} {\bibinfo {author} {\bibfnamefont {Fritz}\ \bibnamefont
  {Haake}},\ }\href@noop {} {\emph {\bibinfo {title} {{Quantum Signatures of
  Chaos}}}}\ (\bibinfo  {publisher} {Springer},\ \bibinfo {address} {Berlin
  Heidelberg},\ \bibinfo {year} {2010})\BibitemShut {NoStop}%
\bibitem [{\citenamefont {Funo}\ \emph {et~al.}(2018)\citenamefont {Funo},
  \citenamefont {Ueda},\ and\ \citenamefont {Sagawa}}]{FunoUedaSagawa2018}%
  \BibitemOpen
  \bibfield  {author} {\bibinfo {author} {\bibfnamefont {Ken}\ \bibnamefont
  {Funo}}, \bibinfo {author} {\bibfnamefont {Masahito}\ \bibnamefont {Ueda}}, \
  and\ \bibinfo {author} {\bibfnamefont {Takahiro}\ \bibnamefont {Sagawa}},\
  }in\ \href {\doibase 10.1007/978-3-319-99046-0_10} {\emph {\bibinfo
  {booktitle} {Thermodynamics in the Quantum Regime}}},\ \bibinfo {editor}
  {edited by\ \bibinfo {editor} {\bibfnamefont {Felix}\ \bibnamefont {Binder}},
  \bibinfo {editor} {\bibfnamefont {Luis~A.}\ \bibnamefont {Correa}}, \bibinfo
  {editor} {\bibfnamefont {Christian}\ \bibnamefont {Gogolin}}, \bibinfo
  {editor} {\bibfnamefont {Janet}\ \bibnamefont {Anders}}, \ and\ \bibinfo
  {editor} {\bibfnamefont {Gerardo}\ \bibnamefont {Adesso}}}\ (\bibinfo
  {publisher} {Springer},\ \bibinfo {year} {2018})\ Chap.~\bibinfo {chapter}
  {10}, pp.\ \bibinfo {pages} {249--273},\ \Eprint
  {http://arxiv.org/abs/arXiv:1803.04778} {arXiv:1803.04778}\BibitemShut
  {NoStop}%
\bibitem [{\citenamefont {Potts}\ and\ \citenamefont
  {Samuelsson}(2018)}]{PottsSamuelsson2018}%
  \BibitemOpen
  \bibfield  {author} {\bibinfo {author} {\bibfnamefont {Patrick~P.}\
  \bibnamefont {Potts}}\ and\ \bibinfo {author} {\bibfnamefont {Peter}\
  \bibnamefont {Samuelsson}},\ }\emph {\enquote {\bibinfo {title} {{Detailed
  Fluctuation Relation for Arbitrary Measurement and Feedback Schemes}},}\
  }\href {\doibase 10.1103/PhysRevLett.121.210603} {\bibfield  {journal}
  {\bibinfo  {journal} {Phys. Rev. Lett.}\ }\textbf {\bibinfo {volume} {121}},\
  \bibinfo {pages} {210603} (\bibinfo {year} {2018})},\ \Eprint
  {http://arxiv.org/abs/arXiv:1807.05034} {arXiv:1807.05034}\BibitemShut
  {NoStop}%
\bibitem [{\citenamefont {Ito}\ \emph {et~al.}(2019)\citenamefont {Ito},
  \citenamefont {Talkner}, \citenamefont {Venkatesh},\ and\ \citenamefont
  {Watanabe}}]{ItoTalknerVenkatesh2019}%
  \BibitemOpen
  \bibfield  {author} {\bibinfo {author} {\bibfnamefont {Kosuke}\ \bibnamefont
  {Ito}}, \bibinfo {author} {\bibfnamefont {Peter}\ \bibnamefont {Talkner}},
  \bibinfo {author} {\bibfnamefont {B.~Prasanna}\ \bibnamefont {Venkatesh}}, \
  and\ \bibinfo {author} {\bibfnamefont {Gentaro}\ \bibnamefont {Watanabe}},\
  }\emph {\enquote {\bibinfo {title} {Generalized energy measurements and
  quantum work compatible with fluctuation theorems},}\ }\href {\doibase
  10.1103/PhysRevA.99.032117} {\bibfield  {journal} {\bibinfo  {journal} {Phys.
  Rev. A}\ }\textbf {\bibinfo {volume} {99}},\ \bibinfo {pages} {032117}
  (\bibinfo {year} {2019})},\ \Eprint {http://arxiv.org/abs/arXiv:1812.07289}
  {arXiv:1812.07289}\BibitemShut {NoStop}%
\bibitem [{\citenamefont {Lostaglio}\ \emph {et~al.}(2015)\citenamefont
  {Lostaglio}, \citenamefont {Jennings},\ and\ \citenamefont
  {Rudolph}}]{LostaglioJenningsRudolph2015}%
  \BibitemOpen
  \bibfield  {author} {\bibinfo {author} {\bibfnamefont {Matteo}\ \bibnamefont
  {Lostaglio}}, \bibinfo {author} {\bibfnamefont {David}\ \bibnamefont
  {Jennings}}, \ and\ \bibinfo {author} {\bibfnamefont {Terry}\ \bibnamefont
  {Rudolph}},\ }\emph {\enquote {\bibinfo {title} {{Description of quantum
  coherence in thermodynamic processes requires constraints beyond free
  energy}},}\ }\href {\doibase 10.1038/ncomms7383} {\bibfield  {journal}
  {\bibinfo  {journal} {Nat. Commun.}\ }\textbf {\bibinfo {volume} {6}},\
  \bibinfo {pages} {6383} (\bibinfo {year} {2015})},\ \Eprint
  {http://arxiv.org/abs/arXiv:1405.2188} {arXiv:1405.2188}\BibitemShut
  {NoStop}%
\bibitem [{\citenamefont {Watrous}(2018)}]{Watrous2018}%
  \BibitemOpen
  \bibfield  {author} {\bibinfo {author} {\bibfnamefont {John}\ \bibnamefont
  {Watrous}},\ }\href {\doibase 10.1017/9781316848142} {\emph {\bibinfo {title}
  {{The Theory of Quantum Information}}}}\ (\bibinfo  {publisher} {Cambridge
  University Press},\ \bibinfo {address} {Cambridge, U.K.},\ \bibinfo {year}
  {2018})\BibitemShut {NoStop}%
\bibitem [{\citenamefont {Sagawa}\ and\ \citenamefont
  {Ueda}(2012)}]{SagawaUeda2012}%
  \BibitemOpen
  \bibfield  {author} {\bibinfo {author} {\bibfnamefont {Takahiro}\
  \bibnamefont {Sagawa}}\ and\ \bibinfo {author} {\bibfnamefont {Masahito}\
  \bibnamefont {Ueda}},\ }\emph {\enquote {\bibinfo {title} {{Nonequilibrium
  thermodynamics of feedback control}},}\ }\href {\doibase
  10.1103/PhysRevE.85.021104} {\bibfield  {journal} {\bibinfo  {journal} {Phys.
  Rev. E}\ }\textbf {\bibinfo {volume} {85}},\ \bibinfo {pages} {021104}
  (\bibinfo {year} {2012})},\ \Eprint {http://arxiv.org/abs/arXiv:1105.3262}
  {arXiv:1105.3262}\BibitemShut {NoStop}%
\bibitem [{\citenamefont {Fannes}(1973)}]{Fannes1973}%
  \BibitemOpen
  \bibfield  {author} {\bibinfo {author} {\bibfnamefont {Mark}\ \bibnamefont
  {Fannes}},\ }\emph {\enquote {\bibinfo {title} {A continuity property of the
  entropy density for spin lattice systems},}\ }\href {\doibase
  10.1007/BF01646490} {\bibfield  {journal} {\bibinfo  {journal} {Commun. Math.
  Phys.}\ }\textbf {\bibinfo {volume} {31}},\ \bibinfo {pages} {291} (\bibinfo
  {year} {1973})}\BibitemShut {NoStop}%
\bibitem [{\citenamefont {Audenaert}(2007)}]{Audenaert2007}%
  \BibitemOpen
  \bibfield  {author} {\bibinfo {author} {\bibfnamefont {Koenraad M.~R.}\
  \bibnamefont {Audenaert}},\ }\emph {\enquote {\bibinfo {title} {A {S}harp
  {C}ontinuity {E}stimate for the von {N}eumann {E}ntropy},}\ }\href {\doibase
  10.1088/1751-8113/40/28/S18} {\bibfield  {journal} {\bibinfo  {journal} {J.
  Phys. A: Math. Theor.}\ }\textbf {\bibinfo {volume} {40}},\ \bibinfo {pages}
  {8127} (\bibinfo {year} {2007})},\ \Eprint
  {http://arxiv.org/abs/arXiv:quant-ph/0610146}
  {arXiv:quant-ph/0610146}\BibitemShut {NoStop}%
\bibitem [{\citenamefont {Haroche}\ and\ \citenamefont
  {Raimond}(2006)}]{HarocheRaimond2006}%
  \BibitemOpen
  \bibfield  {author} {\bibinfo {author} {\bibfnamefont {Serge}\ \bibnamefont
  {Haroche}}\ and\ \bibinfo {author} {\bibfnamefont {Jean-Michel}\ \bibnamefont
  {Raimond}},\ }\href {\doibase 10.1093/acprof:oso/9780198509141.001.0001}
  {\emph {\bibinfo {title} {Exploring the Quantum: Atoms, Cavities, And
  Photons}}}\ (\bibinfo  {publisher} {Oxford University Press},\ \bibinfo
  {address} {Oxford},\ \bibinfo {year} {2006})\BibitemShut {NoStop}%
\end{thebibliography}%

%%%%%%%%%%%%%%%%%%%%%%%%%%%%%%%%%%%%%%%%%%%%%%%%%%%%%%%%%%%%%%%%%%%%%%%%%%%%%%%%%%%%%%%%%%%%%%%%%%%%

\newpage
\hypertarget{sec:appendix}
\onecolumngrid
\appendix
\renewcommand{\thesubsection}{A.\arabic{section}.\alph{subsection}}
\renewcommand{\thesection}{A.\arabic{section}}
\setcounter{equation}{0}
\numberwithin{equation}{section}
\setcounter{figure}{0}
\renewcommand{\thefigure}{A.\arabic{figure}}

%\clearpage
%\newpage
\section*{Appendices}

In these appendices, we provide more details on the mathematical model for non-ideal measurements from Ref.~\cite{GuryanovaFriisHuber2018} and its application for work estimation described in this paper. In Appendix~\ref{appendix:Ideal and Non-Ideal Measurements}, we describe the properties of the important class of \emph{unbiased maximally correlated} (UMC) measurements, as well as of minimal energy UMC measurements and illustrate them for the case of $3$-dim quantum system.  In Appendix~\ref{sec:ttpm}, we give a detailed derivation of the joint probability for the measurement outcomes in the TPM scheme~\cite{TalknerLutzHaenggi2007, CampisiHaenggiTalkner2011} using non-ideal measurements. We then derive the corresponding work estimate in Appendix~\ref{sec:work estimation nonid tpm} and the change in average energy throughout the estimation process in Appendix~\ref{sec:change of average energy}. Finally, we discuss the consequences for fluctuation theorems: In Appendix~\ref{sec:Jarjar for non-ideal measurements}, we discuss the modification of Jarzynski's relation, while Appendix~\ref{sec:non ideal crooks} explains why Crooks' theorem generally no longer holds in the presence of non-ideal measurements. In Appendix~\ref{appsec:driven atom} we present a physical model as an example of the methods discussed in this work.

%%%%%%%%%%%%%%%%%%%%%%%%%%%%%%%%%%%%%%%%%%%%%%%%%%%%%%%%%%%%%%%%%%%%%%%%%%%%%%%%%%%%%%%%%%

\section{Ideal and Non-Ideal Measurements}\label{appendix:Ideal and Non-Ideal Measurements}

We consider a measurement of a system described by an unknown quantum state $\rho\Sys\in\M{D}(\M{H}\Sys)$, where $\M{D}(\M{H}\Sys)$ represents the set of density-matrices over the Hilbert space $\M{H}\Sys$. We model the measurement as an interaction between the system and a measurement apparatus (pointer) described by a quantum system with Hilbert space $\M{H}\Poi$. We consider the pointer to be initially described by the thermal state $\tau\Poi(\beta\Poi) = \exp(-\beta\Poi H\Poi)/\mathcal{Z}\Poi $ at ambient temperature $T\Poi=(k\subtiny{0}{0}{\mathrm{B}}\beta\Poi)^{-1}$ and with Hamiltonian $H\Poi = \sum_{i} E_i\suptiny{0}{0}{({P})} \ketbra{E_i\suptiny{0}{0}{({P})}}{E_i\suptiny{0}{0}{({P})}}$. This ensures that the initial state of the pointer does not contain any extractable work with respect to an environment at temperature $T\Poi$. Alternatively, one can consider the temperature $T\Poi$ to be lower than the environment temperature, assuming that work has been invested to prepare the pointer by cooling it down to $T\Poi$. System and pointer are correlated by a unitary $U_{\mathrm{meas}}$ on the joint space $\M{H}\Sys\otimes\M{H}\Poi$ such that all work supplied to the joint system can be identified with the overall change in average energy due to the unitary transformation, resulting in a post-measurement state
\begin{align}
   \tilde{\rho}\SP & := U_{\mathrm{meas}}(\rho\Sys\otimes\tau\Poi) U_{\mathrm{meas}}^{\dagger}.
\end{align}
Within the pointer Hilbert space, different subspaces are assigned to represent the different measurement outcomes. More specifically, if the system of dimension $d\Sys$ is to be measured in the basis $\{\ket{n}\}_{n=0,\ldots,d\Sys-1}$, then a set of orthogonal projectors $\{\Pi_{n}\}_{n=0,\ldots,d\Sys-1}$ on the pointer Hilbert space is chosen to represent these outcomes, such that $\sum_{n}\Pi_{n}=\mathds{1}\Poi$ and $\Pi_{m}\Pi_{n}=\delta_{mn}\Pi_{n}$ and $\tr(\Pi_n)\d\Poi/d\Sys$ for all $n$. The general setup is illustrated in Fig.~\ref{fig:non_ideal}.

%%%%%%%%%%%%%%%%%%%%%%%%%%%%%%%%%%%%%%%%%%%%%%%%%%%%%%%%%

\begin{figure}[ht!]
\begin{center}
\includegraphics[width=0.2\columnwidth,trim={0cm 0mm 0cm 0mm}]{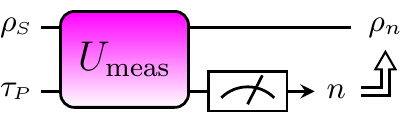}
\caption{Circuit representing non-ideal measurements. The system in an initial state $\rho\Sys$ interacts with the pointer originally in a thermal state $\tau\Poi$ by means of unitary $U_{\mathrm{meas}}$. The probability for obtaining a measurement outcome $n$ is $p_{n}$, and the post-measurement system state conditioned on having obtained outcome $n$ is $\rho_{n}$.}
\label{fig:non_ideal}
\end{center}
\end{figure}
\vspace*{-2mm}

%%%%%%%%%%%%%%%%%%%%%%%%%%%%%%%%%%%%%%%%%%%%%%%%%%%%%%%%%

\subsection{Unbiased Measurements}\label{appsec:unbiased measurements}

As explained in Sec.~\ref{sec:framework} (following Ref.~\cite{GuryanovaFriisHuber2018}), ideal measurements have three characteristic properties, they are (\ref{item:faithful main}) \emph{faithful} (perfect correlation between the pointer and system), (\ref{item:unbiased main}) \emph{unbiased} (pointer exactly reproduces the system statistics in the measured basis), and (\ref{item:non-invasive main}) \emph{non-invasive} (system diagonal is undisturbed by the interaction). However, satisfying property (\ref{item:faithful main}) requires the preparation of pure pointer states, which is not possible with finite resources according to the third law of thermodynamics. At the same time, this implies that realistic non-ideal measurements cannot satisfy both properties (\ref{item:unbiased main}) and (\ref{item:non-invasive main}) simultaneously, see Ref.~\cite[Appendix A.3]{GuryanovaFriisHuber2018}. We hence focus on non-ideal measurement procedures that are \emph{unbiased}, i.e., which satisfy
\begin{align}
    \tr\bigl(\mathds{1}\Sys\otimes\Pi_{n}\tilde{\rho}\SP\bigr)    &=\,\tr\bigl(\ket{n}\!\!\bra{n}\Sys\,\rho\Sys\bigr)\ \forall\,n\ \ \forall \,
    \rho\Sys\,.
    \label{eq:unbiased}
\end{align}
Provided that the pointer cannot be prepared in a pure state (i.e., that the measurement uses finite resources according to the third law of thermodynamics), the corresponding unitary $U_{\mathrm{meas}}$ can be separated into two consecutive unitaries $\tilde{U}$ and $V$, such that $U_{\mathrm{meas}}=V\tilde{U}$, where
\begin{align}
    \tilde{U}    &=\sum\limits_{n=0}^{d\Sys-1}\ket{n}\!\!\bra{n}\Sys\otimes \tilde{U}\suptiny{0}{0}{(n)},\qquad\text{and}\qquad
    V   =\sum\limits_{m,n=0}^{d\Sys-1}\sum\limits_{i=1}^{d\Poi/d\Sys}
    \ket{m}\!\!\bra{n}\Sys\otimes
    \ket{\tilde{\psi}_{i}\suptiny{0}{0}{(n)}}\!\!\bra{\tilde{\psi}_{i}\suptiny{0}{0}{(m)}\!}\Poi.
    \label{eq:lemma 2 from other paper}
\end{align}
Here, $\{\ket{\tilde{\psi}_{i}\suptiny{0}{0}{(n)}}\}_{i,n}$ is an orthonormal basis of $\mathcal{H}\Poi$ such that $\Pi_{n}=\sum_{i} \ket{\tilde{\psi}_{i}\suptiny{0}{0}{(n)}}\!\!\bra{\tilde{\psi}_{i}\suptiny{0}{0}{(n)}}$ and the $\tilde{U}\suptiny{0}{0}{(n)}$ are arbitrary unitaries on $\mathcal{H}\Poi$, and we have assumed that the dimension $d\Poi$ of the pointer is an integer multiple of the system dimension $d\Sys$ (or is truncated to such a dimension). The intuition behind this decomposition is as follows. The operation $V$ ensures unbiasedness by mapping the subspace corresponding to the system state $\ket{n}\Sys$ to the subspace corresponding to the pointer outcome $n$ (i.e., $\Pi_{n}$), that is, $V$ can be understood as a type of swap between these subspaces. The operation $\tilde{U}$, meanwhile, adds additional freedom by allowing unitary transformations within each of the subspaces for the states $\ket{n}\Sys$. To illustrate this transformation, let us consider an example of a $3$-dimensional system with non-degenerate Hamiltonian, and a suitable pointer of dimension $d\Poi=3\lambda$ for $\lambda\in\mathbb{N}$. The initial state $\rho\Sys\otimes\tau\Poi$ of system and pointer can be expressed in matrix form as
\begin{align}
    \rho\SP &=\,
    \left(\begin{array}{c|c|c}
        \tikzmarkin[ver=style green, line width=0mm,]{a1}\raisebox{2pt}{\protect\footnotesize{$p_0\,\tau\Poi$}}\tikzmarkend{a1}  & \cdot & \cdot \\
        \hline
        \cdot & \tikzmarkin[ver=style orange, line width=0mm,]{a2}\raisebox{2pt}{\protect\footnotesize{\hspace*{1pt}$p_1\,\tau\Poi$\hspace*{1pt}}}\tikzmarkend{a2}  & \cdot \\
        \hline
        \cdot  & \cdot & \tikzmarkin[ver=style cyan, line width=0mm,]{a33}\raisebox{2pt}{\protect\footnotesize{$p_2\,\tau\Poi$}}\tikzmarkend{a33}
    \end{array}\right),\nonumber\\[-5mm]
    &\phantom{=\,( }
    \begin{array}{ccc}
            \hspace*{9mm} &  &  \\[-2mm]
        $\upbracefill$ & \hspace*{-1.8mm}$\upbracefill$ & \hspace*{-1.5mm}$\upbracefill$  \\[-2mm]
        \hspace*{1.7mm}\text{\protect\scriptsize{$\ket{0}\Sys$}} & \hspace*{1.7mm}\text{\protect\scriptsize{$\ket{1}\Sys$}} & \hspace*{1.7mm}\text{\protect\scriptsize{$\ket{2}\Sys$}}
    \end{array}
    \label{eq:instate}
\end{align}
where $p_{i}=\bra{i}\rho\Sys\ket{i}$ and the dots indicate potentially nonzero off-diagonal elements of $\rho\Sys$ that are not shown here to keep the example simple, but could be included explicitly if desired. An unbiased measurement procedure realized by a unitary interaction with the pointer then leads to a final joint state of the form
\vspace*{-5mm}
\begin{align}
   &%\hspace*{1mm}
  \begin{array}{ccccccccc}
  \hspace*{5mm} & \hspace*{5mm} & \hspace*{5mm} &
  \hspace*{5mm} & \hspace*{5mm} & \hspace*{5mm} &
  \hspace*{5mm} & \hspace*{5mm} & \hspace*{5mm} \\
  \hspace*{3.5mm}\text{\protect\scriptsize{$\Pi_{0}$}} & \hspace*{0mm}\text{\protect\scriptsize{$\Pi_{1}$}} & \hspace*{0mm}\text{\protect\scriptsize{$\Pi_{2}$}} &
  \hspace*{0mm}\text{\protect\scriptsize{$\Pi_{0}$}} & \hspace*{0mm}\text{\protect\scriptsize{$\Pi_{1}$}} & \hspace*{0mm}\text{\protect\scriptsize{$\Pi_{2}$}} &
  \hspace*{0mm}\text{\protect\scriptsize{$\Pi_{0}$}} & \hspace*{0mm}\text{\protect\scriptsize{$\Pi_{1}$}} & \hspace*{0mm}\text{\protect\scriptsize{$\Pi_{2}$}} \\[-1.5mm]
  \hspace*{2.0mm}$\downbracefill$ & \hspace*{-1.3mm}$\downbracefill$ & \hspace*{-1.3mm}$\downbracefill$ &
  \hspace*{-1.3mm}$\downbracefill$ & \hspace*{-1.3mm}$\downbracefill$ & \hspace*{-1.3mm}$\downbracefill$ &
  \hspace*{-1.3mm}$\downbracefill$ & \hspace*{-1.3mm}$\downbracefill$ & \hspace*{-1.3mm}$\downbracefill$
  \end{array}
  \nonumber\\[-1mm]
{\normalsize \tilde{\rho}\SP}=  &
  \left(\begin{array}{c|c|c|c|c|c|c|c|c}
    %%%%%%%%%%%%%%%%%%%
    \tikzmarkin[ver=style green, line width=0mm,]{g00}
        \raisebox{2pt}{\!\protect\scriptsize{\hspace*{1mm}$\Gamma_{00}$\hspace*{1mm}}}
    \tikzmarkend{g00}
    & \cdot & \cdot &
    \tikzmarkin[ver=style green, line width=0mm,]{g01}
        \phantom{\raisebox{2pt}{\!\protect\scriptsize{\hspace*{1mm}$\Gamma_{00}$\hspace*{1mm}}}}
    \tikzmarkend{g01}
    & \cdot & \cdot &
    \tikzmarkin[ver=style green, line width=0mm,]{g02}
        \phantom{\raisebox{2pt}{\!\protect\scriptsize{\hspace*{1mm}$\Gamma_{00}$\hspace*{1mm}}}}
    \tikzmarkend{g02}
    & \cdot & \cdot \\
    \hline
    %%%%%%%%%%%%%%%%%%%
    \cdot &
    \tikzmarkin[ver=style orange, line width=0mm,]{o00}
        \raisebox{2pt}{\!\protect\scriptsize{\hspace*{1mm}$\Gamma_{01}$\hspace*{1mm}}}
    \tikzmarkend{o00}
    & \cdot & \cdot &
    \tikzmarkin[ver=style orange, line width=0mm,]{o01}
        \phantom{\raisebox{2pt}{\!\protect\scriptsize{\hspace*{1mm}$\Gamma_{01}$\hspace*{1mm}}}}
    \tikzmarkend{o01}
    & \cdot & \cdot &
    \tikzmarkin[ver=style orange, line width=0mm,]{o02}
        \phantom{\raisebox{2pt}{\!\protect\scriptsize{\hspace*{1mm}$\Gamma_{01}$\hspace*{1mm}}}}
    \tikzmarkend{o02}
    & \cdot \\
    \hline
    %%%%%%%%%%%%%%%%%%%
    \cdot & \cdot &
    \tikzmarkin[ver=style cyan, line width=0mm,]{c00}
        \raisebox{2pt}{\!\protect\scriptsize{\hspace*{1mm}$\Gamma_{02}$\hspace*{1mm}}}
    \tikzmarkend{c00}
    & \cdot & \cdot &
    \tikzmarkin[ver=style cyan, line width=0mm,]{c01}
        \phantom{\raisebox{2pt}{\!\protect\scriptsize{\hspace*{1mm}$\Gamma_{02}$\hspace*{1mm}}}}
    \tikzmarkend{c01}
    & \cdot & \cdot &
    \tikzmarkin[ver=style cyan, line width=0mm,]{c02}
        \phantom{\raisebox{2pt}{\!\protect\scriptsize{\hspace*{1mm}$\Gamma_{02}$\hspace*{1mm}}}}
    \tikzmarkend{c02}
    \\
    \hline
    %%%%%%%%%%%%%%%%%%%
    \tikzmarkin[ver=style green, line width=0mm,]{g10}
        \phantom{\raisebox{2pt}{\!\protect\scriptsize{\hspace*{1mm}$\Gamma_{10}$\hspace*{1mm}}}}
    \tikzmarkend{g10}
    & \cdot & \cdot &
    \tikzmarkin[ver=style green, line width=0mm,]{g11}
        \raisebox{2pt}{\!\protect\scriptsize{\hspace*{1mm}$\Gamma_{10}$\hspace*{1mm}}}
    \tikzmarkend{g11}
    & \cdot & \cdot &
    \tikzmarkin[ver=style green, line width=0mm,]{g12}
        \phantom{\raisebox{2pt}{\!\protect\scriptsize{\hspace*{1mm}$\Gamma_{10}$\hspace*{1mm}}}}
    \tikzmarkend{g12}
    & \cdot & \cdot \\
    \hline
    %%%%%%%%%%%%%%%%%%%
    \cdot &
    \tikzmarkin[ver=style orange, line width=0mm,]{o10}
        \phantom{\raisebox{2pt}{\!\protect\scriptsize{\hspace*{1mm}$\Gamma_{11}$\hspace*{1mm}}}}
    \tikzmarkend{o10}
    & \cdot & \cdot &
    \tikzmarkin[ver=style orange, line width=0mm,]{o11}
        \raisebox{2pt}{\!\protect\scriptsize{\hspace*{1mm}$\Gamma_{11}$\hspace*{1mm}}}
    \tikzmarkend{o11}
    & \cdot & \cdot &
    \tikzmarkin[ver=style orange, line width=0mm,]{o12}
        \phantom{\raisebox{2pt}{\!\protect\scriptsize{\hspace*{1mm}$\Gamma_{11}$\hspace*{1mm}}}}
    \tikzmarkend{o12}
    & \cdot \\
    \hline
    %%%%%%%%%%%%%%%%%%%
    \cdot & \cdot &
    \tikzmarkin[ver=style cyan, line width=0mm,]{c10}
        \phantom{\raisebox{2pt}{\!\protect\scriptsize{\hspace*{1mm}$\Gamma_{12}$\hspace*{1mm}}}}
    \tikzmarkend{c10}
    & \cdot & \cdot &
    \tikzmarkin[ver=style cyan, line width=0mm,]{c11}
        \raisebox{2pt}{\!\protect\scriptsize{\hspace*{1mm}$\Gamma_{12}$\hspace*{1mm}}}
    \tikzmarkend{c11}
    & \cdot & \cdot &
    \tikzmarkin[ver=style cyan, line width=0mm,]{c12}
        \phantom{\raisebox{2pt}{\!\protect\scriptsize{\hspace*{1mm}$\Gamma_{12}$\hspace*{1mm}}}}
    \tikzmarkend{c12}
    \\
    \hline
    %%%%%%%%%%%%%%%%%%%%
    \tikzmarkin[ver=style green, line width=0mm,]{g20}
        \phantom{\raisebox{2pt}{\!\protect\scriptsize{\hspace*{1mm}$\Gamma_{20}$\hspace*{1mm}}}}
    \tikzmarkend{g20}
    & \cdot & \cdot &
    \tikzmarkin[ver=style green, line width=0mm,]{g21}
        \phantom{\raisebox{2pt}{\!\protect\scriptsize{\hspace*{1mm}$\Gamma_{20}$\hspace*{1mm}}}}
    \tikzmarkend{g21}
    & \cdot & \cdot &
    \tikzmarkin[ver=style green, line width=0mm,]{g22}
        \raisebox{2pt}{\!\protect\scriptsize{\hspace*{1mm}$\Gamma_{20}$\hspace*{1mm}}}
    \tikzmarkend{g22}
    & \cdot & \cdot \\
    \hline
    %%%%%%%%%%%%%%%%%%%
    \cdot &
    \tikzmarkin[ver=style orange, line width=0mm,]{o20}
        \phantom{\raisebox{2pt}{\!\protect\scriptsize{\hspace*{1mm}$\Gamma_{21}$\hspace*{1mm}}}}
    \tikzmarkend{o20}
    & \cdot & \cdot &
    \tikzmarkin[ver=style orange, line width=0mm,]{o21}
        \phantom{\raisebox{2pt}{\!\protect\scriptsize{\hspace*{1mm}$\Gamma_{21}$\hspace*{1mm}}}}
    \tikzmarkend{o21}
    & \cdot & \cdot &
    \tikzmarkin[ver=style orange, line width=0mm,]{o22}
        \raisebox{2pt}{\!\protect\scriptsize{\hspace*{1mm}$\Gamma_{21}$\hspace*{1mm}}}
    \tikzmarkend{o22}
    & \cdot \\
    \hline
    %%%%%%%%%%%%%%%%%%%
    \cdot & \cdot &
    \tikzmarkin[ver=style cyan, line width=0mm,]{c20}
        \phantom{\raisebox{2pt}{\!\protect\scriptsize{\hspace*{1mm}$\Gamma_{22}$\hspace*{1mm}}}}
    \tikzmarkend{c20}
    & \cdot & \cdot &
    \tikzmarkin[ver=style cyan, line width=0mm,]{c21}
        \phantom{\raisebox{2pt}{\!\protect\scriptsize{\hspace*{1mm}$\Gamma_{22}$\hspace*{1mm}}}}
    \tikzmarkend{c21}
    & \cdot & \cdot &
    \tikzmarkin[ver=style cyan, line width=0mm,]{c22}
        \raisebox{2pt}{\!\protect\scriptsize{\hspace*{1mm}$\Gamma_{22}$\hspace*{1mm}}}
    \tikzmarkend{c22}
  \end{array}\right),\nonumber\\[-5.5mm]
  &\,
  \begin{array}{ccc}
  \hspace*{24.0mm} & \hspace*{21.5mm} & \hspace*{21.5mm}\\
  \hspace*{2.0mm}$\upbracefill$ & \hspace*{-1.3mm}$\upbracefill$ & \hspace*{-1.3mm}$\upbracefill$\\[-0.5mm]
  \hspace*{4.5mm}\text{\protect\scriptsize{$\ket{0}\Sys$}} & \hspace*{0.7mm}\text{\protect\scriptsize{$\ket{1}\Sys$}} & \hspace*{0.7mm}\text{\protect\scriptsize{$\ket{2}\Sys$}}
  \end{array}
  \nonumber
\end{align}
\vspace*{-9mm}
\begin{align}
\label{eq:appendrow}
\end{align}
where the $\Gamma_{ij}$ are the $(d\Poi/d\Sys)\times(d\Poi/d\Sys)$ matrices corresponding to $\ket{i}\!\!\bra{i}\Sys\otimes\Pi_{j}\,\tilde{\rho}\SP \ket{i}\!\!\bra{i}\Sys\otimes\Pi_{j}$. The colour coding corresponds to the projections on to the subspaces for fixed outcomes. That is, the operators $\mathds{1}\Sys\otimes\Pi_{n}\,\tilde{\rho}\SP \mathds{1}\Sys\otimes\Pi_{n}$, can be written as
\begin{align}
    \mathds{1}\Sys\otimes\Pi_{n}\,\tilde{\rho}\SP \mathds{1}\Sys\otimes\Pi_{n}
    &=\,
    \sum\limits_{m,m'=0}^{d\Sys-1}\sum\limits_{i,j=1}^{d\Poi/d\Sys}
    \bra{n}\rho\Sys\ket{n}\,
    \bra{\tilde{\psi}_{i}\suptiny{0}{0}{(m)}}
    \tilde{U}\suptiny{0}{0}{(n)}\tau\Poi \tilde{U}\suptiny{0}{0}{(n)}{}^{\dagger}
    \ket{\tilde{\psi}_{j}\suptiny{0}{0}{(m')}}
    \
    \ket{m}\!\!\bra{m'}\Sys\otimes
    \ket{\tilde{\psi}_{i}\suptiny{0}{0}{(n)}}\!\!\bra{\tilde{\psi}_{j}\suptiny{0}{0}{(n)}\!}\Poi\nonumber\\[1.5mm]
    &\mathrel{\hat=}
    \,p_{n}\,\tilde{U}\suptiny{0}{0}{(n)}\tau\Poi \tilde{U}\suptiny{0}{0}{(n)}{}^{\dagger}.
    \label{eq:equivalence}
\end{align}
In other words, the $\mathds{1}\Sys\otimes\Pi_{n}\,\tilde{\rho}\SP \mathds{1}\Sys\otimes\Pi_{n}$ for $n\in\{0,1,\ldots,d\Sys-1\}$ have rank $d\Poi$ and can be understood as $d\Poi\times d\Poi$ matrices with the same spectra as $p_{n}\tau\Poi$. They are in this sense unitarily equivalent [via application of the unitaries $\tilde{U}\suptiny{0}{0}{(n)}$ from Eq.~(\ref{eq:lemma 2 from other paper})] to $p_{n}\tau\Poi$, such that the trace of these operators gives $p_{n}$, as required by unbiasedness in Eq.~(\ref{eq:unbiased}). In terms of the indicated diagonal blocks of each subspace for fixed $\Pi_{n}$ (fixed colour), we have $\sum_{i=0}^{2}\tr(\Gamma_{in})=p_{n} \ \forall\,n$. More specifically, since the measurement has to be unbiased independently of the initial system state $\rho\Sys$, one has to have
\begin{align}
    \tr\Poi\bigl(\mathds{1}\Sys\otimes\Pi_{n}\,\tilde{\rho}\SP \mathds{1}\Sys\otimes\Pi_{n}\bigr)
    &=\,p_{n}\,\rho_{n},
    \label{eq:unbiased conditional post meas state}
\end{align}
where $\rho_{n}\in\M{D}(\M{H}\Sys)$ is the post-measurement system state conditioned on observing the pointer outcome $n$. The conditional states $\rho_{n}$ correspond to coarse-grainings of the operators $\tilde{U}\suptiny{0}{0}{(n)}\tau\Poi \tilde{U}\suptiny{0}{0}{(n)}{}^{\dagger}$ in the sense of the equivalence ``$\mathrel{\hat=}$" of Eq.~(\ref{eq:equivalence}). The states $\rho_{n}$ are hence independent of $\rho\Sys$. We write them as $\rho_{n}    = \sum\limits_{l,l\pr}q_{ll\pr|n}\,\ket{l}\!\!\bra{l\pr}$ with $\sum_{l}q_{ll|n}=1\ \forall\,n$ and
\begin{align}
    q_{ll\pr|n} &=\,
    \sum\limits_{i=1}^{d\Poi/d\Sys}
    \bra{\tilde{\psi}_{i}\suptiny{0}{0}{(l)}}
    \tilde{U}\suptiny{0}{0}{(n)}\tau\Poi \tilde{U}\suptiny{0}{0}{(n)}{}^{\dagger}
    \ket{\tilde{\psi}_{i}\suptiny{0}{0}{(l')}}.
\end{align}

%%%%%%%%%%%%%%%%%%%%%%%%%%%%%%%%%%%%%%%%%%%%%%%%%%%%%%%%%

\subsection{Unbiased Maximally Correlated Measurements}\label{appsec:unbiased max corr measurements}

Within the set of unbiased measurement procedures, one may then wish to select those measurements which maximise the correlation measure $C(\tilde{\rho}\SP)$ that is used in property (\ref{item:faithful main}) to define faithful measurements (for which $C=1$). For unbiased measurements that are realised unitarily from an initially thermal pointer state, the maximal value $C=C_{\mathrm{max}}$ for unbiased measurements is obtained when the $d\Poi/d\Sys$ largest eigenvalues of $\tau\Poi$ appear as the eigenvalues of the matrices $\Gamma_{nn}$ for each $n$. This requires the off-diagonal elements in $\tilde{\rho}\SP$ connecting the correlated subspaces defined by the projectors $\ket{n}\!\!\bra{n}\otimes\Pi_{n}$ with the uncorrelated subspaces defined by the projectors $\ket{m}\!\!\bra{m}\otimes\Pi_{n}$ for $m\neq n$ to vanish, resulting in a joint state of the form
\vspace*{-5mm}
\begin{align}
   &%\hspace*{1mm}
    \begin{array}{ccccccccc}
  \hspace*{5mm} & \hspace*{5mm} & \hspace*{5mm} &
  \hspace*{5mm} & \hspace*{5mm} & \hspace*{5mm} &
  \hspace*{5mm} & \hspace*{5mm} & \hspace*{5mm} \\
  \hspace*{3.5mm}\text{\protect\scriptsize{$\Pi_{0}$}} & \hspace*{0mm}\text{\protect\scriptsize{$\Pi_{1}$}} & \hspace*{0mm}\text{\protect\scriptsize{$\Pi_{2}$}} &
  \hspace*{0mm}\text{\protect\scriptsize{$\Pi_{0}$}} & \hspace*{0mm}\text{\protect\scriptsize{$\Pi_{1}$}} & \hspace*{0mm}\text{\protect\scriptsize{$\Pi_{2}$}} &
  \hspace*{0mm}\text{\protect\scriptsize{$\Pi_{0}$}} & \hspace*{0mm}\text{\protect\scriptsize{$\Pi_{1}$}} & \hspace*{0mm}\text{\protect\scriptsize{$\Pi_{2}$}} \\[-1.5mm]
  \hspace*{2.0mm}$\downbracefill$& \hspace*{-1.3mm}$\downbracefill$& \hspace*{-1.3mm}$\downbracefill$&
  \hspace*{-1.3mm}$\downbracefill$& \hspace*{-1.3mm}$\downbracefill$& \hspace*{-1.3mm}$\downbracefill$&
  \hspace*{-1.3mm}$\downbracefill$& \hspace*{-1.3mm}$\downbracefill$& \hspace*{-1.3mm}$\downbracefill$
  \end{array}
  \nonumber\\[-1mm]
{\normalsize \tilde{\rho}\SP}=  &
  \left(\begin{array}{c|c|c|c|c|c|c|c|c}
    %%%%%%%%%%%%%%%%%%%
    \tikzmarkin[ver=style green, line width=0mm,]{gg00}
        \raisebox{2pt}{\!\protect\scriptsize{\hspace*{1mm}$\Gamma_{00}$\hspace*{1mm}}}
    \tikzmarkend{gg00}
    & \cdot & \cdot &
    0
    & \cdot & \cdot &
    0
    & \cdot & \cdot \\
    \hline
    %%%%%%%%%%%%%%%%%%%
    \cdot &
    \tikzmarkin[ver=style orange, line width=0mm,]{oo00}
        \raisebox{2pt}{\!\protect\scriptsize{\hspace*{1mm}$\Gamma_{01}$\hspace*{1mm}}}
    \tikzmarkend{oo00}
    & \cdot & \cdot &
    0
    & \cdot & \cdot &
    \tikzmarkin[ver=style orange, line width=0mm,]{oo02}
        \phantom{\raisebox{2pt}{\!\protect\scriptsize{\hspace*{1mm}$\Gamma_{01}$\hspace*{1mm}}}}
    \tikzmarkend{oo02}
    & \cdot \\
    \hline
    %%%%%%%%%%%%%%%%%%%
    \cdot & \cdot &
    \tikzmarkin[ver=style cyan, line width=0mm,]{cc00}
        \raisebox{2pt}{\!\protect\scriptsize{\hspace*{1mm}$\Gamma_{02}$\hspace*{1mm}}}
    \tikzmarkend{cc00}
    & \cdot & \cdot &
    \tikzmarkin[ver=style cyan, line width=0mm,]{cc01}
        \phantom{\raisebox{2pt}{\!\protect\scriptsize{\hspace*{1mm}$\Gamma_{02}$\hspace*{1mm}}}}
    \tikzmarkend{cc01}
    & \cdot & \cdot &
    0
    \\
    \hline
    %%%%%%%%%%%%%%%%%%%
    0
    & \cdot & \cdot &
    \tikzmarkin[ver=style green, line width=0mm,]{gg11}
        \raisebox{2pt}{\!\protect\scriptsize{\hspace*{1mm}$\Gamma_{10}$\hspace*{1mm}}}
    \tikzmarkend{gg11}
    & \cdot & \cdot &
    \tikzmarkin[ver=style green, line width=0mm,]{gg12}
        \phantom{\raisebox{2pt}{\!\protect\scriptsize{\hspace*{1mm}$\Gamma_{10}$\hspace*{1mm}}}}
    \tikzmarkend{gg12}
    & \cdot & \cdot \\
    \hline
    %%%%%%%%%%%%%%%%%%%
    \cdot &
    0
    & \cdot & \cdot &
    \tikzmarkin[ver=style orange, line width=0mm,]{oo11}
        \raisebox{2pt}{\!\protect\scriptsize{\hspace*{1mm}$\Gamma_{11}$\hspace*{1mm}}}
    \tikzmarkend{oo11}
    & \cdot & \cdot &
    0
    & \cdot \\
    \hline
    %%%%%%%%%%%%%%%%%%%
    \cdot & \cdot &
    \tikzmarkin[ver=style cyan, line width=0mm,]{cc10}
        \phantom{\raisebox{2pt}{\!\protect\scriptsize{\hspace*{1mm}$\Gamma_{12}$\hspace*{1mm}}}}
    \tikzmarkend{cc10}
    & \cdot & \cdot &
    \tikzmarkin[ver=style cyan, line width=0mm,]{cc11}
        \raisebox{2pt}{\!\protect\scriptsize{\hspace*{1mm}$\Gamma_{12}$\hspace*{1mm}}}
    \tikzmarkend{cc11}
    & \cdot & \cdot &
    0
    \\
    \hline
    %%%%%%%%%%%%%%%%%%%%
    0
    & \cdot & \cdot &
    \tikzmarkin[ver=style green, line width=0mm,]{gg21}
        \phantom{\raisebox{2pt}{\!\protect\scriptsize{\hspace*{1mm}$\Gamma_{20}$\hspace*{1mm}}}}
    \tikzmarkend{gg21}
    & \cdot & \cdot &
    \tikzmarkin[ver=style green, line width=0mm,]{gg22}
        \raisebox{2pt}{\!\protect\scriptsize{\hspace*{1mm}$\Gamma_{20}$\hspace*{1mm}}}
    \tikzmarkend{gg22}
    & \cdot & \cdot \\
    \hline
    %%%%%%%%%%%%%%%%%%%
    \cdot &
    \tikzmarkin[ver=style orange, line width=0mm,]{oo20}
        \phantom{\raisebox{2pt}{\!\protect\scriptsize{\hspace*{1mm}$\Gamma_{21}$\hspace*{1mm}}}}
    \tikzmarkend{oo20}
    & \cdot & \cdot &
    0
    & \cdot & \cdot &
    \tikzmarkin[ver=style orange, line width=0mm,]{oo22}
        \raisebox{2pt}{\!\protect\scriptsize{\hspace*{1mm}$\Gamma_{21}$\hspace*{1mm}}}
    \tikzmarkend{oo22}
    & \cdot \\
    \hline
    %%%%%%%%%%%%%%%%%%%
    \cdot & \cdot &
    0
    & \cdot & \cdot &
    0
    & \cdot & \cdot &
    \tikzmarkin[ver=style cyan, line width=0mm,]{cc22}
        \raisebox{2pt}{\!\protect\scriptsize{\hspace*{1mm}$\Gamma_{22}$\hspace*{1mm}}}
    \tikzmarkend{cc22}
  \end{array}\right),\nonumber\\[-5.5mm]
  &\,
  \begin{array}{ccc}
  \hspace*{24.0mm} & \hspace*{21.5mm} & \hspace*{21.5mm}\\
  \hspace*{2.0mm}$\upbracefill$ & \hspace*{-1.3mm}$\upbracefill$ & \hspace*{-1.3mm}$\upbracefill$\\[-0.5mm]
  \hspace*{4.5mm}\text{\protect\scriptsize{$\ket{0}\Sys$}} & \hspace*{0.7mm}\text{\protect\scriptsize{$\ket{1}\Sys$}} & \hspace*{0.7mm}\text{\protect\scriptsize{$\ket{2}\Sys$}}
  \end{array}
  \nonumber
\end{align}
%\end{tiny}
\vspace*{-9mm}
\begin{align}
\label{eq:unbiased max corr state}
\end{align}
where the remaining blocks in the correlated subspace now satisfy $\tr(\Gamma_{nn})=p_{n}C_{\mathrm{max}}$, see Ref.~\cite[Eq.~(5)]{GuryanovaFriisHuber2018}. From Eq.~(\ref{eq:unbiased conditional post meas state}) we can then conclude that the coefficients of the conditional states for unbiased maximally correlated measurements (UMC) satisfy
\begin{align}
    q_{nn|n} &=\,C_{\mathrm{max}}\,=\, \sum_{i=0}^{d\Poi/d\Sys-1} \frac{ \exp\left(-\beta\Poi E_i\suptiny{0}{0}{({P})}\right)}{\mathcal{Z}\Poi}\qquad \forall\ n,
    \label{eq:Cmax appendix}
\end{align}
where we have assumed that the eigenvalues of the pointer Hamiltonian are ordered non-decreasingly, i.e., $E_i\suptiny{0}{0}{({P})}\geq E_j\suptiny{0}{0}{({P})}$ for $i\geq j$,
while some of the off-diagonals of $\rho_{n}$ are more restricted, $q_{nl|n}=q_{ln|n}=0 \forall l\neq n$. In particular, this last condition implies that the conditional post-measurement state of the system for UMC measurements takes the block-diagonal form
\begin{align}
    \rho_{n}    &=\,C_{\mathrm{max}}\,\ket{n}\!\!\bra{n}
    \,+\,(1-C_{\mathrm{max}})\,\rho_{n}\suptiny{0}{0}{\mathrm{error}},
\end{align}
where $\rho_{n}\suptiny{0}{0}{\mathrm{error}}$ is density operator on the Hilbert space spanned by the vectors $\{\ket{m}\Sys\}$ for $m\in\{0,1,\ldots,d\Sys-1\}/n$ such that $\bra{n}\rho_{n}\suptiny{0}{0}{\mathrm{error}}\ket{l}=\bra{n}\rho_{l}\suptiny{0}{0}{\mathrm{error}}\ket{n}=0\,\forall\,l$.

The maximal correlation value $C_{\mathrm{max}}$ can further be used to quantify the distance between the conditional post-measurement state $\rho_{n}$ of UMC measurements and the pure state $\ket{n}$. Using the trace distance $D(X,Y)=\frac{1}{2}\N{X-Y}_1 \equiv \frac{1}{2} \tr{ \sqrt{(X-Y)(X-Y)^{\dagger}}}$ between two operators $X$ and $Y$, and the block-diagonal structure of $\rho_{n}$ we can calculate
\begin{align}
D(\rho_{n}, \ket{n}\!\!\bra{n})
    &=\,\frac{1}{2}
    \tr\sqrt{\bigl[\rho_{n}-\ket{n}\!\!\bra{n}\bigr]^{2}}
    \,=\,\frac{1}{2}
    \tr\sqrt{
        \bigl[
            (C_{\mathrm{max}}-1)\ket{n}\!\!\bra{n}
            +(1-C_{\mathrm{max}})\rho_{n}\suptiny{0}{0}{\mathrm{error}}
        \bigr]^{2}}
    \nonumber\\[1mm]
    & =\,\frac{1}{2}
    \tr\sqrt{
            (C_{\mathrm{max}}-1)^{2}\ket{n}\!\!\bra{n}
            +(1-C_{\mathrm{max}})^{2}
            (\rho_{n}\suptiny{0}{0}{\mathrm{error}})^{2}
        }
    \,=\,
    \frac{1}{2}\bigl[
    (1-C_{\mathrm{max}})\tr(\ket{n}\!\!\bra{n})
            +(1-C_{\mathrm{max}})\tr(\rho_{n}\suptiny{0}{0}{\mathrm{error}})\bigr]
    \,\nonumber\\[1mm]
    &=\,1 - C_{\mathrm{max}}.
    \label{eq:cmax.dist}
\end{align}

%%%%%%%%%%%%%%%%%%%%%%%%%%%%%%%%%%%%%%%%%%%%%%%%%%%%%%%%%

\subsection{Minimal Energy UMC Measurements}\label{appsec:min energy unbiased max corr measurements}

One can then further restrict the set of considered unitaries $U_{\mathrm{meas}}$ by demanding minimal energy consumption. That is, that the measurement implemented by $U_{\mathrm{meas}}$ achieves the algebraic maximum of correlations \textit{and} spends the least amount of energy as compared with all other unitary operations achieving $C_{\mathrm{max}}$, i.e.,
\begin{align}
    \Delta E_{\mathrm{meas}} &=\, \min_{U_{\mathrm{meas}}}\tr\bigl[(H\Sys + H\Poi)(\tilde{\rho}\SP \,-\, \rho\Sys \otimes \tau\Poi)\bigr]\quad \mathrm{s.t.}\quad C(\tilde{\rho}\SP) \,=\, C_{\mathrm{max}}.
    \label{eq:min-energy}
\end{align}
A requirement in order to spend the minimum amount of energy in the measurement process is not to waste energy on creating coherences~\cite{LostaglioJenningsRudolph2015}. Conversely, any coherence with respect to energy eigenstate with different energies would imply that there exists a unitary transformation that shifts probability from higher energies to lower energies. In other words, such states would not be passive~\cite{PuszWoronowicz1978}. For the sake of illustrating the effect on the form of the final state for our example, let us assume that both the system and pointer Hamiltonians are non-degenerate. Then, the final state for a minimal energy UMC measurement must be of the form
\vspace*{-5mm}
\begin{align}
   &%\hspace*{1mm}
   \begin{array}{ccccccccc}
  \hspace*{5mm} & \hspace*{5mm} & \hspace*{5mm} &
  \hspace*{5mm} & \hspace*{5mm} & \hspace*{5mm} &
  \hspace*{5mm} & \hspace*{5mm} & \hspace*{5mm} \\
  \hspace*{3.5mm}\text{\protect\scriptsize{$\Pi_{0}$}} & \hspace*{0mm}\text{\protect\scriptsize{$\Pi_{1}$}} & \hspace*{0mm}\text{\protect\scriptsize{$\Pi_{2}$}} &
  \hspace*{0mm}\text{\protect\scriptsize{$\Pi_{0}$}} & \hspace*{0mm}\text{\protect\scriptsize{$\Pi_{1}$}} & \hspace*{0mm}\text{\protect\scriptsize{$\Pi_{2}$}} &
  \hspace*{0mm}\text{\protect\scriptsize{$\Pi_{0}$}} & \hspace*{0mm}\text{\protect\scriptsize{$\Pi_{1}$}} & \hspace*{0mm}\text{\protect\scriptsize{$\Pi_{2}$}} \\[-1.5mm]
  \hspace*{2.0mm}$\downbracefill$& \hspace*{-1.3mm}$\downbracefill$& \hspace*{-1.3mm}$\downbracefill$&
  \hspace*{-1.3mm}$\downbracefill$& \hspace*{-1.3mm}$\downbracefill$& \hspace*{-1.3mm}$\downbracefill$&
  \hspace*{-1.3mm}$\downbracefill$& \hspace*{-1.3mm}$\downbracefill$& \hspace*{-1.3mm}$\downbracefill$
  \end{array}
  \nonumber\\[-1mm]
{\normalsize \tilde{\rho}\SP}=  &
  \left(\begin{array}{c|c|c|c|c|c|c|c|c}
    %%%%%%%%%%%%%%%%%%%
    \tikzmarkin[ver=style green, line width=0mm,]{ggg00}
        \raisebox{2pt}{\!\protect\scriptsize{\hspace*{1mm}$\Gamma_{00}$\hspace*{1mm}}}
    \tikzmarkend{ggg00}
    & \cdot & \cdot &
    0
    & \cdot & \cdot &
    0
    & \cdot & \cdot \\
    \hline
    %%%%%%%%%%%%%%%%%%%
    \cdot &
    \tikzmarkin[ver=style orange, line width=0mm,]{ooo00}
        \raisebox{2pt}{\!\protect\scriptsize{\hspace*{1mm}$\Gamma_{01}$\hspace*{1mm}}}
    \tikzmarkend{ooo00}
    & \cdot & \cdot &
    0
    & \cdot & \cdot &
    0
    & \cdot \\
    \hline
    %%%%%%%%%%%%%%%%%%%
    \cdot & \cdot &
    \tikzmarkin[ver=style cyan, line width=0mm,]{ccc00}
        \raisebox{2pt}{\!\protect\scriptsize{\hspace*{1mm}$\Gamma_{02}$\hspace*{1mm}}}
    \tikzmarkend{ccc00}
    & \cdot & \cdot &
    0
    & \cdot & \cdot &
    0
    \\
    \hline
    %%%%%%%%%%%%%%%%%%%
    0
    & \cdot & \cdot &
    \tikzmarkin[ver=style green, line width=0mm,]{ggg11}
        \raisebox{2pt}{\!\protect\scriptsize{\hspace*{1mm}$\Gamma_{10}$\hspace*{1mm}}}
    \tikzmarkend{ggg11}
    & \cdot & \cdot &
    0
    & \cdot & \cdot \\
    \hline
    %%%%%%%%%%%%%%%%%%%
    \cdot &
    0
    & \cdot & \cdot &
    \tikzmarkin[ver=style orange, line width=0mm,]{ooo11}
        \raisebox{2pt}{\!\protect\scriptsize{\hspace*{1mm}$\Gamma_{11}$\hspace*{1mm}}}
    \tikzmarkend{ooo11}
    & \cdot & \cdot &
    0
    & \cdot \\
    \hline
    %%%%%%%%%%%%%%%%%%%
    \cdot & \cdot &
    0
    & \cdot & \cdot &
    \tikzmarkin[ver=style cyan, line width=0mm,]{ccc11}
        \raisebox{2pt}{\!\protect\scriptsize{\hspace*{1mm}$\Gamma_{12}$\hspace*{1mm}}}
    \tikzmarkend{ccc11}
    & \cdot & \cdot &
    0
    \\
    \hline
    %%%%%%%%%%%%%%%%%%%%
    0
    & \cdot & \cdot &
    0
    & \cdot & \cdot &
    \tikzmarkin[ver=style green, line width=0mm,]{ggg22}
        \raisebox{2pt}{\!\protect\scriptsize{\hspace*{1mm}$\Gamma_{20}$\hspace*{1mm}}}
    \tikzmarkend{ggg22}
    & \cdot & \cdot \\
    \hline
    %%%%%%%%%%%%%%%%%%%
    \cdot &
    0
    & \cdot & \cdot &
    0
    & \cdot & \cdot &
    \tikzmarkin[ver=style orange, line width=0mm,]{ooo22}
        \raisebox{2pt}{\!\protect\scriptsize{\hspace*{1mm}$\Gamma_{21}$\hspace*{1mm}}}
    \tikzmarkend{ooo22}
    & \cdot \\
    \hline
    %%%%%%%%%%%%%%%%%%%
    \cdot & \cdot &
    0
    & \cdot & \cdot &
    0
    & \cdot & \cdot &
    \tikzmarkin[ver=style cyan, line width=0mm,]{ccc22}
        \raisebox{2pt}{\!\protect\scriptsize{\hspace*{1mm}$\Gamma_{22}$\hspace*{1mm}}}
    \tikzmarkend{ccc22}
  \end{array}\right),\nonumber\\[-5.5mm]
  &\,
  \begin{array}{ccc}
  \hspace*{24.0mm} & \hspace*{21.5mm} & \hspace*{21.5mm}\\
  \hspace*{2.0mm}$\upbracefill$ & \hspace*{-1.3mm}$\upbracefill$ & \hspace*{-1.3mm}$\upbracefill$\\[-0.5mm]
  \hspace*{4.5mm}\text{\protect\scriptsize{$\ket{0}\Sys$}} & \hspace*{0.7mm}\text{\protect\scriptsize{$\ket{1}\Sys$}} & \hspace*{0.7mm}\text{\protect\scriptsize{$\ket{2}\Sys$}}
  \end{array}
  \nonumber
\end{align}
%\end{tiny}
\vspace*{-9mm}
\begin{align}
\label{eq:min energy unbiased max corr state}
\end{align}
where all remaining diagonal blocks $\Gamma_{mn}$ must now be diagonal matrices themselves, and $\tr(\Gamma_{nn})= p_n C_{\mathrm{max}}\ \forall\,n$ as before. In addition, the energy minimisation imposes constraints on (I) the assignment of the remaining eigenvalues of $\tau\Poi$ to the matrices $\Gamma_{mn}$ for $m\neq n$ in the uncorrelated subspaces, and (II) on the ordering of the eigenvalues within each specific block. What is crucial here is the observation that (I) corresponds to arranging the eigenvalues of $\tau\Poi$ in descending order and splitting the resulting list into three (in general $d\Sys$) sets $a\suptiny{0}{0}{(0)}$, $a\suptiny{0}{0}{(1)}$, and $a\suptiny{0}{0}{(2)}$ that we can interpret as diagonal $(d\Poi/d\Sys)$ matrices with diagonal elements
\begin{align}
    a\suptiny{0}{0}{(n)}_{k}    &=\,\tau_{n\,\frac{d\Poi}{d\Sys}+k}
\end{align}
for $k=1,\ldots,d\Poi/d\Sys$ and $n=0,\ldots,d\Sys-1$. Minimising the energy while maintaining a UMC measurement then means that the set $a\suptiny{0}{0}{(0)}$ of largest entries is assigned to the correlated subspaces, $\Gamma_{nn}=p_{n}\,a\suptiny{0}{0}{(0)}\,\forall\,n$, and the other $a\suptiny{0}{0}{(i)}$ with $i>0$ are assigned to the blocks $\Gamma_{mn}$ with $m\neq n$ such that the lower energies are combined with the larger weights (higher values of $i$).

At this point, it becomes useful to switch to a reduced description of the final state that captures only the diagonal blocks $\Gamma_{ij}$. We therefore define a quantity we refer to as the correlation matrix $\Gamma=(\Gamma_{ij})$ in this context, where, as before, the $\Gamma_{ij}$ are the $(d\Poi/d\Sys)\times(d\Poi/d\Sys)$ matrices corresponding to $\ket{i}\!\!\bra{i}\Sys\otimes\Pi_{j}\,\tilde{\rho}\SP \ket{i}\!\!\bra{i}\Sys\otimes\Pi_{j}$. With this notation, the correlation matrix of the final state for a minimal energy UMC measurement takes the form
\begin{align}
\Gamma &=
\,\left[\begin{array}{c|c|c}
        \tikzmarkin[ver=style green, line width=0mm,]{yd1}\raisebox{2pt}{\protect\footnotesize{
        $\Gamma_{00}$}}\tikzmarkend{yd1}  &
        \tikzmarkin[ver=style orange, line width=0mm,]{ydd2}\raisebox{2pt}{\protect\footnotesize{
        $\Gamma_{01}$}}\tikzmarkend{ydd2}  &
        \tikzmarkin[ver=style cyan, line width=0mm,]{yw1}\raisebox{2pt}{\protect\footnotesize{
        $\Gamma_{02}$}}\tikzmarkend{yw1}  \\
        \hline
        \tikzmarkin[ver=style green, line width=0mm,]{ydd1}\raisebox{2pt}{\protect\footnotesize{$\Gamma_{10}$}}\tikzmarkend{ydd1} & \tikzmarkin[ver=style orange, line width=0mm,]{yd2}\raisebox{2pt}{\protect\footnotesize{$\Gamma_{11}$}}\tikzmarkend{yd2}  & \tikzmarkin[ver=style cyan, line width=0mm,]{yw3}\raisebox{2pt}{\protect\footnotesize{$\Gamma_{12}$}}\tikzmarkend{yw3} \\
        \hline
        \tikzmarkin[ver=style green, line width=0mm,]{yd33}\raisebox{2pt}{\protect\footnotesize{$\Gamma_{20}$}}\tikzmarkend{yd33}  &  \tikzmarkin[ver=style orange, line width=0mm,]{yd44}\raisebox{2pt}{\protect\footnotesize{$\Gamma_{21}$}}\tikzmarkend{yd44} &  \tikzmarkin[ver=style cyan, line width=0mm,]{yd5}\raisebox{2pt}{\protect\footnotesize{$\Gamma_{22}$}}\tikzmarkend{yd5}
    \end{array}\right]
    \,=\,
    \left[\begin{array}{c|c|c}
        \tikzmarkin[ver=style green, line width=0mm,]{xd1}\raisebox{2pt}{\protect\footnotesize{
        $p_{0}\,a\suptiny{0}{0}{(0)}$}}\tikzmarkend{xd1}  &
        \tikzmarkin[ver=style orange, line width=0mm,]{xdd2}\raisebox{2pt}{\protect\footnotesize{
        $p_{1}\,a\suptiny{0}{0}{(1)}$}}\tikzmarkend{xdd2}  &
        \tikzmarkin[ver=style cyan, line width=0mm,]{xw1}\raisebox{2pt}{\protect\footnotesize{
        $p_{2}\,a\suptiny{0}{0}{(1)}$}}\tikzmarkend{xw1}  \\
        \hline
        \tikzmarkin[ver=style green, line width=0mm,]{xdd1}\raisebox{2pt}{\protect\footnotesize{$p_{0}\,a\suptiny{0}{0}{(1)}$}}\tikzmarkend{xdd1} & \tikzmarkin[ver=style orange, line width=0mm,]{xd2}\raisebox{2pt}{\protect\footnotesize{$p_{1}\,a\suptiny{0}{0}{(0)}$}}\tikzmarkend{xd2}  & \tikzmarkin[ver=style cyan, line width=0mm,]{xw3}\raisebox{2pt}{\protect\footnotesize{$p_{2}\,a\suptiny{0}{0}{(2)}$}}\tikzmarkend{xw3} \\
        \hline
        \tikzmarkin[ver=style green, line width=0mm,]{xd33}\raisebox{2pt}{\protect\footnotesize{$p_{0}\,a\suptiny{0}{0}{(2)}$}}\tikzmarkend{xd33}  &  \tikzmarkin[ver=style orange, line width=0mm,]{xd44}\raisebox{2pt}{\protect\footnotesize{$p_{1}\,a\suptiny{0}{0}{(2)}$}}\tikzmarkend{xd44} &  \tikzmarkin[ver=style cyan, line width=0mm,]{xd5}\raisebox{2pt}{\protect\footnotesize{$p_{2}\,a\suptiny{0}{0}{(0)}$}}\tikzmarkend{xd5}
    \end{array}\right].
    \label{eq:newrep}
\end{align}
Alternatively, we may write the correlation matrix elements $\Gamma_{ij}$ as $\Gamma_{ij}=p_{j}\,a\suptiny{0}{0}{(\tilde{\pi}_{ij})}$ where $\tilde{\pi}=(\tilde{\pi}_{ij})$ is a $d\Sys\times d\Sys$ matrix with entries $\tilde{\pi}_{ij}\in\{0,1,\ldots,d\Sys-1\}$. Unbiasedness requires that each element of $\{0,1,\ldots,d\Sys-1\}$ appears exactly once in each column of $\tilde{\pi}$, that is,  $\{\tilde{\pi}_{ij}\}_{i=0,1,\ldots,d\Sys}=\{0,1,\ldots,d\Sys-1\}$. Achieving maximal correlation $C_{\mathrm{max}}$ for an unbiased measurement further fixes the diagonal of $\tilde{\pi}$, i.e., $\tilde{\pi}_{ii}=0\,\forall\,i$. Finally, minimal energy of UMC measurements implies that the off-diagonal elements of $\tilde{\pi}$ are arranged such that the lowest values fill up the top rows (for lowest row index $i$) of $\tilde{\pi}$ first, that is, for a fixed $j$ $\tilde{\pi}_{ij}=i+1$ if $i<j$ and $\tilde{\pi}_{ij}=j$ if $i>j$.
More details on minimal energy UMC measurements can be found in Ref.~\cite[Appendices A.8 and A.9]{GuryanovaFriisHuber2018}.

%%%%%%%%%%%%%%%%%%%%%%%%%%%%%%%%%%%%%%%%%%%%%%%%%%%%%%%%%

\subsection{Minimally Invasive UMC Measurements}\label{appsec:min invasive UMC measurements}

As we have discussed in Sec.~\ref{sec:framework} (following Ref.~\cite{GuryanovaFriisHuber2018}), non-ideal measurement procedures cannot be both unbiased and non-invasive. However, while one cannot construct an unbiased measurement that is non-invasive for all initial system states $\rho\Sys$, one can indeed construct unbiased measurements that leave (the diagonal of) certain system states invariant. One case that is of specific interest here (and goes beyond what is considered in Ref.~\cite{GuryanovaFriisHuber2018}) is the case of UMC measurements, which leave the maximally mixed state $\rho\Sys=\mathds{1}\Sys/d\Sys$ invariant. For this maximum entropy state, all initial probabilities are the same, $p_{n}=1/d\Sys\,\forall\,n$. At the same time, we observe that the diagonal elements $\bra{n}\tr\Poi(\tilde{\rho}\SP)\ket{n}$ of the post-measurement system state are obtained by taking the trace of the sum of the elements of the $n$th row of the correlation matrix $\Gamma$.

We thus observe that such a map, which we call \emph{minimally invasive UMC} measurements, can be realized if each row of $\Gamma$ features each superscript index $i$ of $a\suptiny{0}{0}{(i)}$ exactly once. Together with the unbiasedness requirement we see that unbiased measurements are minimally invasive if and only if $\tilde{\pi}$ is a Latin square. That is, each row and each column of $\tilde{\pi}$ features each element once and only once. It is interesting to note that any minimal energy UMC measurement can thus be turned into a minimally invasive UMC measurement, and vice versa, by rearranging the entries of the $\Gamma$ within its columns only. This implies that the additional energy cost for moving away from the energy minimum depends only on the spectrum of the system Hamiltonian, but not on the pointer Hamiltonian. At the same time, this can be done in such a way that the energy is only minimally increased with respect to the minimal energy UMC measurements, obtaining a \emph{minimal energy minimally invasive UMC measurement}. For the example in $d\Sys=3$ discussed above, the correlation matrix for such a measurement takes the form
\begin{align}
\Gamma &=
\,\left[\begin{array}{c|c|c}
        \tikzmarkin[ver=style green, line width=0mm,]{zyd1}\raisebox{2pt}{\protect\footnotesize{
        $\Gamma_{00}$}}\tikzmarkend{zyd1}  &
        \tikzmarkin[ver=style orange, line width=0mm,]{zydd2}\raisebox{2pt}{\protect\footnotesize{
        $\Gamma_{01}$}}\tikzmarkend{zydd2}  &
        \tikzmarkin[ver=style cyan, line width=0mm,]{zyw1}\raisebox{2pt}{\protect\footnotesize{
        $\Gamma_{02}$}}\tikzmarkend{zyw1}  \\
        \hline
        \tikzmarkin[ver=style green, line width=0mm,]{zydd1}\raisebox{2pt}{\protect\footnotesize{$\Gamma_{10}$}}\tikzmarkend{zydd1} & \tikzmarkin[ver=style orange, line width=0mm,]{zyd2}\raisebox{2pt}{\protect\footnotesize{$\Gamma_{11}$}}\tikzmarkend{zyd2}  & \tikzmarkin[ver=style cyan, line width=0mm,]{zyw3}\raisebox{2pt}{\protect\footnotesize{$\Gamma_{12}$}}\tikzmarkend{zyw3} \\
        \hline
        \tikzmarkin[ver=style green, line width=0mm,]{zyd33}\raisebox{2pt}{\protect\footnotesize{$\Gamma_{20}$}}\tikzmarkend{zyd33}  &  \tikzmarkin[ver=style orange, line width=0mm,]{zyd44}\raisebox{2pt}{\protect\footnotesize{$\Gamma_{21}$}}\tikzmarkend{zyd44} &  \tikzmarkin[ver=style cyan, line width=0mm,]{zyd5}\raisebox{2pt}{\protect\footnotesize{$\Gamma_{22}$}}\tikzmarkend{zyd5}
    \end{array}\right]
    \,=\,
    \left[\begin{array}{c|c|c}
        \tikzmarkin[ver=style green, line width=0mm,]{zxd1}\raisebox{2pt}{\protect\footnotesize{
        $p_{0}\,a\suptiny{0}{0}{(0)}$}}\tikzmarkend{zxd1}  &
        \tikzmarkin[ver=style orange, line width=0mm,]{zxdd2}\raisebox{2pt}{\protect\footnotesize{
        $p_{1}\,a\suptiny{0}{0}{(2)}$}}\tikzmarkend{zxdd2}  &
        \tikzmarkin[ver=style cyan, line width=0mm,]{zxw1}\raisebox{2pt}{\protect\footnotesize{
        $p_{2}\,a\suptiny{0}{0}{(1)}$}}\tikzmarkend{zxw1}  \\
        \hline
        \tikzmarkin[ver=style green, line width=0mm,]{zxdd1}\raisebox{2pt}{\protect\footnotesize{$p_{0}\,a\suptiny{0}{0}{(1)}$}}\tikzmarkend{zxdd1} & \tikzmarkin[ver=style orange, line width=0mm,]{zxd2}\raisebox{2pt}{\protect\footnotesize{$p_{1}\,a\suptiny{0}{0}{(0)}$}}\tikzmarkend{zxd2}  & \tikzmarkin[ver=style cyan, line width=0mm,]{zxw3}\raisebox{2pt}{\protect\footnotesize{$p_{2}\,a\suptiny{0}{0}{(2)}$}}\tikzmarkend{zxw3} \\
        \hline
        \tikzmarkin[ver=style green, line width=0mm,]{zxd33}\raisebox{2pt}{\protect\footnotesize{$p_{0}\,a\suptiny{0}{0}{(2)}$}}\tikzmarkend{zxd33}  &  \tikzmarkin[ver=style orange, line width=0mm,]{zxd44}\raisebox{2pt}{\protect\footnotesize{$p_{1}\,a\suptiny{0}{0}{(1)}$}}\tikzmarkend{zxd44} &  \tikzmarkin[ver=style cyan, line width=0mm,]{zxd5}\raisebox{2pt}{\protect\footnotesize{$p_{2}\,a\suptiny{0}{0}{(0)}$}}\tikzmarkend{zxd5}
    \end{array}\right].
    \label{eq:newrep min invasive}
\end{align}
A property of minimally invasive UMC measurements that follows directly from the form described above is that the conditional probabilities $q_{ll|n}$ for the system to be left in the $l$th energy eigenstate given a pointer outcome $n$ form a doubly stochastic matrix, i.e.,
\begin{align}
    \sum_{l} q_{ll|n} & = \sum_{n} q_{ll|n} \,=\, 1.
   \label{eq.non-inv}
\end{align}

As we will see in Appendix~\ref{sec:Jarjar for non-ideal measurements}, minimally invasive UMC measurements allow satisfying the Jarzynski relation.

%%%%%%%%%%%%%%%%%%%%%%%%%%%%%%%%%%%%%%%%%%%%%%%%%%%%%%%%%%%%%%%%%%%%%%%%%%%%%%%%%%%%%%%%%%

\section{The Non-Ideal TPM Scheme}\label{sec:ttpm}

The estimation of the work performed/extracted during the process $\left(\tau\suptiny{0}{0}{(0)}, H\suptiny{0}{0}{(0)}\right) \xrightarrow{\phantom{b}\Lambda \phantom{b}}\left( \rho\suptiny{0}{0}{(\mathrm{f})} , H\suptiny{0}{0}{(\mathrm{f})}\right)$
can be calculated by means of the so called two projective measurement (TPM) process~\cite{TalknerLutzHaenggi2007, CampisiHaenggiTalkner2011},  which consists of three steps. In the first, one performs a projective measurement on $\tau\suptiny{0}{0}{(0)}$ in the eigenbasis $\{\ket{E_n\suptiny{0}{0}{(0)}}\}_{n=0,\ldots,d\Sys-1}$ of the initial Hamiltonian $H\suptiny{0}{0}{(0)}$, obtaining an outcome $n$. Then one lets the resulting post-measurement system state evolve under the action of $U_{\Lambda}$. Finally, a second projective measurement is performed in the eigenbasis $\{\ket{E_n\suptiny{0}{0}{(\mathrm{f})}}\}_{n=0,\ldots,d\Sys-1}$ of the final Hamiltonian $H\suptiny{0}{0}{(\mathrm{f})}$, resulting in an outcome $m$. The aim of the protocol is to estimate the work performed (or extracted) by the process $\Lambda$ based on the joint probability distribution $p(n,m)$, which itself is estimated by repeating the procedure sufficiently many times. For ideal projective measurements, the joint probability distribution can be written as~\cite{CampisiHaenggiTalkner2011}
\begin{align}
    p(n,m) &=\, p_{n}\suptiny{0}{0}{(0)} \, p_{n \rightarrow m},
    \label{pideal}
\end{align}
where $p_n\suptiny{0}{0}{(0)} = \frac{1}{\mathcal{Z}\suptiny{0}{0}{(0)}} \exp(-\beta E_n\suptiny{0}{0}{(0)})$ are the diagonal elements of the initial thermal state $\tau\suptiny{0}{0}{(0)}$ of the system, $\mathcal{Z}\suptiny{0}{0}{(0)}=\sum_{n}\exp(-\beta E_n\suptiny{0}{0}{(0)})$ is the partition function w.r.t. the initial Hamiltonian and inverse temperature $\beta$. Note that the inverse system temperature $\beta$ need not match the initial inverse temperature $\beta\Poi$ of the pointer, which enters $C_{\mathrm{max}}$ from Eq.~(\ref{eq:Cmax appendix}). The symbol $p_{\raisebox{-1pt}{\scriptsize{$n\!\rightarrow\!m$}}}=|\bra{E_{m}\suptiny{0}{0}{(\mathrm{f})}}U_{\Lambda}\ket{E_{n}\suptiny{0}{0}{(0)}}|^{2}$ denotes the transition probability from an initial energy eigenstate $\ket{E_n\suptiny{0}{0}{(0)}}$ to a final energy eigenstate $\ket{E_m\suptiny{0}{0}{(\mathrm{f})}}$. The work distribution, i.e., the probability density for performing (or extracting) the work $W$ given the process $\Lambda$ is defined as
\begin{align}
    P(W) & = \sum_{m,n} p(n,m) \,\delta\ch{E_m\suptiny{0}{0}{(\mathrm{f})} - E_n\suptiny{0}{0}{(0)} - W}.
    \label{eq:pw.ideal}
\end{align}
On average the work spent or extracted by the process $\Lambda$ is then obtained by integration, i.e.,
\begin{align}
    \mean{W} &=\, \int P(W)\, W\, dW,
    \label{eq:meanw.ideal}
\end{align}
where the integral is taken over all possible values of work. The average work can be written in terms of the joint probability $p(n,m)$ by inserting Eq.~\eqref{eq:pw.ideal} into Eq.~\eqref{eq:meanw.ideal}, resulting in the expression
\begin{align}
    \mean{W} &=\, \sum_{m,n} p(n,m) \ch{E_{m}\suptiny{0}{0}{(\mathrm{f})} - E_{n}\suptiny{0}{0}{(0)}}.
\end{align}

Let us now consider the TPM scheme when the ideal projective measurements are replaced by non-ideal measurements, more specifically, minimal energy UMC measurements as described in Appendix~\ref{appendix:Ideal and Non-Ideal Measurements}. We discuss each step of the process in detail below.

%%%%%%%%%%%%%%%%%%%%%%%%%%%%%%%%%%%%%%%%%%%%%%%%%%%%%%%%%%%%%%%%%%%%%%%%%%%%%%%%%%%%%%%%%%

\subsection{First Measurement}

First, note that the initial system state for the TPM scheme is a Gibbs equilibrium state at the ambient temperature. That is, here the system state to be measured during the first measurement is $\rho\Sys=\tau\suptiny{0}{0}{(0)}(\beta)$, given by
\begin{align}
    \tau\suptiny{0}{0}{(0)}(\beta) &=\,   \frac{1}{\mathcal{Z}\suptiny{0}{0}{(0)}} \exp(-\beta H\suptiny{0}{0}{(0)}) \,=\,\frac{1}{\mathcal{Z}\suptiny{0}{0}{(0)}} \sum_{n} \exp(-\beta E_n\suptiny{0}{0}{(0)})\ketbra{E_n\suptiny{0}{0}{(0)}}{E_n\suptiny{0}{0}{(0)}},
\end{align}
with $H\suptiny{0}{0}{(0)} = \sum_n E_n\suptiny{0}{0}{(0)}\ketbra{E_n\suptiny{0}{0}{(0)}}{E_n\suptiny{0}{0}{(0)}}$ at time $t_0$. Then, we assume that the first non-ideal measurement is performed in the eigenbasis of $H\suptiny{0}{0}{(0)}$. We assume that this measurement, though non-ideal, is unbiased [property~(\ref{item:unbiased main}) described in Sec.~\ref{sec:ideal measurements}], such that a measurement result $n$ is obtained with probability $p_{n}\suptiny{0}{0}{(0)}=\bra{E_n\suptiny{0}{0}{(0)}}\tau\suptiny{0}{0}{(0)}\ket{E_n\suptiny{0}{0}{(0)}}\,=\,\exp(-\beta E_n\suptiny{0}{0}{(0)})/\mathcal{Z}\suptiny{0}{0}{(0)}$. Moreover, from here on, we restrict our investigation to non-ideal unbiased measurements that achieve maximal correlation $C_{\mathrm{max}}$ (UMC measurements) and either have \emph{minimal energy} or are \emph{minimal energy minimally invasive measurements}, as described in Appendices~\ref{appsec:min energy unbiased max corr measurements} and~\ref{appsec:min invasive UMC measurements}, respectively, such that the post-measurement state of the system is diagonal w.r.t. the measurement basis, i.e.,
\begin{align}
    \rho_{n}\suptiny{0}{0}{(0)} &= \sum\limits_{l} q\suptiny{0}{0}{(0)}_{ll|n}  \ketbra{E_l\suptiny{0}{0}{(0)}}{E_l\suptiny{0}{0}{(0)}}
    \quad \forall \,n\,,
    \label{rhon}
\end{align}
where $q\suptiny{0}{0}{(0)}_{ll|n}$ is the the probability to find the post-measurement system in the state $\ket{E_l\suptiny{0}{0}{(0)}}$ given the measurement result $n$. In particular, the probability to correctly guess the energy eigenstate $\ket{E_n\suptiny{0}{0}{(0)}}$, when the pointer shows $n$ is given by $q\suptiny{0}{0}{(0)}_{nn|n}$, which has the value
\begin{align}
    q\suptiny{0}{0}{(0)}_{nn|n} &=\,C_{\mathrm{max}}
    \label{eq.cmax}
\end{align}
for UMC measurements. The conditional probabilities $q\suptiny{0}{0}{(0)}_{ll|n}$ are functions of the inverse temperature $\beta$, and the pointer energies (since the initial pointer state is also a thermal state at ambient temperature $T=1/k\subtiny{0}{0}{\mathrm{B}} \beta$), but because the measurement is unbiased, the conditional state $\rho_{n}\suptiny{0}{0}{(0)}$ is independent of the system state $\tau\suptiny{0}{0}{(0)}$ before the measurement.

%%%%%%%%%%%%%%%%%%%%%%%%%%%%%%%%%%%%%%%%%%%%%%%%%%%%%%%%%%%%%%%%%%%%%%%%%%%%%%%%%%%%%%%%%%

\subsection{Evolution and Second Measurement}

After the initial non-ideal measurement, the post-measurement conditional state $\rho_{n}\suptiny{0}{0}{(0)}$ is evolving according to the unitary $U_{\Lambda}$ corresponding to the process $\Lambda$. The unitary $U_{\Lambda}$ can be written as time evolution from an initial time $t_{0}$ to a final time $t_{\mathrm{f}}$,
$U_{\Lambda} =  \mathcal{T}_{+}\exp\left(-i\int_{t_0}^{t_{\mathrm{f}}}\! H(\lambda_{t})\right)\,dt$,
where $\mathcal{T}_{+}$ denotes time-ordering, and $\lambda_{t}$ is a control parameter for the time-dependent Hamiltonian such that $H(\lambda_{t_{0}})=H\suptiny{0}{0}{(0)}$ and $H(\lambda_{t_{\mathrm{f}}})=H\suptiny{0}{0}{(\mathrm{f})}$. The time-evolved state prior to the second measurement is thus
\begin{align}
    \rho_{n}\suptiny{0}{0}{\Lambda} &=\, U_{\Lambda} \rho_n\suptiny{0}{0}{(0)} U^{\dagger}_{\Lambda} \,=\,
    \sum\limits_{l} q\suptiny{0}{0}{(0)}_{ll|n}
    U_{\Lambda}\ketbra{E_{l}\suptiny{0}{0}{(0)}}{E_{l}\suptiny{0}{0}{(0)}} U_{\Lambda}^{\dagger}\,.
    \label{ev1}
\end{align}

Finally, a second non-ideal measurement is performed on $\rho_{n}\suptiny{0}{0}{\Lambda}$ w.r.t. the eigenbasis $\{\ket{E_m\suptiny{0}{0}{(\mathrm{f})}}\}$ of $H\suptiny{0}{0}{(\mathrm{f})}$. Once again, since the measurement is unbiased, the probability to obtain any particular outcome $m$ only depends on the specific system state prior to the measurement. However, in this case this state is $\rho_{n}\suptiny{0}{0}{\Lambda}$ from Eq.~(\ref{ev1}) and hence depends on the first outcome $n$. The conditional probability to obtain outcome $m$ in the second measurement given outcome $n$ in the first measurement is thus
\begin{align}
    p(m|n)  &=\,\bra{E_m\suptiny{0}{0}{(\mathrm{f})}}\rho_{n}\suptiny{0}{0}{\Lambda}\ket{E_m\suptiny{0}{0}{(\mathrm{f})}}
    \,=\,
    \sum\limits_{l} q\suptiny{0}{0}{(0)}_{ll|n}
    \bra{E_m\suptiny{0}{0}{(\mathrm{f})}} U_{\Lambda} \ketbra{E_{l}\suptiny{0}{0}{(0)}}{E_{l}\suptiny{0}{0}{(0)}} U_{\Lambda}^{\dagger} \ket{E_m\suptiny{0}{0}{(\mathrm{f})}}
    \,=\,
    \sum\limits_{l}  q\suptiny{0}{-1}{(0)}_{ll|n}\,p_{\raisebox{-1pt}{\scriptsize{$l\!\rightarrow\!m$}}},
    \label{eq:cond prob sec meas}
\end{align}
where $p_{\raisebox{-1pt}{\scriptsize{$l\!\rightarrow\!m$}}} = |\bra{E_{m}\suptiny{0}{0}{(\mathrm{f})}}U_{\Lambda}\ket{E_{l}\suptiny{0}{0}{(0)}}|^{2}$. Consequently, the joint probability $p(n,m)$ is obtained by multiplying with the probability to obtain outcome $n$ in the first measurement, i.e.,
\begin{align}
    p(n,m)  &=\,p(m|n) \,p_{n}\suptiny{0}{0}{(0)}\,=\,
    p_{n}\suptiny{0}{0}{(0)}\sum\limits_{l}  q\suptiny{0}{-1}{(0)}_{ll|n}\,p_{\raisebox{-1pt}{\scriptsize{$l\!\rightarrow\!m$}}}.
    \label{eq: joint prob}
\end{align}
Meanwhile, the post-measurement state after the second measurement, conditioned on having obtained outcome $n$ and $m$ in the first and second measurement, respectively, is
\begin{align}
    \rho\suptiny{0}{-1}{(\mathrm{f})}_{m|n}
    &= \sum\limits_{k} q\suptiny{0}{0}{(\mathrm{f})}_{kk|m}  \ketbra{E_k\suptiny{0}{0}{(\mathrm{f})}}{E_k\suptiny{0}{0}{(\mathrm{f})}}
    \quad \forall \,m\,.
    \label{eq:cond state 2nd meas}
\end{align}
As in the first measurement, the conditional post-measurement state of the system is independent of the pre-measurement system state. In this case, this implies that $\rho\suptiny{0}{-1}{(\mathrm{f})}_{m|n}$ is indeed independent of the outcome $n$ of the first measurement.

%%%%%%%%%%%%%%%%%%%%%%%%%%%%%%%%%%%%%%%%%%%%%%%%%%%%%%%%%%%%%%%%%%%%%%%%%%%%%%%%%%%%%%%%%%

\section{Work Estimation in the Non-Ideal TPM Scheme}\label{sec:work estimation nonid tpm}

\subsection{Non-Ideal Work Distribution and Estimate}\label{sec:work nonid tpm} %%

To estimate the work in the TPM scheme based on non-ideal measurements (as before, we assume minimal energy UMC measurements as described in Appendix~\ref{appendix:Ideal and Non-Ideal Measurements}), some energy must be spent on the measurement processes, where the precision in the estimate is directly dependent on the energy spent. Considering the joint probability in Eq.~\eqref{eq: joint prob} and the definition in Eq.~\eqref{eq:pw.ideal}, the probability distribution for inferring the work value $W$ is
\begin{align}
    P(W) &= \sum_{m,n} p_n\suptiny{0}{0}{(0)} \sum_l  q_{ll|n}\suptiny{0}{0}{(0)} |\bra{E_m\suptiny{0}{0}{(\mathrm{f})}}U_{\Lambda}\ket{E_l\suptiny{0}{0}{(0)}}|^{2}\, \delta\ch{(E_m\suptiny{0}{0}{(\mathrm{f})} - E_n\suptiny{0}{0}{(0)}) - W}.
\end{align}
We can now calculate the average work,
\begin{align}\label{nworknice}
\mean{W}_{\mathrm{non-id}} &= \int dW P(W) W %\\
= \sum_{m,n} p_n\suptiny{0}{0}{(0)} \sum_l  q_{ll|n}\suptiny{0}{0}{(0)} |\bra{E_m\suptiny{0}{0}{(\mathrm{f})}}U_{\Lambda}\ket{E_l\suptiny{0}{0}{(0)}}|^2 \ch{E_m\suptiny{0}{0}{(\mathrm{f})} - E_n\suptiny{0}{0}{(0)} },\\[1mm]
&=  \sum_{m,n}  q_{nn|n}\suptiny{0}{0}{(0)} \,p_n\suptiny{0}{0}{(0)} |\bra{E_m\suptiny{0}{0}{(\mathrm{f})}}U_{\Lambda}\ket{E_n\suptiny{0}{0}{(0)}}|^2 \ch{E_m\suptiny{0}{0}{(\mathrm{f})} - E_n\suptiny{0}{0}{(0)} } %\\ &
+ \sum_{m,n} \sum_{l\neq n } q_{ll|n}\suptiny{0}{0}{(0)} \,p_n\suptiny{0}{0}{(0)} |\bra{E_m\suptiny{0}{0}{(\mathrm{f})}}U_{\Lambda}\ket{E_l\suptiny{0}{0}{(0)}}|^2 \ch{E_m\suptiny{0}{0}{(\mathrm{f})} - E_n\suptiny{0}{0}{(0)} }.
\nonumber
\end{align}
Substituting for $C_{\text{max}}$ in Eq.~\eqref{eq.cmax}, we arrive at
\begin{align}
\mean{W}_{\mathrm{non-id}}  &= C_{\text{max}}  \sum_{m,n} p_n\suptiny{0}{0}{(0)} |\bra{E_m\suptiny{0}{0}{(\mathrm{f})}}U_{\Lambda}\ket{E_n\suptiny{0}{0}{(0)}}|^2 \ch{E_m\suptiny{0}{0}{(\mathrm{f})} - E_n\suptiny{0}{0}{(0)} } + \sum_{m,n} \sum_{l\neq n } q_{ll|n}\suptiny{0}{0}{(0)}\, p_n\suptiny{0}{0}{(0)} |\bra{E_m\suptiny{0}{0}{(\mathrm{f})}}U_{\Lambda}\ket{E_l\suptiny{0}{0}{(0)}}|^2 \ch{E_m\suptiny{0}{0}{(\mathrm{f})} - E_n\suptiny{0}{0}{(0)} }.
\end{align}
Since $\mean{W}_{\Lambda}= \sum_{m,n} p_{n}\suptiny{0}{0}{(0)} |\bra{E_m\suptiny{0}{0}{(\mathrm{f})}}U_{\Lambda}\ket{E_{n}\suptiny{0}{0}{(0)}}|^2 \ch{E_m\suptiny{0}{0}{(\mathrm{f})} - E_n\suptiny{0}{0}{(0)} }$ we can write the non-ideal work estimate as
\begin{align}\label{eq.non-ideal work}
    \mean{W}_{\mathrm{non-id}}  &= C_{\mathrm{max}}\,\mean{W}_{\Lambda}
    \,+\, \sum_{m,n} \sum_{l\neq n} q_{ll|n}\suptiny{0}{0}{(0)}\, p_n\suptiny{0}{0}{(0)} |\bra{E_m\suptiny{0}{0}{(\mathrm{f})}}U_{\Lambda}\ket{E_l\suptiny{0}{0}{(0)}}|^2 \ch{E_m\suptiny{0}{0}{(\mathrm{f})} - E_n\suptiny{0}{0}{(0)}}.
\end{align}

%%%%%%%%%%%%%%%%%%%%%%%%%%%%%%%%%%%%%%%%%%%%%%%%%%%

\subsection{Ideal versus Non-Ideal Work Estimate}\label{sec:id vs nonid tpm}

Before moving on, let us discuss more about the difference between the work estimates based on ideal and non-ideal (UMC) measurements. Noting that $\mean{W}_{\Lambda}=\Delta E_{\Lambda}= \tr(\U{\tau\suptiny{0}{0}{(0)}}\, H\suptiny{0}{0}{(\mathrm{f})}) - \tr(\tau\suptiny{0}{0}{(0)}H\suptiny{0}{0}{(0)})$ and $\tr(\tau\suptiny{0}{0}{(0)}H\suptiny{0}{0}{(0)})=\sum_{n}p_n\suptiny{0}{0}{(0)}E_n\suptiny{0}{0}{(0)}$, we can use the expression of $\mean{W}_{\mathrm{non-id}}$ from Eq.~\eqref{nworknice} to write
\begin{align}
    \mean{W}_{\mathrm{non-id}} - \mean{W}_{\Lambda} &= \sum_{m,n} p_n\suptiny{0}{0}{(0)} \sum_l  q_{ll|n}\suptiny{0}{0}{(0)}\,  |\bra{E_m\suptiny{0}{0}{(\mathrm{f})}}U\ket{E_l\suptiny{0}{0}{(0)}}|^2E_m\suptiny{0}{0}{(\mathrm{f})}  \,-\,   \tr(\U{\tau\suptiny{0}{0}{(0)}} \,H\suptiny{0}{0}{(\mathrm{f})}).
    \label{eq:something1}
\end{align}
The first term on the right-hand side of Eq.~(\ref{eq:something1}) can be rewritten as
\begin{align}
    \sum_{m,n} p_n\suptiny{0}{0}{(0)} \sum_l  q_{ll|n}\suptiny{0}{0}{(0)} |\bra{E_m\suptiny{0}{0}{(\mathrm{f})}}U_{\Lambda}\ket{E_l\suptiny{0}{0}{(0)}}|^2E_m\suptiny{0}{0}{(\mathrm{f})}    &= \tr\Bigl( \sum_{n} p_n\suptiny{0}{0}{(0)} \sum_l  q_{ll|n}\suptiny{0}{0}{(0)}\U{\ketbra{E_l\suptiny{0}{0}{(0)}}{E_l\suptiny{0}{0}{(0)}}}\,H\suptiny{0}{0}{(\mathrm{f})}   \Bigr)
    \,=\,
    \tr\Bigl( U_{\Lambda}\,\tilde \rho_S\suptiny{0}{0}{(0)}\,
    U_{\Lambda}^{\dagger}\,H\suptiny{0}{0}{(\mathrm{f})}\Bigr),
\end{align}
where we have denoted the unconditional post-measurement state after the first measurement as $\tilde \rho \suptiny{0}{0}{(0)} := \sum_n p_{n}\suptiny{0}{0}{(0)}  \rho_{n}\suptiny{0}{0}{(0)}$ and we have used $\rho_{n}\suptiny{0}{0}{(0)} =  \sum_l q\suptiny{0}{-1}{(0)}_{ll|n}\ketbra{E_{l}\suptiny{0}{0}{(0)}}{E_{l}\suptiny{0}{0}{(0)}}$. With this, we can rewrite Eq.~(\ref{eq:something1}) as
\begin{align}\label{eq:work_dif}
    \mean{W}_{\mathrm{non-id}} - \mean{W}_{\Lambda} &=\, \tr\ch{\U{\bigl(\tilde{\rho}\suptiny{0}{0}{(0)} - \tau\suptiny{0}{0}{(0)}\bigr)} \,H\suptiny{0}{0}{(\mathrm{f})}}.
\end{align}
If $T\rightarrow 0$, then $\tilde{\rho}\suptiny{0}{0}{(0)}=  \tau\suptiny{0}{0}{(0)}$ for any process $\Lambda$, which means that $\mean{W}_{\mathrm{non-id}} =\mean{W}_{\Lambda}$, i.e., one obtains ideal projective measurements. On the other hand if $T$ is nonzero, $\mean{W}_{\mathrm{non-id}} =\mean{W}_{\Lambda}$ can be achieved only for  specific processes ${\Lambda}$. For example, for the unitary $U_{\Lambda}$ of the form
\begin{equation}
U_{\Lambda} = \frac{1}{\sqrt{d\Sys}}\sum_{j,k}e^{-\frac{2\pi i}{d\Sys}j\nr k}\ketbra{E_{k}\suptiny{0}{0}{(0)}}{E_{j}\suptiny{0}{0}{(0)}},
\end{equation}
and $[H\suptiny{0}{0}{(0)},H\suptiny{0}{0}{(\mathrm{f})}]=0$ we have $|\bra{E_m\suptiny{0}{0}{(\mathrm{f})}}U_{\Lambda}\ket{E_l\suptiny{0}{0}{(0)}}|^2 = \frac{1}{{d_{\mathrm{S}}}}$. It is then simple to check that
\begin{equation}\label{eq: condition w non-id equal w ideal}
    \tr\ch{\U{\tilde{\rho}\suptiny{0}{0}{(0)}}\,H\suptiny{0}{0}{(\mathrm{f})}} =  \tr\ch{\U{\tau\suptiny{0}{0}{(0)}} \,H\suptiny{0}{0}{(\mathrm{f})}} \, = \, \frac{1}{d\Sys}\sum_{m}E_{m}\suptiny{0}{0}{(\mathrm{f})}.
\end{equation}
Consequently, one finds $\mean{W}_{\mathrm{non-id}}=\mean{W}_{\Lambda}$. The intuition behind this example is the following. The disturbances due to non-ideal UMC measurements do not generate coherences (off-diagonal elements) between states with different energies in the initial thermal state. Consequently, the unconditional state upon which the process $\Lambda$ acts is diagonal in the energy eigenbasis. For any such state, the process chosen in our example results in a state whose diagonal elements (w.r.t. the energy eigenbasis) are all equal to $1/d\Sys$. As a result, the average energy at the end of the protocol is independent of the disturbance induced by the first measurement. We thus see that there exist processes such that  $\mean{W}_{\Lambda}$ can be precisely estimated independently of the temperature of the pointer.\\

Let us now bound the deviation of the non-ideal work estimate $\mean{W}_{\mathrm{non-id}}$ from the ideal estimate $\mean{W}_{\Lambda}$. Starting from the expression in Eq.~\eqref{eq:work_dif}, we can use H{\"o}lder's inequality $|\tr(XY^{\dagger})|\leq \N{X}_{1}\N{Y}_{\infty}$ to write
\begin{align}
    |\mean{W}_{\mathrm{non-id}} - \mean{W}_{\Lambda}| &\leq\,
    \N{\U{\bigl(\tilde{\rho}\suptiny{0}{0}{(0)} - \tau\suptiny{0}{0}{(0)}\bigr)}}_{1}\,
    \N{H\suptiny{0}{0}{(\mathrm{f})}}_{\infty}
    \,=\,
    \N{\tilde{\rho}\suptiny{0}{0}{(0)} - \tau\suptiny{0}{0}{(0)}}_{1}\,
    \N{H\suptiny{0}{0}{(\mathrm{f})}}_{\infty},
\end{align}
where in last step we have applied the trace invariance under unitary. For details about the properties of the trace distance we refer the reader to, e.g., Ref.~\cite{Watrous2018}. Inserting $\tilde{\rho}\suptiny{0}{0}{(0)} := \sum_n p_{n}\suptiny{0}{0}{(0)}  \rho_{n}\suptiny{0}{0}{(0)}$ and $\tau\suptiny{0}{0}{(0)}=\sum_n p_{n}\suptiny{0}{0}{(0)}\ket{E_n\suptiny{0}{0}{(0)}}\!\!\bra{E_n\suptiny{0}{0}{(0)}}$, we can further write
\begin{align}
    \N{\tilde{\rho}\suptiny{0}{0}{(0)} - \tau\suptiny{0}{0}{(0)}}_{1} &=\,
    \bigl\Vert\sum\limits_{n}
    p_{n}\suptiny{0}{0}{(0)}
    \bigl(\rho_{n}\suptiny{0}{0}{(0)}
    -\ket{E_n\suptiny{0}{0}{(0)}}\!\!\bra{E_n\suptiny{0}{0}{(0)}}\bigr)
    \bigr\Vert_{1}
    \,\leq\,
    \sum\limits_{n}
    \bigl\Vert
    p_{n}\suptiny{0}{0}{(0)}
    \bigl(\rho_{n}\suptiny{0}{0}{(0)}
    -\ket{E_n\suptiny{0}{0}{(0)}}\!\!\bra{E_n\suptiny{0}{0}{(0)}}\bigr)
    \bigr\Vert_{1}
    \,=\,
    \sum\limits_{n} p_{n}\suptiny{0}{0}{(0)}
    \bigl\Vert
    \rho_{n}\suptiny{0}{0}{(0)}
    -\ket{E_n\suptiny{0}{0}{(0)}}\!\!\bra{E_n\suptiny{0}{0}{(0)}}
    \bigr\Vert_{1},
\end{align}
where we have used the triangle inequality in the second step. Finally, we recall that the $1$-norm coincides with the trace norm for the operators we consider here, and we insert from Eq.~\eqref{eq:cmax.dist} before using $\sum_{n} p_{n}\suptiny{0}{0}{(0)}=1$ to arrive at
\begin{align}
    |\mean{W}_{\mathrm{non-id}} - \mean{W}_{\Lambda}| &\leq\,
    (1-C_{\mathrm{max}})\,
    \N{H\suptiny{0}{0}{(\mathrm{f})}}_{\infty}.
\end{align}

%%%%%%%%%%%%%%%%%%%%%%%%%%%%%%%%%%%%%%%%%%%%%%%%%%%%%%%%%%%%%%%%%%%%%%%%%%%%%%%%%%%%%%%%%%

\section{Energy Variation in Non-Ideal Work Estimation}\label{sec:change of average energy}

In the scenario in which ideal projective measurements are not assumed, the change in the average energy of the system in the TPM process is also different from the work estimate. Once again, we therefore consider non-ideal minimal energy UMC measurements as described in Appendix~\ref{appendix:Ideal and Non-Ideal Measurements}. To express the work estimate in this case, let us first define the difference between the energy $\Delta E_{\mathrm{non-id}}$ at the beginning and at the end of the TPM process
\begin{align}
    \Delta E_{\mathrm{non-id}}    \,=\, \sum\limits_{k} p_{k}\suptiny{0}{-1}{(\mathrm{f})}\, E_k\suptiny{0}{0}{(\mathrm{f})}
    - \sum\limits_{n} p_{n}\suptiny{0}{0}{(0)} E_{n}\suptiny{0}{0}{(0)},
\end{align}
where $p_{k}\suptiny{0}{-1}{(\mathrm{f})}$ is the probability to find the post-measurement system after the second measurement in the eigenstate $\ket{E_k\suptiny{0}{0}{(\mathrm{f})}}$ of $H\suptiny{0}{0}{(\mathrm{f})}$. This probability can be expressed as
\begin{align}
    p_{k}\suptiny{0}{-1}{(\mathrm{f})} &=\sum\limits_{m,n}
    q\suptiny{0}{-1}{(\mathrm{f})}_{kk|m}\,p(m|n) \,p_{n}\suptiny{0}{0}{(0)},
\end{align}
where $q\suptiny{0}{-1}{(\mathrm{f})}_{kk|m}$ is the conditional probability to find the system in the final energy eigenstate $\ket{E_k\suptiny{0}{0}{(\mathrm{f})}}$ given the measurement result $m$.
The conditional probability  $p(m|n) = \bra{E_m\suptiny{0}{0}{(\mathrm{f})}} \U{\,\rho_{n}\suptiny{0}{0}{(0)}\,}\ket{E_m\suptiny{0}{0}{(\mathrm{f})}}$ is the probability
to obtain outcome {$m$} in the second measurement given that the result of the first measurement was $n$, {as in Eq.~\eqref{ev1}}. Collecting all these expressions, we can rewrite the average energy difference between the initial system state and the system state after the non-ideal TPM scheme as
\begin{align}
    \Delta E_{\mathrm{non-id}}    &=\,\sum_{m,n,k}
    q\suptiny{0}{-1}{(\mathrm{f})}_{kk|m}\,p(m|n)\,p_{n}\suptiny{0}{0}{(0)}\,
    \bigl(E_{k}\suptiny{0}{0}{(\mathrm{f})}-E_{n}\suptiny{0}{0}{(0)}\bigr)
    \, =\, \sum\limits_{m,n} p_n\suptiny{0}{0}{(0)} \sum_{k,l} q\suptiny{0}{-1}{(0)}_{ll|n}
    q\suptiny{0}{-1}{(\mathrm{f})}_{kk|m}  |\bra{E_{m}\suptiny{0}{0}{(\mathrm{f})}} U_{\Lambda}\ket{E_{l}\suptiny{0}{0}{(0)}}|^{2}
    \bigl(E_{k}\suptiny{0}{0}{(\mathrm{f})}-E_{n}\suptiny{0}{0}{(0)}\bigr).
\end{align}
{We then split this expression into several sums, where the first collects all terms for which $n=l$ and $m=k$, i.e.,}
\begin{align}
    \Delta E_{\mathrm{non-id}}    &=\,
    \sum\limits_{m,n} p_{n}\suptiny{0}{0}{(0)}
    \bigl[
        q\suptiny{0}{-1}{(0)}_{nn|n}\,
        |\bra{E_{m}\suptiny{0}{0}{(\mathrm{f})}} U_{\Lambda}\ket{E_{n}\suptiny{0}{0}{(0)}}|^{2}
        \,+\,
        \sum\limits_{l\neq n} q\suptiny{0}{-1}{(0)}_{ll|n}\,
        |\bra{E_{m}\suptiny{0}{0}{(\mathrm{f})}} U_{\Lambda}\ket{E_{l}\suptiny{0}{0}{(0)}}|^{2}
    \bigr]
    \bigl[
        q\suptiny{0}{-1}{(\mathrm{f})}_{mm|m}\,
        \bigl(E_{m}\suptiny{0}{0}{(\mathrm{f})}\!-\!E_{n}\suptiny{0}{0}{(0)}\bigr)
        \,+\,
        \sum\limits_{k\neq m} q\suptiny{0}{-1}{(\mathrm{f})}_{kk|m}\,
        \bigl(E_{k}\suptiny{0}{0}{(\mathrm{f})}\!-\!E_{n}\suptiny{0}{0}{(0)}\bigr)
    \bigr]\nonumber\\[1mm]
    &=\,
    C_{\mathrm{max}}^{2}\,\mean{W}_{\Lambda}\,
    +\,C_{\mathrm{max}}\sum\limits_{m,n} p_{n}\suptiny{0}{0}{(0)}
    \bigl[
        \sum\limits_{k\neq m}q\suptiny{0}{-1}{(\mathrm{f})}_{kk|m}
        p_{\raisebox{-1pt}{\scriptsize{$n\!\rightarrow\!m$}}}
        \bigl(E_{k}\suptiny{0}{0}{(\mathrm{f})}\!-\!E_{n}\suptiny{0}{0}{(0)}\bigr)
    +
        \sum\limits_{l\neq n}q\suptiny{0}{-1}{(0)}_{ll|n}
        p_{\raisebox{-1pt}{\scriptsize{$l\!\rightarrow\!m$}}}
        \bigl(E_{m}\suptiny{0}{0}{(\mathrm{f})}\!-\!E_{n}\suptiny{0}{0}{(0)}\bigr)
    \bigr]\nonumber\\[1mm]
    &\ \ +\!\!\sum\limits_{\substack{m,n\\ l\neq n\\ k\neq m}}
        p_{n}\suptiny{0}{0}{(0)}\,
        q\suptiny{0}{-1}{(0)}_{ll|n}\,
        p_{\raisebox{-1pt}{\scriptsize{$l\!\rightarrow\!m$}}}\,
        q\suptiny{0}{-1}{(\mathrm{f})}_{kk|m}\,
        \bigl(E_{k}\suptiny{0}{0}{(\mathrm{f})}\!-\!E_{n}\suptiny{0}{0}{(0)}\bigr)\,.
\end{align}
where we have used the shorthand $p_{\raisebox{-1pt}{\scriptsize{$l\!\rightarrow\!m$}}}
=|\bra{E_{m}\suptiny{0}{0}{(\mathrm{f})}} U_{\Lambda}\ket{E_{l}\suptiny{0}{0}{(0)}}|^{2}$ and we note that $q\suptiny{0}{-1}{(0)}_{nn|n} = q\suptiny{0}{-1}{(\mathrm{f})}_{mm|m}=C_{\mathrm{max}}$ for all $n$ and $m$, since we are considering UMC measurements. For ideal projective measurements, $C_{\mathrm{max}}\rightarrow 1$ and consequently  $q\suptiny{0}{-1}{(0)}_{ll|n} = \delta_{l,n}$, and therefore one can notice that  $\Delta E_{\mathrm{non-id}}(C_{\mathrm{max}}\rightarrow 1) =\mean{W}_{\Lambda}$.

%%%%%%%%%%%%%%%%%%%%%%%%%%%%%%%%%%%%%%%%%%%%%%%%%%%%%%%%%%%%%%%%%%%%%%%%%%%%%%%%%%%%%%%%%

\section{Jarzynski Equality for Non-Ideal Projective Measurements}\label{sec:Jarjar for non-ideal measurements}

\subsection{Characteristic Function and Jarzynski Equality}

In~\cite{CampisiHaenggiTalkner2011} it was shown that by taking the Fourier transform of the work probability distribution $P(W)$ one can define a following characteristic function, which we parametrise by $u$
\begin{align}
G(u) &= \int \text{d}W P(W) \exp{(iuW)} \,=\, \sum_{n,m} \exp{\left(iu(E_m\suptiny{0}{0}{(\mathrm{f})} - E_n\suptiny{0}{0}{(0)})\right)}p(n,m).
\end{align}
We can calculate this function explicitly for the non-ideal projective measurement by substituting for the probabilities $p(n,m)$ from Eq.~\eqref{eq: joint prob}:
\begin{align}
G(u) &=\sum_{m,n}\sum_l p_n\suptiny{0}{0}{(0)}          q_{ll|n}\suptiny{0}{0}{(0)}|\bra{E_m\suptiny{0}{0}{(\mathrm{f})}}U_{{\Lambda}}\ket{E_l\suptiny{0}{0}{(0)}}|^2\\
    &= \sum_{m}\bra{E_m\suptiny{0}{0}{(\mathrm{f})}} \exp\ch{iuH\suptiny{0}{0}{(\mathrm{f})}} \sum_{n,l} \exp{\ch{-iuE_n\suptiny{0}{0}{(0)}}}p_n\suptiny{0}{0}{(0)} q_{ll|n}\suptiny{0}{0}{(0)}U_{{\Lambda}}\ket{E_l\suptiny{0}{0}{(0)}}\bra{E_l\suptiny{0}{0}{(0)}} U^{\dagger}_{{\Lambda}} \ket{E_m\suptiny{0}{0}{(\mathrm{f})}}\\
     &= \sum_{m}\bra{E_m\suptiny{0}{0}{(\mathrm{f})}} \exp\ch{iuH\suptiny{0}{0}{(\mathrm{f})}} \sum_{n}  \exp{\ch{-iuE_n\suptiny{0}{0}{(0)}}} p_n\suptiny{0}{0}{(0)} U_{{\Lambda}}\rho_n\suptiny{0}{0}{(0)} U_{{\Lambda}}^{\dagger} \ket{E_m\suptiny{0}{0}{(\mathrm{f})}},
\end{align}
where $\sum_l q_{ll|n}\suptiny{0}{0}{(0)} \ket{E_l\suptiny{0}{0}{(0)}}\bra{E_l\suptiny{0}{0}{(0)}} = \rho_n\suptiny{0}{0}{(0)}$ is the post measurement state of the pointer that indicates outcome $n$ with probability $p_n\suptiny{0}{0}{(0)} = \exp\ch{-\beta E_n\suptiny{0}{0}{(0)} }/\mathcal{Z}\suptiny{0}{0}{(0)}$. Here, we have again restricted our analysis to minimal energy UMC measurements as described in Appendix~\ref{appendix:Ideal and Non-Ideal Measurements}.

Let us then write
\begin{align}
\sigma(u) &= \sum_n  \exp{\ch{-iuE_n\suptiny{0}{0}{(0)}}} p_n\suptiny{0}{0}{(0)} \rho_n\suptiny{0}{0}{(0)} = \frac{1}{\mathcal{Z}\suptiny{0}{0}{(0)}} \sum_n  \exp{\ch{-(iu +\beta) E_n\suptiny{0}{0}{(0)} }} \rho_n\suptiny{0}{0}{(0)},
\end{align}
which we then use to {obtain the characteristic function}
\begin{equation}
G(u) = \tr\ch{ \exp\ch{iuH\suptiny{0}{0}{(\mathrm{f})}} \U{\sigma(u) } }.
\end{equation}
To calculate the work average $\mean{\exp(-\beta W)}$, we further calculate $G(u = i\beta)$
\begin{align}
\mean{\exp(-\beta W)}  & = G(u = i\beta)= \tr\left[ \exp\left(-\beta H\suptiny{0}{0}{(\mathrm{f})}\right) \U{\sigma(i\beta)} \right],
\end{align}
where
\begin{align}
\sigma(i\beta) &= \frac{1}{\mathcal{Z}\suptiny{0}{0}{(0)}} \sum_n \exp\ch{-(-\beta +\beta)E_n\suptiny{0}{0}{(0)}}\sum_l q_{ll|n}\suptiny{0}{0}{(0)}\, \ketbra{E_l\suptiny{0}{0}{(0)}}{E_l\suptiny{0}{0}{(0)}}\,=\, \frac{1}{\mathcal{Z}\suptiny{0}{0}{(0)}} \sum_l  \ch{\sum_n q_{ll|n}\suptiny{0}{0}{(0)}}\,  \ketbra{E_l\suptiny{0}{0}{(0)}}{E_l\suptiny{0}{0}{(0)}}.
\end{align}
Consequently, Jarzynski's equality is satisfied if  $\sum_n q_{ll|n}\suptiny{0}{0}{(0)} = 1$, i.e.,
\begin{align}
\mean{\exp(-\beta W)}  = \exp\ch{-\beta \Delta F},
\end{align}
where $\Delta F$ is the free energy $\Delta F =\frac{1}{\mathcal{Z}\suptiny{0}{0}{(0)}} \tr\ch{\exp(-\beta H\suptiny{0}{0}{(\mathrm{f})}) } =  \frac{\mathcal{Z}\suptiny{0}{0}{(\mathrm{f})}}{\mathcal{Z}\suptiny{0}{0}{(0)}}$. The condition  $\sum_n q_{ll|n}\suptiny{0}{0}{(0)} = 1$ is not generally met for all unbiased measurements, and in particular not by any minimal energy UMC measurements beyond dimension\footnote{For the special case $d\Sys=2$, the unbiasedness condition already leaves no choice but to make the matrix $\tilde{\pi}$ a Latin square, and hence the measurement can be a minimal energy UMC measurement that is also minimally invasive.} $d\Sys=2$. However, minimally invasive UMC measurements discussed in Appendix~\ref{appsec:min invasive UMC measurements} satisfy exactly the desired condition, thus allowing to satisfy the Jarzynski equality, while providing the correct (unbiased) measurement statistics and achieving maximal correlation $C_{\mathrm{max}}$ between the pointer outcomes and post-measurement system states.

%%%%%%%%%%%%%%%%%%%%%%%%%%%%%%%%%%%%%%%%%%%%%%%%%%%%%%%%%%%%%%%%%%%%%%%%%%%%%%%%%%%%%%%%%%

\section{Crook's Relation in the Presence of Non-Ideal Measurements}\label{sec:non ideal crooks}

\subsection{Backward Process for Non-Ideal TPM Scheme}

Crook's theorem~\cite{Crooks1999} quantifies the relation between the probability of observing a work value during a realisation of the two projective measurement scheme for a given process $\Lambda$ with the probability of observing the same amount of work for the time-reversed process $\tilde{\Lambda}$. The time-reverse is defined in an operational sense, meaning that the sequence of external operations (driving, measurements and so forth) used to bring the system out-of-equilibrium
during the original process is inverted in the time-reversed process. A prerequisite for obtaining Crook's relation (as well as the Jarzynski equalit~\cite{Jarzynski1997}), is that both
the forward and backward processes start with the system in equilibrium at some given inverse temperature $\beta$. Therefore, the first step in studying whether Crook's relation holds in the non-ideal projective measurement setting (more specifically, restrict our analysis to minimal energy UMC measurements described in Appendix~\ref{appendix:Ideal and Non-Ideal Measurements}), is to define a meaningful time-reversed (backward) process. This is achieved in three steps.
\begin{enumerate}
\item{We start with a system with Hamiltonian
    $H\suptiny{0}{-1}{(\mathrm{f})} = \sum_m E\suptiny{0}{-1}{(\mathrm{f})}_m \ket{E\suptiny{0}{-1}{(\mathrm{f})}_m}\bra{E\suptiny{0}{-1}{(\mathrm{f})}_m}$ that is in equilibrium at inverse temperature $\beta$. The initial state of the backward process, which is to be transformed according to
    \begin{equation}
        \left(\tau\suptiny{0}{0}{(\mathrm{f})}, H\suptiny{0}{0}{(\mathrm{f})}\right) \xrightarrow{\phantom{b}\tilde{\Lambda} \phantom{b}}\left( \rho_{\mathrm{B}}\suptiny{0}{0}{(0)} , H\suptiny{0}{0}{(0)}\right)\,,
    \end{equation}
    reads $\tau\suptiny{0}{0}{(\mathrm{f})}= e^{-\beta H\suptiny{0}{-1}{(\mathrm{f})}}/\mathcal{Z}\suptiny{0}{0}{(\mathrm{f})} = \sum_m p\suptiny{0}{-1}{(\mathrm{f})}_m \ket{E\suptiny{0}{-1}{(\mathrm{f})}_m}\bra{E\suptiny{0}{-1}{(\mathrm{f})}_m}$. Prior to this transformation, the first non-ideal measurement is performed in the eigenbasis $\{ \ket{E\suptiny{0}{-1}{(\mathrm{f})}_m} \}_{m}$ of $H\suptiny{0}{-1}{(\mathrm{f})}$. Given that outcome $m$ is obtained, which occurs with probability $p\suptiny{0}{-1}{(\mathrm{f})}_m$, the post-measurement state reads
    \begin{equation}
        \rho_{m}\suptiny{0}{-1}{(\mathrm{f})}=\sum_{k}q\suptiny{0}{-1}{(\mathrm{f})}_{kk|m}\,\ket{E_{k}\suptiny{0}{-1}{(\mathrm{f})}}\!\!\bra{E_{k}\suptiny{0}{-1}{(\mathrm{f})}}.
    \end{equation}
}
\item{The second step of the backward process consists of driving the system according to the time-reversed           protocol parameterized by $\{ \tilde{\lambda}_t = \lambda_{t_\mathrm{f} - t}; 0 \leq t \leq t_\mathrm{f} \}$,     where $\{\lambda_{t}; 0 \leq t \leq t_\mathrm{f} \}$ is a parameterizsation of the forward process. In the        time-reversed protocol, the system thus evolves according to
    \begin{align}
        \rho_{\mathrm{B}}\suptiny{0}{0}{(0)} = U_{{\tilde{\Lambda}}}\, \rho_{\mathrm{B}}\suptiny{0}{0}{(\mathrm{f})}  \, U_{{\tilde{\Lambda}}}^{\dagger},
        \qquad U_{{\tilde{\Lambda}}}(t_\mathrm{f},0) =  \mathcal{T}_{+}\exp\left(- i \int_{0}^{t_\mathrm{f}}\! H(\tilde{\lambda}_t)\, \text{d}t\right),
    \end{align}
    where $\rho_{\mathrm{B}}\suptiny{0}{0}{(\mathrm{f})} = \sum_n p_n\suptiny{0}{0}{(\mathrm{f})} \rho_n \suptiny{0}{0}{(\mathrm{f})}$ is the average post-measurement state after the first measurement in the backward process.
%The relation between the forward and backward processes has been previously defined in Ref.~\cite{AlhambraMasanesOppenheimPerry2016}.
}
\item{The last step consists of the second non-ideal measurement in the eigenbasis
    $\{\ket{E_{n}\suptiny{0}{-1}{(0)}}\}_{n}$ of $H\suptiny{0}{-1}{(0)}$. We are then interested in determining the joint probability $P_{\mathrm{B}}(m,n)$ to obtain outcome $m$ in the first and outcome $n$ in the second non-ideal measurement, before and after the backward process, respectively. Given outcome $m$ in the first measurement, the probability for the system to be in an energy eigenstate $\ket{E_{k}\suptiny{0}{-1}{(\mathrm{f})}}$ is $q\suptiny{0}{-1}{(\mathrm{f})}_{kk|m}$, and the probability for the backward process to further map the system from the state $\ket{E_{k}\suptiny{0}{-1}{(\mathrm{f})}}$ to $\ket{E_{n}\suptiny{0}{0}{(0)}}$ is $\tilde p_{\raisebox{-1pt}{\scriptsize{$k\!\rightarrow\!n$}}} =|\bra{E_{n}\suptiny{0}{0}{(0)}}U_{{\tilde{\Lambda}}}\ket{E_{k}\suptiny{0}{0}{(\mathrm{f})}}|^2$. Consequently, we obtain the joint probability
    \begin{align}\label{eq:pmn_back}
        P_{\mathrm{B}}(n,m) = \sum_{k}
        \tilde p_{\raisebox{-1pt}{\scriptsize{$k\!\rightarrow\!n$}}}\,
        q\suptiny{0}{-1}{(\mathrm{f})}_{kk|m}\,p_{m}\suptiny{0}{0}{(\mathrm{f})}.
    \end{align}
    The micro-reversibility principle for non-autonomous systems implies $U_{{\tilde{\Lambda}}} = \Theta\, U_{{\Lambda}}^{\dagger} \Theta^{\dagger}$~\cite{CampisiHaenggiTalkner2011}, where $\Theta$ is the anti-unitary time-reversal operator satisfying $\Theta i = - i \Theta$ and $\Theta \Theta^\dagger = \Theta^\dagger \Theta = \mathds{1}$. Micro-reversibility, together with a time-reversal symmetric Hamiltonian, $\Theta H\suptiny{0}{-1}{(k)} \Theta^\dagger = H\suptiny{0}{-1}{(k)}$ for $k= 0, \mathrm{f}$, leads to the relationship  $|\bra{E_{n}\suptiny{0}{0}{(0)}}U_{{\tilde{\Lambda}}}\ket{E_{k}\suptiny{0}{0}{(\mathrm{f})}}|^2 = |\bra{E_{k}\suptiny{0}{0}{(\mathrm{f})}}U_{{\Lambda}}\ket{E_{n}\suptiny{0}{0}{(0)}}|^2$, that is $\tilde p_{\raisebox{-1pt}{\scriptsize{$k\!\rightarrow\!n$}}} =
    p_{\raisebox{-1pt}{\scriptsize{$n\!\rightarrow\!k$}}}$.
}
\end{enumerate}

For non-ideal projective measurements, Crook's relation can be satisfied if we make some changes to the probability distributions, as we shall explain now. To begin, let us consider the work probabilities of the forward and backward process, respectively, i.e.,
\begin{align}
    P_{\mathrm{F}}(W) = \sum_{m,n} P_{\mathrm{F}}(n,m) \delta\ch{(E_m\suptiny{0}{-1}{(\mathrm{f})} - E\suptiny{0}{-1}{(\mathrm{0})}_n) - W},\\\label{workprobs}
    P_{\mathrm{B}}(-W) = \sum_{m,n} P_{\mathrm{B}}(n,m) \delta\ch{(E\suptiny{0}{-1}{(\mathrm{0})}_n - E_m\suptiny{0}{-1}{(\mathrm{f})})+W},
\end{align}
where $P_{\mathrm{F}}(n,m)$ is given in Eq.~\eqref{eq: joint prob}. We then assume that the measurement apparatus in the forward and backward process operates in the same way. More specifically, this means that the pointer is prepared in the same initial state and the same unitary $U_{\mathrm{meas}}$ is used to couple system and pointer. In this case, we have $q\suptiny{0}{-1}{(0)}_{mm|n} = q\suptiny{0}{-1}{(\mathrm{f})}_{mm|n} \equiv q_{mm|n}$.
As in Ref.~\cite{SagawaUeda2012}, we write the ratio between
the joint probabilities of a given transition between $E_{n}\suptiny{0}{0}{(0)}$ to $E_{m}\suptiny{0}{0}{(\mathrm{f})}$ in the forward and backward process as
\begin{equation}\label{local:crooks}
    e^{-\sigma(n,m)} \,:=\, \frac{P_{\mathrm{B}}(n,m)}{P_{\mathrm{F}}(n,m)}.
\end{equation}
The average of the quantity $\sigma(n,m)$ defined by this ratio can be seen to be precisely the relative entropy between the probability of the forwards and backwards processes
\begin{align}
    \mean{\sigma} &= \sum_{m,n} P_{\mathrm{F}}(m,n) \sigma(m,n)\,=\,  \sum_{m,n}P_{\mathrm{F}}(m,n) \log\left(  \frac{P_{\mathrm{F}}(m,n)}{P_{\mathrm{B}}(m,n)}\right)\, =\, D(P_{\mathrm{F}}||P_{\mathrm{B}}),
\end{align}
where the relative entropy of two random variables $Q$ and $P$ is defined $D(P||Q) := \sum_{x}P(x) \log\left(  \frac{P(x)}{Q(x)}\right)$.

In order to express the relation between performing or extracting the same amount of work
in the forward and backward process, one can substitute \eqref{local:crooks} into \eqref{workprobs}
\begin{equation} P_{\mathrm{B}}(-W) = \sum_{m,n} e^{-\sigma(n,m)}P_{\mathrm{F}}(n,m) \, \delta\ch{(E\suptiny{0}{-1}{(\mathrm{0})}_n - E_m\suptiny{0}{-1}{(\mathrm{f})})+W}.  \end{equation}
It is straightforward to see that
\begin{align} \exp{\left(-\sigma(n,m)\right)} &= \frac{p_m\suptiny{0}{-1}{(\mathrm{f})}}{p_n\suptiny{0}{-1}{(0)}}
\frac{P_{\mathrm{B}}(n|m)}{P_{\mathrm{F}}(m|n)}
    \,=\, \exp{\left(-\beta(E_m\suptiny{0}{-1}{(\mathrm{f})} - E\suptiny{0}{-1}{(\mathrm{0})}_n - \Delta F)\right)}\, \frac{P_{\mathrm{B}}(n|m)}{P_{\mathrm{F}}(m|n)}
\end{align}
where $\Delta F = \beta^{-1}\log(\mathcal{Z}\suptiny{0}{0}{(0)}/\mathcal{Z}\suptiny{0}{0}{(\mathrm{f})})$ is the equilibrium free energy difference, and the conditional probability to find the state with energy $n$ ($m$) after the backward (forward) process if initially it was $m$ ($n$) is ${P_{\mathrm{B}}(n|m) = \sum_{k}p_{\raisebox{-1pt}{\scriptsize{$k\!\rightarrow\!n$}}}q_{kk|m}}$ (${P_{\mathrm{F}}(m|n)}=\sum_{l}p_{\raisebox{-1pt}{\scriptsize{$l\!\rightarrow\!m$}}}q_{ll|n}$).
Therefore, a Crook's-like relation for non-ideal projective measurements can be written as
\begin{align}\label{eq:crooks.sec:ttpm}
    P_{\mathrm{B}}(-W) = e^{-\beta(W - \Delta F)} {\tilde{P}_{\mathrm{F}}}(W),
\end{align}
with
\begin{align}
{\tilde{P}_{\mathrm{F}}}(W) = \sum_{m,n} \exp{\left(-\gamma(m,n)\right)}  P_{\mathrm{F}}(m,n) \delta\ch{(E_m\suptiny{0}{-1}{(\mathrm{f})} - E\suptiny{0}{-1}{(0)}_n) - W},
\end{align}
where $ \gamma(m,n) =\log{\left({P_{\mathrm{F}}(m|n)}{/P_{\mathrm{B}}(n|m)}\right)}$. If ideal projective measurements are assumed in the estimation process (in other words, if an infinite amount of resources is available) no disturbance is created by the measurements on the system in the forward and backward processes, which implies that $q_{mm'|n} =\delta_{m,m'} \delta_{m,n}$, and $\gamma(m,n)=0$.
Therefore Eq.~\eqref{eq:crooks.sec:ttpm} results in the well known Crook's relation
\begin{equation}\label{eq:crooks}
    P_{\mathrm{B}}(-W) = e^{-\beta(W - \Delta F)} {P}_{\mathrm{F}}(W).
\end{equation}

%%%%%%%%%%%%%%%%%%%%%%%%%%%%%%%%%%%%%%%%%%%%%%%%%%%%%%%%%%%%%%%%%%%%%%%%%%%%%%%%%%%%%%%%%%

\subsection{Irreversibility and Dissipation} \label{app:irre}

Consider a thermal state $\tau\suptiny{0}{0}{(0)}(\beta)$ of a system with Hamiltonian $H\suptiny{0}{0}{(0)}$ at time $t_0$ that is driven out of equilibrium by means of the unitary
\begin{equation} \label{eq:unitary}
    U_{\Lambda}(t_\mathrm{f},0) =  \mathcal{T}_{+}\exp\left(- i \int_{0}^{t_\mathrm{f}}\! H(\lambda_t)\, \text{d}t\right).
\end{equation}
At time $t_\mathrm{f}$, the system is in the  (out-of-equilibrium) state $\rho\suptiny{0}{0}{(\mathrm{f})} = \U{\tau\suptiny{0}{0}{(0)}}$, and the Hamiltonian of the system is $H\suptiny{0}{0}{(\mathrm{f})}$. Then the system is coupled again to a thermal bath at temperature $T=1/\beta$, and left to thermalise to the equilibrium state $\tau\suptiny{0}{0}{(\mathrm{f})}(\beta)$.

The work dissipated in the driving process, defined as $\mean{W}_\Lambda - \Delta F$, is the extra amount of energy that is transferred to the bath in the final thermalisation step, leading to entropy production. This energy cannot be reversibly recovered by reversing the protocol. A closed-form expression for this amount of work presented in Ref.~\cite{ParrondoVanDenBroeckKawai2009} is
\begin{equation}\label{eq:diss.ideal}
    \mean{W}_\Lambda - \Delta F = k\subtiny{0}{0}{\mathrm{B}} T D(\rho_{\mathrm{F}}(t)|| \Theta^\dagger \rho_{\mathrm{B}}(t_\mathrm{f} - t) \Theta),
\end{equation}
where $D(\rho_{\mathrm{F}}(t)||\Theta^\dagger \rho_{\mathrm{B}}(t_\mathrm{f} - t) \Theta)$ is the relative entropy between the state $\rho_{\mathrm{F}}(t)= U_{\Lambda}(t,0) \tau\suptiny{0}{0}{(0)} U_{\Lambda}^\dagger(t,0)$ out of equilibrium at time $t$ and $\rho_{\mathrm{B}}(t_\mathrm{f} - t)= U_{\tilde \Lambda}(t_\mathrm{f} - t,0) \tau\suptiny{0}{0}{(\mathrm{f})} U_{\tilde \Lambda}^\dagger(t_\mathrm{f} - t,0)$ is the state of the system in the backward process. Here $U_{\tilde \Lambda}(t_\mathrm{f}-t,0)$ is the unitary evolution generated by the time-reversed protocol $\tilde \Lambda$. If the process is reversible, $\rho_{\mathrm{F}}(t) = \Theta^\dagger \rho_{\mathrm{B}}(t_\mathrm{f} - t) \Theta$ for any $0\leq t \leq t_\mathrm{f}$, and there is no work being dissipated in the process. In particular, the equilibrium state is reached at time $t_\mathrm{f}$, that is, $\rho\suptiny{0}{0}{(\mathrm{f})} = \tau\suptiny{0}{0}{(\mathrm{f})}$.

Now, consider non-ideal projective measurements used for work estimation, more specifically, minimal energy UMC measurements described in Appendix~\ref{appendix:Ideal and Non-Ideal Measurements}. Here, we see that more energy is dissipated in the process because some entropy is produced during the measurement process.
\begin{prop}
For work estimation based on non-ideal minimal energy UMC measurements measurements, the work dissipated in the driving process is
\begin{align}\label{eq:proof61}
    \mean{W}_{\mathrm{non-id}} - \Delta F  =  k\subtiny{0}{0}{\mathrm{B}} T \left[D(\tilde{\rho}\suptiny{0}{0}{(\mathrm{f})}|| \tau\suptiny{0}{0}{(\mathrm{f})})
 +\Delta S_0\right],
\end{align}
where $\tilde{\rho}\suptiny{0}{0}{(\mathrm{f})} = \U{\ch{\sum_n p_n\suptiny{0}{0}{(0)}  \sum_l  q_{ll|n}\suptiny{0}{0}{(0)} \ketbra{E_l\suptiny{0}{0}{(0)}}{E_l\suptiny{0}{0}{(0)}}}}$ and  $\Delta S_0 = S(\tilde{\rho}\suptiny{0}{0}{(0)})- S(\tau\suptiny{0}{0}{(0)})$ is the entropy change in the system due to the first non-ideal measurement.
\end{prop}

\begin{proof}
For the TPM scheme with non-ideal measurements, the work estimate is as in Eq.~\eqref{nworknice}, i.e.,
\begin{align}
    \mean{W}_{\mathrm{non-id}} &= \int dW P(W) W \,=\, \sum_{m,n} p_n\suptiny{0}{0}{(0)} \sum_l  q_{ll|n}\suptiny{0}{0}{(0)} p_{\raisebox{-1pt}{\scriptsize{$l\!\rightarrow\!m$}}}\ch{E_m\suptiny{0}{0}{(\mathrm{f})} - E_n\suptiny{0}{0}{(0)}},
\end{align}
where $p_{\raisebox{-1pt}{\scriptsize{$l\!\rightarrow\!m$}}} = |\bra{E_{m}\suptiny{0}{0}{(\mathrm{f})}}U_{{\Lambda}}\ket{E_{l}\suptiny{0}{0}{(0)}}|^2$.
We then note that the thermal states with respect to the initial and final Hamiltonians are
\begin{subequations}
    \begin{align}
        \tau\suptiny{0}{0}{(0)} &= \exp(-\beta H\suptiny{0}{0}{(0)})/\mathcal{Z}\suptiny{0}{0}{(0)} = \sum_n p_n\suptiny{0}{0}{(0)}\ketbra{E_n\suptiny{0}{0}{(0)}}{E_n\suptiny{0}{0}{(0)}},\\
        \tau\suptiny{0}{0}{(\mathrm{f})}  &= \exp(-\beta H\suptiny{0}{0}{(\mathrm{f})})/\mathcal{Z}\suptiny{0}{0}{(\mathrm{f})}=  \sum_m p_m\suptiny{0}{0}{(\mathrm{f})}\ketbra{E_m\suptiny{0}{0}{(\mathrm{f})}}{E_m\suptiny{0}{0}{(\mathrm{f})}},
    \end{align}
\end{subequations}
where the partition functions are $\mathcal{Z}\suptiny{0}{0}{(0)} = \tr\ch{\exp(-\beta H\suptiny{0}{0}{(0)})}$ and $\mathcal{Z}\suptiny{0}{0}{(\mathrm{f})} = \tr\ch{\exp(-\beta H\suptiny{0}{0}{(\mathrm{f})})}$, and the probabilities for the individual energy eigenstates are $p_n\suptiny{0}{0}{(0)}=\exp(-\beta E_n\suptiny{0}{0}{(0)})/\mathcal{Z}\suptiny{0}{0}{(0)}$ and $p_m\suptiny{0}{0}{(\mathrm{f})}=\exp(-\beta E_m\suptiny{0}{0}{(\mathrm{f})})/\mathcal{Z}\suptiny{0}{0}{(\mathrm{f})}$. The logarithms of the probabilities above are $\log p_n\suptiny{0}{0}{(0)}=-\beta E_n\suptiny{0}{0}{(0)} - \log \mathcal{Z}\suptiny{0}{0}{(0)}$ and $\log p_m\suptiny{0}{0}{(\mathrm{f})}=-\beta E_m\suptiny{0}{0}{(\mathrm{f})} - \log \mathcal{Z}\suptiny{0}{0}{(\mathrm{f})}$. With this, we can rewrite the factor $\ch{E_m\suptiny{0}{0}{(\mathrm{f})} - E_n\suptiny{0}{0}{(0)}}$ in ${\mean{W}_{\mathrm{non-id}}}$ and obtain
\begin{align}
    \beta \mean{W}_{\mathrm{non-id}}    &=\, \sum_{m,n} p_n\suptiny{0}{0}{(0)} \sum_l p_{\raisebox{-1pt}{\scriptsize{$l\!\rightarrow\!m$}}}\,  q_{ll|n}\suptiny{0}{0}{(0)} \ch{ (-\log p_m\suptiny{0}{0}{(\mathrm{f})} - \log \mathcal{Z}\suptiny{0}{0}{(\mathrm{f})}) - (-\log p_n\suptiny{0}{0}{(0)}  - \log \mathcal{Z}\suptiny{0}{0}{(0)}) }.
\end{align}
We then note that $\sum_{m}\sum_{l}p_{\raisebox{-1pt}{\scriptsize{$l\!\rightarrow\!m$}}}\,  q_{ll|n}\suptiny{0}{0}{(0)}=1$ and identify the free energy $\Delta F = \frac{1}{\beta} ( \log \mathcal{Z}\suptiny{0}{0}{(0)}- \log \mathcal{Z}\suptiny{0}{0}{(\mathrm{f})})$, as well as the initial thermal state entropy $S(\tau\suptiny{0}{0}{(0)}) = - \sum_n p_{n}\suptiny{0}{0}{(0)} \log p_{n}\suptiny{0}{0}{(0)}$ to write
\begin{align}
    \beta\left ({\mean{W}_{\mathrm{non-id}}} - \Delta F \right) &= - S(\tau\suptiny{0}{0}{(0)}) - \sum_{m,n} p_n\suptiny{0}{0}{(0)} \sum_l  p_{\raisebox{-1pt}{\scriptsize{$l\!\rightarrow\!m$}}}\, q_{ll|n}\suptiny{0}{0}{(0)} \log p_m\suptiny{0}{0}{(\mathrm{f})}= - S(\tau\suptiny{0}{0}{(0)}) -
    \tr\ch{\U{\ch{\sum_n  p_{n}\suptiny{0}{0}{(0)}
    \sum_l  q_{ll|n}\suptiny{0}{0}{(0)} \ketbra{E_{l}\suptiny{0}{0}{0}}{E_{l}\suptiny{0}{0}{0}}} }\log \tau\suptiny{0}{0}{(\mathrm{f})}}\nonumber\\[1mm]
    &=- S(\tau\suptiny{0}{0}{(0)}) - \tr\left(\U{\tilde\rho\suptiny{0}{0}{(0)}}\, \log \tau\suptiny{0}{0}{(\mathrm{f})}\right),
\end{align}
where we have recognised the conditional post-measurement state $\rho_{n}\suptiny{0}{0}{(0)}=\sum_l  q_{ll|n}\suptiny{0}{0}{(0)} \ketbra{E_{l}\suptiny{0}{0}{0}}{E_{l}\suptiny{0}{0}{0}}$, and we have denoted the unconditional post-measurement state as $\tilde{\rho} \suptiny{0}{0}{(0)}= \sum_n p_n\suptiny{0}{0}{(0)} \rho_{n}\suptiny{0}{0}{(0)}$. We can then add and subtract the entropy of $\tilde\rho\suptiny{0}{0}{(\mathrm{f})} =  \U{\tilde\rho\suptiny{0}{0}{(0)}}$ and use the invariance of the von~Neumann entropy under unitaries, in particular, $S(\tilde\rho\suptiny{0}{0}{(\mathrm{f})})=S(\tilde\rho\suptiny{0}{0}{(0)})$ to arrive at
\begin{align}
 \beta\left ({\mean{W}_{\mathrm{non-id}}} - \Delta F \right) = D(\tilde\rho\suptiny{0}{0}{(\mathrm{f})}|| \tau\suptiny{0}{0}{(\mathrm{f})})+S(\tilde\rho\suptiny{0}{0}{(0)})- S(\tau\suptiny{0}{0}{(0)}) \,=\,
 D(\tilde\rho\suptiny{0}{0}{(\mathrm{f})}|| \tau\suptiny{0}{0}{(\mathrm{f})})+\Delta S_{0}.
\end{align}
\end{proof}

This result can be seen as a version of Eq.~\eqref{eq:diss.ideal} applicable to non-ideal projective measurement when taking $t= t_\mathrm{f}$, which expresses the amount of irreversible work in the process $\Lambda$ calculated by means of the non-ideal TPM scheme. $\Delta S_0$ represents the entropy change in the system due to the initial non-ideal measurement process, which may be either positive or negative in general. However, for unital measurements, the entropy of the pointer does not change~\cite{ManzanoHorowitzParrondo2015} and, as a consequence, the second law implies $\Delta S_0 \geq 0$, which can now be interpreted as the entropy production in the measurement process.

\begin{prop}
Consider a system of dimension $d\Sys$, and an interaction between system and pointer such that their correlations are maximal, $C(\tilde{\rho}\SP)=(C_{max})$, for instance, any UMC measurement. Then the entropy produced in the measurement process satisfies
\begin{equation}\label{eq.bla}
\Delta S_0 \leq (1-C_{max}) \log(d\Sys-1) + H_2(C_{max}),
\end{equation}
where $H_2(x) = -x\log x - (1-x)\log(1-x)$ is the binary entropy of the random variable $x$ with $0\leq x \leq 1$.
\end{prop}

\begin{proof}
This inequality comes directly from the Fannes-Audenaert inequality in Refs.~\cite{Fannes1973, Audenaert2007}
\begin{equation}
S(\tilde\rho\suptiny{0}{0}{(0)})- S(\tau_0) \leq D(\tilde\rho\suptiny{0}{0}{0)}, \tau_0) \log(d\Sys-1) + H_2(D(\tilde\rho\suptiny{0}{0}{(0)}, \tau_0\suptiny{0}{0}{(0)})),
\end{equation}
where $D(\tilde\rho\suptiny{0}{0}{(0)}, \tau\suptiny{0}{0}{(0)}) = \frac{1}{2} \N{\tilde\rho\suptiny{0}{0}{(0)} - \tau\suptiny{0}{0}{(0)}}_1$ is the trace distance. To obtain Eq.~\eqref{eq.bla}, the trace distance must be  $D(\tilde\rho\suptiny{0}{0}{(0)}, \tau\suptiny{0}{0}{(0)}) \leq 1$. Since the trace distance is convex and monotonic:
\begin{align}
D(\tilde\rho\suptiny{0}{0}{(0)}, \tau\suptiny{0}{0}{(0)}) &= \frac{1}{2} \N{\tilde\rho\suptiny{0}{0}{(0)} - \tau\suptiny{0}{0}{(0)}}_1
     \,\leq\, \frac{1}{2}  \sum_n p_n\suptiny{0}{0}{(0)}  \N{\rho_n\suptiny{0}{0}{(0)}-\ketbra{E_n\suptiny{0}{0}{(0)}}{E_n\suptiny{0}{0}{(0)}}}_1
    \,=\, 1 - C_{max},
\end{align}
where the last equality comes from Eq.~\eqref{eq:cmax.dist}.
\end{proof}

\begin{prop}
Assuming non-ideal (minimal energy UMC) projective measurements for work estimation, the work dissipated in the driving process can be written as
\begin{align} \label{eq:fbproof}
    {\mean{W}_{\mathrm{non-id}}} - \Delta F  =  k\subtiny{0}{0}{\mathrm{B}} T \left[ \Delta S_0 + \Delta D_\mathrm{f} +D(\rho_{\mathrm{F}}(t)||\hspace*{1pt} \Theta^\dagger \rho_{\mathrm{B}}(t_\mathrm{f} - t) \Theta) \right],
\end{align}
where $\rho_{\mathrm{F}}(t) = U_\Lambda(t,0) \tilde \rho \suptiny{0}{0}{(\mathrm{f})} U^\dagger_\Lambda(t,0)$ is the system state at intermediate time $0 \leq t \leq t_\mathrm{f}$ in the forward process, $\rho_{\mathrm{B}}(t_\mathrm{f}-t) =  U_{\tilde \Lambda}(t_\mathrm{f}-t,0) \rho_{\mathrm{B}} \suptiny{0}{0}{(\mathrm{f})} U^\dagger_{\tilde \Lambda}(t_\mathrm{f}-t,0)$ is the (inverted) state of the system at the same instant of time in the backward process, and we introduced the correction term
\begin{align}
    \Delta D_{\mathrm{f}}
    &= D(\tilde{\rho}\suptiny{0}{0}{(\mathrm{f})}||\hspace*{1pt} \tau\suptiny{0}{0}{(\mathrm{f})}) - D(\tilde{\rho}\suptiny{0}{0}{(\mathrm{f})}||\hspace*{1pt} \rho_{\mathrm{B}}\suptiny{0}{0}{(\mathrm{f})}) = \tr[\tilde{\rho}\suptiny{0}{0}{(\mathrm{f})} \left(\log \rho_{\mathrm{B}}\suptiny{0}{0}{(\mathrm{f})} -  \log \tau\suptiny{0}{0}{(\mathrm{f})} \right)].
\end{align}

\begin{proof}
We start from Eq.~\eqref{eq:proof61}, adding and subtracting the quantity $D(\tilde{\rho}\suptiny{0}{0}{(\mathrm{f})}||\hspace*{1pt} \rho_{\mathrm{B}}\suptiny{0}{0}{(\mathrm{f})})$ to arrive at
\begin{align}\label{eq:p1}
    {\mean{W}_{\mathrm{non-id}}} - \Delta F  =  k\subtiny{0}{0}{\mathrm{B}} T \left[ \Delta S_0 + \Delta D_\mathrm{f} + D(\tilde{\rho}\suptiny{0}{0}{(\mathrm{f})}||\hspace*{1pt} \rho_{\mathrm{B}}\suptiny{0}{0}{(\mathrm{f})}) \right].
\end{align}
Then we use the properties of the quantum relative entropy and the unitary evolution~\eqref{eq:unitary} to write
\begin{align} \label{eq:p2}
    D(\tilde{\rho}\suptiny{0}{0}{(\mathrm{f})}||\hspace*{1pt} \rho_{\mathrm{B}}\suptiny{0}{0}{(\mathrm{f})}) &=
    D(U_\Lambda(t_\mathrm{f},t) U_\Lambda(t,0) \tilde{\rho}\suptiny{0}{0}{(0)} U^\dagger_\Lambda(t,0) U^\dagger_\Lambda(t_\mathrm{f},t) ||\hspace*{1pt} \rho_{\mathrm{B}}\suptiny{0}{0}{(\mathrm{f})}) \nonumber \\
    &= D(U_\Lambda(t,0) \tilde{\rho}\suptiny{0}{0}{(0)} U^\dagger_\Lambda(t,0) ||\hspace*{1pt} U^\dagger_\Lambda(t_\mathrm{f},t) \rho_{\mathrm{B}}\suptiny{0}{0}{(\mathrm{f})}U_\Lambda(t_\mathrm{f},t)) \nonumber \\
    &=D(U_\Lambda(t,0) \tilde{\rho}\suptiny{0}{0}{(0)} U^\dagger_\Lambda(t,0) ||\hspace*{1pt} \Theta^\dagger U_{\tilde{\Lambda}}(t_\mathrm{f}-t,0) \rho_{\mathrm{B}}\suptiny{0}{0}{(\mathrm{f})}U_{\tilde \Lambda}^\dagger(t_\mathrm{f}- t, 0)\Theta) = D(\rho_{\mathrm{F}}(t)||\hspace*{1pt} \Theta^\dagger \rho_{\mathrm{B}}(t_\mathrm{f} - t) \Theta),
\end{align}
where we have used the micro-reversibility principle for non-autonomous systems~\cite{CampisiHaenggiTalkner2011} in the last line , i.e., $U^\dagger_\Lambda(t_\mathrm{f},t) = \Theta^\dagger U_{\tilde{\Lambda}}(t_\mathrm{f}-t,0) \Theta$, and we have identified the expressions for $\rho_{\mathrm{F}}(t)$ and $\rho_{\mathrm{B}}(t_\mathrm{f}-t)$. Insering Eq.~\eqref{eq:p2} into Eq.~\eqref{eq:p1}, we directly obtain Eq.~\eqref{eq:fbproof}.
\end{proof}
\end{prop}

\begin{prop}
In the case of open quantum systems, Eq.~\eqref{eq:fbproof} becomes the inequality
\begin{align} \label{eq:fbproofineq}
 {\mean{W}_{\mathrm{non-id}}} - \Delta F  \geq  k\subtiny{0}{0}{\mathrm{B}} T \left[ \Delta S_0 + \Delta D_\mathrm{f} +D(\rho_{\mathrm{F}}(t)||\hspace*{1pt} \Theta^\dagger \rho_{\mathrm{B}}(t_\mathrm{f}- t) \Theta) \right].
\end{align}

\begin{proof}
Let us consider our system of interest as before along with a thermal bath to which it is coupled. The joint system can be considered to be closed, and Eq.~\eqref{eq:fbproof} hence holds for the joint system, i.e.,
\begin{align} \label{eq:p1b}
 {\mean{W}_{\mathrm{non-id}}} - \Delta F  =  k\subtiny{0}{0}{\mathrm{B}} T \left[ \Delta S_0 + \Delta D_\mathrm{f} +D(\rho^\prime_{\mathrm{F}}(t)||\hspace*{1pt} \Theta^\dagger \rho^\prime_B(t_\mathrm{f}-t)\Theta) \right] \leq k\subtiny{0}{0}{\mathrm{B}} T \left[ \Delta S_0 + \Delta D_\mathrm{f} +D(\rho_{\mathrm{F}}(t)||\hspace*{1pt} \Theta^\dagger \rho_{\mathrm{B}}(t_\mathrm{f}-t)\Theta) \right].
\end{align}
The left-hand side is the same as before as long as work is performed by implementing the protocol $\Lambda$ only involving system degrees of freedom, whereas the primed quantities on the right-hand side correspond to the global state of system and bath. Notice that assuming that non-ideal measurements are only performed on the system implies that we recover the same terms $\Delta S_0$ and $\Delta D_\mathrm{f}$. Finally, applying monotonicity of the relative entropy under the partial trace, i.e., $S(\rho | \sigma )\geq S(\rho^\prime || \sigma^\prime)$ for any $\rho=\tr_{\mathrm{bath}}[\rho^\prime]$ and $\sigma=\tr_{\mathrm{bath}}[\sigma^\prime]$, the last inequality is obtained.
\end{proof}
\end{prop}

%%%%%%%%%%%%%%%%%%%%%%%%%%%%%%%%%%%%%%%%%%%%%%%%%%%%%%%%%%%%%%%%%%%%%%%%%%%%%%%%%%%

\section{Two-level system driven by a classical field} \label{appsec:driven atom}

To illustrate the formalism presented in this work, let us consider a two-level atom driven out of equilibrium by a classical field, described by the Hamiltonian $H = H\Sys + H\Field$, where $H\Sys = -\tfrac{E\Sys}{2}\sigma_z$ is the atomic Hamiltonian with $\sigma_z=\ket{0}\!\!\bra{0}-\ket{1}\!\!\bra{1}$, and the interaction Hamiltonian is $H\Field = - \mathbf{D}\cdot \mathbf{E}\Field$, where $\mathbf{E}\Field= i \varepsilon\Field \ch{\text{\boldmath$\epsilon$}\Field e^{-i\omega\Field t}e^{-i\phi_0} + \text{\boldmath$\epsilon$}\Field^{*} e^{i\omega\Field t}e^{i\phi_0}}$ is a classical field with real-amplitude $\varepsilon\Field$, polarization vector $\text{\boldmath$\epsilon$}\Field$, phase $\phi_0$ and angular frequency $\omega\Field$. The atomic dipole is $\mathbf{D} = d(\text{\boldmath$\epsilon$}\Sys^{*}\sigma_{+}+ \text{\boldmath$\epsilon$}\Sys\sigma_{-})$, with $d$ being the dipole strength, and the vector $\text{\boldmath$\epsilon$}\Sys$ describes the atomic transition polarization given the energy transition operators $\sigma_{+} = \ketbra{1}{0}$ and $\sigma_{-} = \ketbra{0}{1}$. After applying the Rotating Wave Approximation (RWA), the Hamiltonian describes a Rabi oscillation with angular frequency $\tilde{\Omega}\Field$, i.e.,
\begin{equation}
    H \approx \frac{\tilde{\Omega}\Field}{2} \text{\boldmath$\sigma$}\cdot \mathbf{n},
    \label{eq:Rabi osc Hamiltonian}
\end{equation}
where $\text{\boldmath$\sigma$}= (\sigma_x, \sigma_y, \sigma_z)$ and $\mathbf{n}=(\Delta\Field\mathbf{u}_z+\Omega\Field \mathbf{u}_y)/\tilde{\Omega}\Field$ with corresponding unit vectors $(\mathbf{u}_x,\mathbf{u}_y,\mathbf{u}_z)$, and we use units where $\hbar=1$. The angular frequency is defined in terms of atom and field variables as
\begin{equation}
    \tilde{\Omega}\Field^2 = \Delta\Field^2 + \Omega\Field^2,
\end{equation}
where $\Delta\Field = E\Sys - \omega\Field$ is the difference between the angular frequency of the field and the energy gap of the atom (recall that $\hbar=1$), and $\Omega\Field = 2d \varepsilon\Field \ch{\text{\boldmath$\epsilon$}\Field \cdot \text{\boldmath$\epsilon$}\Sys^{*}} e^{i\phi_0}$ is the classical Rabi frequency of the atom-field interaction. The transformation generated by the Hamiltonian in Eq.~(\ref{eq:Rabi osc Hamiltonian}) is a unitary of the form
\begin{equation}\label{eq:U theta}
    U(\theta) = \exp\ch{-i\frac{\theta}{2}\text{\boldmath$\sigma$}\cdot \mathbf{n} } = \cos\ch{\frac{\theta}{2}}\id - i \sin\ch{\frac{\theta}{2}} \text{\boldmath$\sigma$}\cdot \mathbf{n},
\end{equation}
for $\theta=\tilde{\Omega}\Field t$, and $t$ is the duration of the transformation. For the purpose of illustration, let us restrict our further analysis to the resonant case, where $\Delta\Field \approx 0$. This implies $\mathbf{n}=\mathbf{u}_y$ and hence that the transformation $U(\theta)$ is a rotation around the $y$ axis and $\tilde{\Omega}\Field=\Omega\Field$. Further details about this physical system, and a more complete and general scenario can be found in Ref.~\cite{HarocheRaimond2006}.

Let us now consider a process in which the atom is initially prepared in a thermal state at inverse temperature $\beta\Sys=1/(k\suptiny{0}{0}{\mathrm{B}}T\Sys)$ with respect to the initial system Hamiltonian $H\Sys\suptiny{0}{0}{(0)} = -E\Sys \sigma_z/2$, i.e.,
\begin{equation}
    \rho\suptiny{0}{0}{(0)}\Sys \,=\, \tau\suptiny{0}{0}{(0)}\Sys\,=\,\frac{\exp\ch{-\beta\Sys H\Sys\suptiny{0}{0}{(0)}}}{\mathcal{Z}\suptiny{0}{0}{(0)}\Sys}.
    \label{eq: initial state atom}
\end{equation}

We then consider the TPM scheme with non-ideal measurements to estimate the work that is performed on the system by the transformation $U(\theta)$, where we assume the measurements to be non-ideal but minimal energy UMC measurements (see Appendix~\ref{appsec:min energy unbiased max corr measurements}). After the first measurement, the interaction with the field $(H\Field)$ is instantaneously switched on at $t=0$, and the field evolves until time $t=t_\mathrm{f}$ to an out-of-equilibrium state given by $\rho\Sys\suptiny{0}{0}{(\mathrm{f})}(\theta) = U(\theta) \rho\suptiny{0}{0}{(0)}\Sys  U(\theta)^{\dagger}$, with $\theta=t_\mathrm{f}\Omega\Field$. At the time $t=t_\mathrm{f}$ the interaction with the field $(H\Field)$ is instantaneously switched off, such that $H\suptiny{0}{0}{(\mathrm{f})}\Sys = H\Sys\suptiny{0}{0}{(0)}$, and the second non-ideal measurement is performed.

The work done by the field by means of $U(\theta)$ is the energy difference between the initial and final configurations,
\begin{equation}
    \mean{W}_{\Lambda} \,=\, \tr\ch{\rho\suptiny{0}{0}{(\mathrm{f})}\Sys(\theta)H\suptiny{0}{0}{(\mathrm{f})}\Sys} - \tr\ch{\rho\suptiny{0}{0}{(0)}\Sys H\suptiny{0}{0}{(0)}\Sys}
    \,=\, -\frac{E\Sys}{2} \tr\ch{\sigma_z (\rho\suptiny{0}{0}{(\mathrm{f})}\Sys(\theta) - \rho\suptiny{0}{0}{(0)}\Sys )}
    \,=\,E\Sys\,\sin^{\hspace*{-0.5pt}2}\bigl(\tfrac{\theta}{2}\bigr)\tanh\bigl(\tfrac{\beta\Sys E\Sys}{2}\bigr).
\end{equation}
However, when we estimate this work using non-ideal measurements, we obtain a different value $\mean{W}_{\mathrm{non-id}}$. For instance, let us consider minimal energy UMC measurements using a three-qubit pointer that is prepared in a thermal state
\begin{equation}
    \tau(\beta\Poi) = \ch{\frac{\exp\ch{-\beta\Poi H\Poi} }{\mathcal{Z}\Poi} }^{\otimes 3},
\end{equation}
where $H\Poi = -\tfrac{E\Poi}{2}\sigma_{z}=\tfrac{E\Poi}{2}\bigl(\ketbra{1}{1}-\ketbra{0}{0}\bigr)$ is the Hamiltonian of each single-qubit subsystem of the pointer and $\beta\Poi=1/(k\subtiny{0}{0}{\mathrm{B}} T\Poi)$. As discussed before, the maximum correlation created between system and pointer depends only on the preparation of the measurement apparatus, which in this example is the sum of the $d\Poi/d\Sys = 4$ biggest eigenvalues of the pointer, i.e.,
\begin{equation}
    C_{\mathrm{max}}  = \frac{\ch{1+3\exp\ch{-\beta\Poi E\Poi }}}{\mathcal{Z}\Poi^3}.
\end{equation}
Following the approach developed in Sec.~\ref{sec:work estimation nonid tpm}, the non-ideal work estimate can be written as
\begin{equation}
   \mean{W}_{\mathrm{non-id}} =  \tr\Bigl( U(\theta)\,\tilde{\rho}\Sys\suptiny{0}{0}{(0)}\,
    U(\theta)^{\dagger}\,H\Sys\suptiny{0}{0}{(\mathrm{f})}\Bigr) - \tr\ch{\rho\suptiny{0}{0}{(0)}\Sys H\Sys\suptiny{0}{0}{(0)}},
\end{equation}
where the unconditional post-measurement system state is $\tilde{\rho}\Sys\suptiny{0}{0}{(0)} := \sum_n p_{n}\suptiny{0}{0}{(0)}  \rho_{n}\suptiny{0}{0}{(0)}$. The conditional post-measurement system states are $\rho_{n}\suptiny{0}{0}{(0)} =  \sum_l q\suptiny{0}{-1}{(0)}_{ll|n}\ketbra{{l}}{{l}}$ and $p_{0}\suptiny{0}{0}{(0)} = 1/\mathcal{Z}\Sys$ and $p_{1}\suptiny{0}{0}{(0)} = \exp\ch{-\beta\Sys E\Sys}/\mathcal{Z}\Sys$ and $\mathcal{Z}\Sys=1+\exp\ch{-\beta\Sys E\Sys}$. The probability $q\suptiny{0}{-1}{(0)}_{ll|n}$ to find the post-measurement state $\rho_n\suptiny{0}{0}{(0)}$ in a specific eigenstate $\ket{l}$ of the system, depends only on the amount of correlation created between system and pointer, that is,
 \begin{align}
     q\suptiny{0}{-1}{(0)}_{nn|n} & = C_{\mathrm{max}}  = \frac{1+3\exp\ch{-\beta\Poi E\Poi }}{\mathcal{Z}\Poi^3}\ \ \text{for}\ \ n=0,1, \\
     q\suptiny{0}{-1}{(0)}_{00|1}=q\suptiny{0}{-1}{(0)}_{11|0} & =1- C_{\mathrm{max}}  = \frac{\exp\ch{-3\beta\Poi E\Poi } + 3\exp\ch{-2\beta\Poi E\Poi }}{\mathcal{Z}\Poi^3}.
 \end{align}
Given the initial state of the atom in Eq.~\eqref{eq: initial state atom}, the unconditional  post-measurement system state
evolves as
\begin{equation}
    U(\theta)\,\tilde{\rho}\Sys\suptiny{0}{0}{(0)}\,
    U(\theta)^{\dagger} = \frac{1}{\mathcal{Z}\Sys}\ch{U(\theta)\,\rho_{n=0}\suptiny{0}{0}{(0)}\,
    U(\theta)^{\dagger} + \exp\ch{-\beta\Sys E\Sys}U(\theta)\,\rho_{n=1}\suptiny{0}{0}{(0)}\,U(\theta)^{\dagger} }.
\end{equation}
Figure~\ref{fig:non-ideal work delta}~(a) shows $\mean{W}_{\mathrm{non-id}}$ as a function of  $\theta=\Omega\Field t$ for different temperatures of the pointer, where each of the three pointer qubits is assumed to have an energy gap $E\Poi=E\Sys/10$. The system is taken to be at room temperature $T\Sys = 300K$ initially, with an energy gap in the microwave regime, such that $\beta\Sys E\Sys \approx 1/30$. The pointer is initially at the same temperature as the system, but the pointer can be cooled in order to obtain a better precision in the work estimation as illustrated in Fig.~\ref{fig:non-ideal work delta}~(b).

\begin{figure}[t!]
    %%%trim={<left> <lower> <right> <upper>}
        (a)\includegraphics[width=0.45\columnwidth,trim={0cm 0mm 0cm 0mm}]{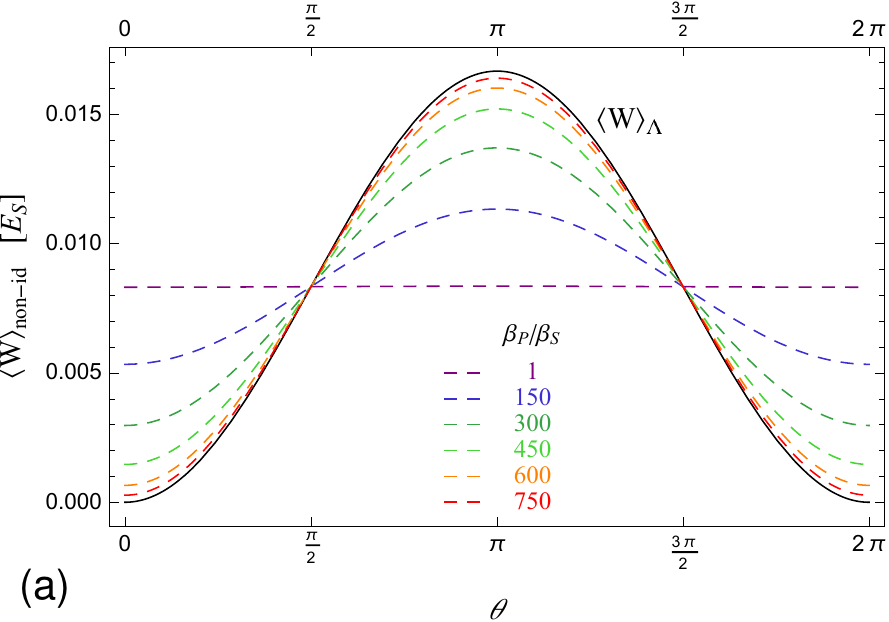}
        \hspace*{2mm}(b)
        \includegraphics[width=0.45\columnwidth,trim={0cm 0mm 0cm 0mm}]{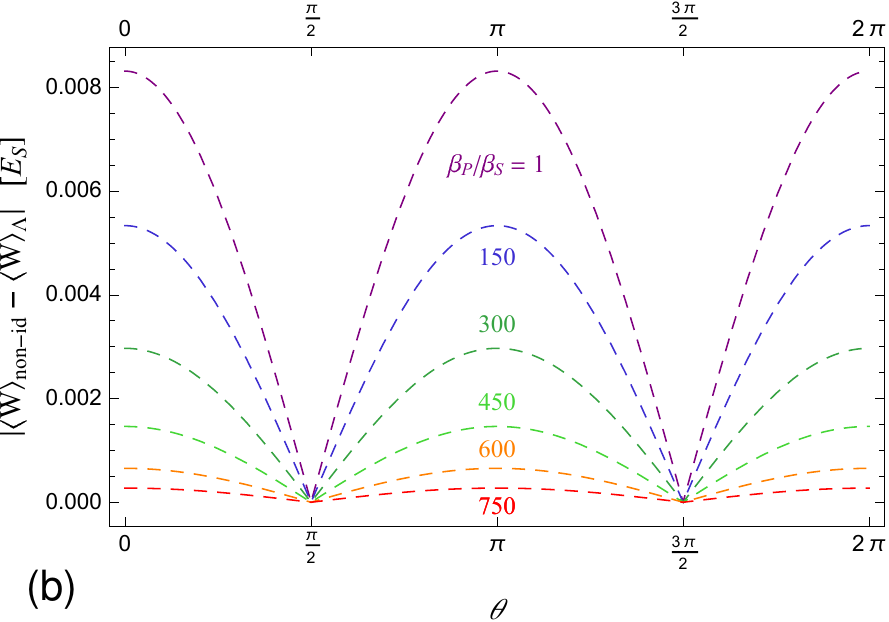}
        \caption{The non-ideal work estimate $\mean{W}_{\mathrm{non-id}}$ and its deviation from the ideal work estimate $\mean{W}_{\Lambda}$ are shown in  (a) and (b), respectively, as functions of $\theta= \Omega\Field t \in [0, 2\pi]$, for selected initial temperatures of the pointer, represented by different ratios of the system and pointer temperature, i.e., $\beta\Poi/\beta\Sys=1$, and from $\beta\Poi/\beta\Sys=150$ to $750$ in steps of $150$. For $\theta = \pi/2$ and $\theta= 3\pi/2$ the non-ideal work estimate coincides with the work performed by the process $U(\theta)$, independently of the temperature of the pointer. For the purpose of the this illustration, the initial temperature of the system has been chosen to be room temperature $T\Sys = 300$~K for an energy gap in the microwave regime such that $\beta\Sys E\Sys \approx 1/30$, and $E\Poi=E\Sys/10$.}
         \label{fig:non-ideal work delta}
\end{figure}

For $\theta=\Omega\Field t = \pi/2$ and $\theta = 3\pi/2$, the non-ideal work estimate coincides with the work realised by the process $U(\theta)$ independently of the temperature of the pointer. However, for $\theta=\pi$ one notices that the \emph{deviation} $\left|\mean{W}_{\Lambda} - \mean{W}_{\mathrm{non-id}} \right|$ of the work-estimate from its ideal value $\mean{W}_{\Lambda}$ has the same order of magnitude as $\mean{W}_{\Lambda}$ for a pointer at temperature close to room temperature, $\beta\Poi/\beta\Sys=1$, where we find that $\left|\mean{W}_{\Lambda} - \mean{W}_{\mathrm{non-id}} \right|/\mean{W}_{\Lambda}=0.49875$. Indeed, for values of $\theta$ close to integer multiples of $2\pi$, the non-ideal estimate can even be arbitrarily far away from the ideal estimate (which vanishes at these points) in the sense that $\left|\mean{W}_{\Lambda} - \mean{W}_{\mathrm{non-id}} \right|/\mean{W}_{\Lambda}$ diverges as $\theta\rightarrow2\pi$. In other words, if the temperature of the pointer is not taken into account in the work-estimation, the imprecision of the estimate can be as bigger than the work performed or extracted by the process $U(\theta)$.  In these cases $U(\theta)$ satisfies the condition in Eq.~\eqref{eq: condition w non-id equal w ideal}. At the same time, we note that, since we consider the special case of a qubit system, the minimal energy UMC measurements we consider are also minimally invasive, meaning that Jarzynski's relation is satisfied for all $\theta$ and $\beta\Poi/\beta\Sys$.

\begin{figure}[t]
        \includegraphics[scale=0.30]{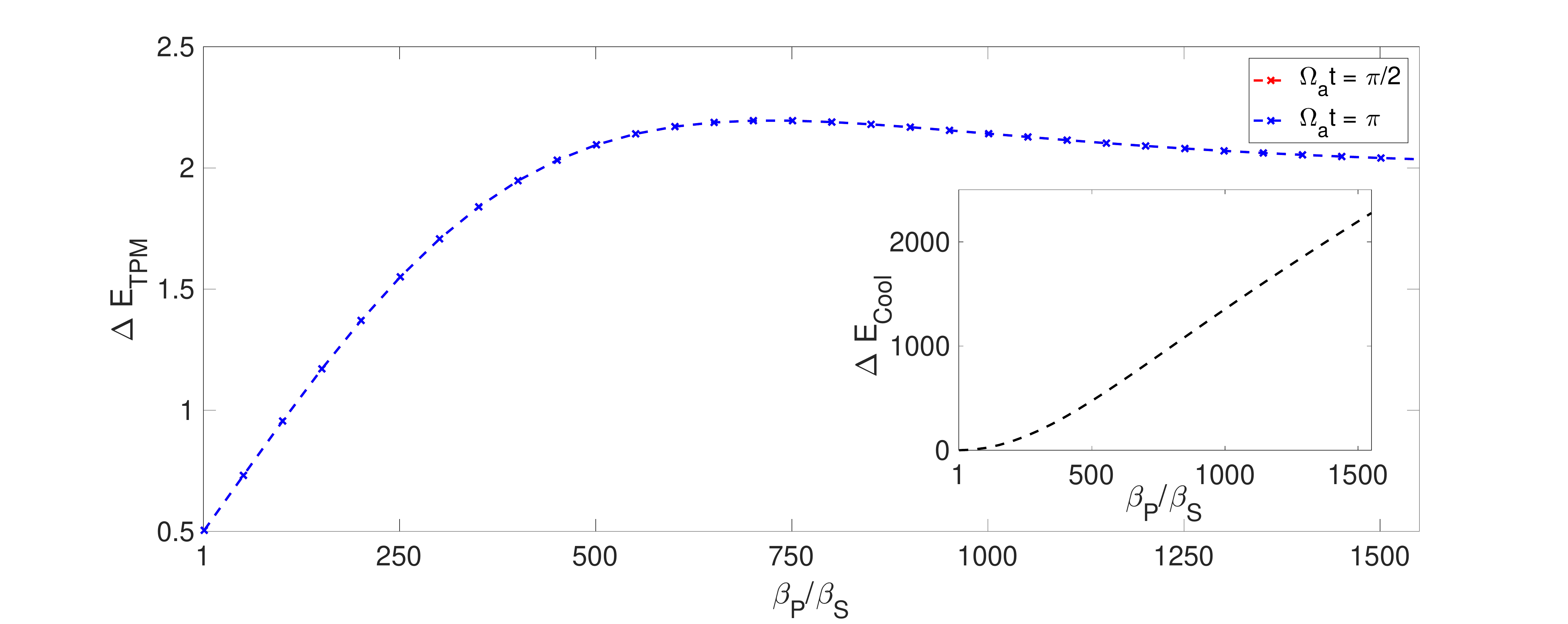}
        \caption{The energy cost $\Delta E_{\mathrm{TPM}}$ to perform the (non-ideal) TPM measurements is shown in units of $E\Sys$ as a function of the ratio between the temperatures of the system and pointer $\beta\Poi/\beta\Sys = T\Sys/T\Poi$, where the system is assumed to be at room temperature $T\Sys = 300K$ initially, for two exemplary durations $\theta=\Omega\Field t = \pi/2$ and $\theta=\Omega\Field t = \pi$ of the driving protocol. We see that there is no discernible dependence on $\theta=\Omega\Field t$. More importantly, the inset plot shows the energy cost $\Delta E_{\mathrm{Cool}}$ (in units of $E\Sys$) for cooling the pointer from $1/\beta\Sys$ to $1/\beta\Poi\leq1/\beta\Sys$ as a function of the ratio $\beta\Poi/\beta\Sys$.}
         \label{fig:TPM Energy cost}
\end{figure}

Let us now consider the energy spent to perform the measurements in the TPM process in the first place, which is given by $\Delta E_{\mathrm{TPM}} := \Delta E_{\mathrm{meas}}\suptiny{0}{0}{(0)} + \Delta E_{\mathrm{meas}}\suptiny{0}{0}{(\mathrm{f})}$. The contributions from the two respective measurements are
\begin{align}
    \Delta E_{\mathrm{meas}}\suptiny{0}{0}{(0)} &=\, \tr\bigl[(H\suptiny{0}{0}{(0)}\Sys + \sum_{i=1,2,3}H\PoiN{i})(\tilde{\rho}\suptiny{0}{0}{(0)}\SP \,-\, \rho\suptiny{0}{0}{(0)} \Sys \otimes \tau\Poi)\bigr], \\
    \Delta E_{\mathrm{meas}}\suptiny{0}{0}{(\mathrm{f})} &=\, \tr\bigl[(H\Sys\suptiny{0}{0}{(\mathrm{f})} + \sum_{i=1,2,3}H\PoiN{i})(\tilde{\rho}\suptiny{0}{0}{(\mathrm{f})}\SP \,-\, \rho\Sys\suptiny{0}{0}{(\mathrm{f})} \otimes \tau\Poi)\bigr],
\end{align}
where $H\PoiN{i}$ is the Hamiltonian for qubit $i$ (for $i=1,2,3$) and we have assumed that the pointer is prepared in the same initial state for both measurements. The states $\tilde{\rho}\suptiny{0}{0}{(0)}\SP$ and $\tilde{\rho}\suptiny{0}{0}{(\mathrm{f})}\SP$ are the joint system-pointer post-measurement states after the two respective measurements, as in Eq.~\eqref{eq:min energy unbiased max corr state}, for the initial state $\rho\Sys\suptiny{0}{0}{(0)}$ at $t=0$, and the final state at $t=t_\mathrm{f}$ is  $\rho\Sys\suptiny{0}{0}{(\mathrm{f})} = \rho\Sys(\theta)$.

In Fig.~\ref{fig:TPM Energy cost} we plot $\Delta E_{\mathrm{TPM}}$ as a function of the initial temperature of the pointer for two distinct durations of the driving protocol. We notice that there is no discernible dependence of $\Delta E_{\mathrm{TPM}}$ on the duration of the protocol. Furthermore, we calculate the energy cost of cooling the pointer within the single-qubit refrigerator paradigm~\cite{ClivazSilvaHaackBohrBraskBrunnerHuber2019a,ClivazSilvaHaackBohrBraskBrunnerHuber2019b}.
For a refrigerator with an energy gap $E_{\mathrm{F}}$, the energy needed to cool the pointer from a temperature $1/\beta\Sys$ to the lower temperature $1/\beta\Poi$ is at least
\begin{equation}
   \Delta E_{\mathrm{Cool}} = N(E_{\mathrm{F}} - 1)\ch{\frac{1}{e^{-\beta\Sys E_{\mathrm{F}}}+1 } -   \frac{1}{e^{-\beta\Sys E\Poi}+1 }},
\end{equation}
where $N$ is the number of qubits to be cooled, and $E_{\mathrm{F}}=E\Poi \beta\Poi/\beta\Sys$. In the example we consider $N=3$ for each measurement and $E\Poi =E\Sys/10$. From Fig.~\ref{fig:non-ideal work delta}, we see that for $\theta=\pi$ and  $\beta\Poi/\beta\Sys \geq 750$, the deviation of the non-ideal work-estimate from the ideal one is nearly half as big as $\mean{W}_{\Lambda}$, but $\left|\mean{W}_{\Lambda} - \mean{W}_{\mathrm{non-id}} \right|< 0.01 E\Sys$. However, as illustrated in Fig.~\ref{fig:TPM Energy cost}, the energy cost for cooling the pointer such that $\beta\Poi/\beta\Sys = 750$ is more than two times $E\Sys$. In other words, the energy cost of cooling the pointer to the required temperature can outweigh $\mean{W}_{\Lambda}$ by orders of magnitude. That is, the cost of estimating the work done by a physical process can by far exceed the amount of work that is done or consumed during the process.

\end{document}